\documentclass{amsart} 

\usepackage{graphicx}
\usepackage{algorithm}
\usepackage{algorithmic}
\usepackage{fullpage}
\usepackage{verbatim}
\usepackage{subfigure}

\numberwithin{equation}{section}

\newtheorem{theorem}{Theorem}[section]

\newtheorem{lemma}[theorem]{Lemma}

\theoremstyle{definition}

\newtheorem{remark}[theorem]{Remark}

\newcommand{\rset}{\mathbb{R}}      %
\newcommand{\tset}{\mathbb{T}}      %
\newcommand{\tpsi}{\psi} %
\newcommand{\hpsi}{\phi}

\def\REF#1{{Ref.~\refcite{{#1}}}}
\def\REFS#1{{Refs.~\refcite{{#1}}}}
\def\REF#1{{~\cite{{#1}}}}
\def\REFS#1{{~\cite{{#1}}}}

\def\LPROD#1#2{\langle {#1},{#2}\rangle}
     
\def\EPROD#1#2{{#1}\cdot {#2}}

\def\CHECKP{{\check P}}                       
\def\CHECKX{{\check X}}                       
\def\CHECKQ{{\check Q}}                      
\def\HATP{{\hat P}}                          
\def\HATX{{\hat X}}

\def\MD{{\mathrm{MD}}}
\def\BO{{\mathrm{BO}}}
\def\SCH{{\mathrm{S}}}

\def\HOPER{\mathcal{H}}
\def\VOPER{\mathcal{V}}
\def\SOPER{\mathcal{S}}
\def\ROPER{\mathcal{R}}
\def\BOPER{\mathcal{B}}
\def\AOPER{\mathcal{A}}
\def\QOPER{\mathcal{Q}}
\def\BTOPER{{\tilde{\mathcal{B}}}}

\def\PSHARP#1{\Pi{#1}}
\def\PSHARPT#1#2{\Pi({#1}){#2}}

\def\PHASE{\Theta(\CHECKX,\HATX,\CHECKP)}
\def\PHASED{\Theta'(\CHECKX,\HATX,\CHECKP)}
\def\DTHETA{\delta\LFT{\theta}(\CHECKP)}

\def\CHECKR{\check R}
\def\CHECKS{\check S}

\def\SGNF{\mathrm{Sgn}}

\def\SFUN{F}

\def\WIDEHATPHI{\Phi^{(h)}}

\def\ID{\mathrm{Id}}

\def\PSIA{\phi} 
\def\EXPWKB{\,\EXP{\Iunit M^{1/2}\theta(X)}}

\def\JAC#1{{J({#1})}}

\def\AIRY{\mathrm{A}}          
\def\HATAIRY{\hat{\mathrm{A}}} 
\def\LFT#1{{#1}^*}             
\def\DPLFT#1{{\nabla_{\CHECKP}\LFT{#1}}}

\def\NEIGH{\mathcal{U}}        

\def\TLDG{{\tilde g}}          

\def\Oo{\mathcal{O}}

\def\C{\mathbb{C}}

\def\R{\mathbb{R}}

\def\dd{\,d}

\def\dX{\dd X}

\def\EXP#1{e^{#1}}
\def\DET{\det}
\def\TRACE{\mathrm{Tr}\,}

\def\GRAD{\nabla}
\def\GRADX{\nabla_X}
\def\GRADP{\nabla_P}
\def\LAP{\Delta}
\def\DIV{{\mathrm{div}}}

\def\DT#1{{\dot{#1}}}
\def\DDT#1{{\ddot{#1}}}

\def\SGN{\mathrm{sgn}}

\def\COMMA{\,,}             
\def\PERIOD{\,.}            
\def\SEP{{\,|\,}}           

\def\VIZ#1{(\ref{#1})}      

\def\Iunit{i}
\def\REAL{\mathrm{Re}\,}
\def\IMAG{\mathrm{Im}\,}

\def\BIGO{\Oo}

\def\LTWO{L^2}                   

\begin{document}
i

\title[Accuracy of molecular dynamics]{How accurate is molecular dynamics?}

\thanks{
The research of P.P.  and A.S. was partially supported by the National Science Foundation under the grant
NSF-DMS-0813893 and  Swedish Research Council grant 621-2010-5647, respectively.
}
\author{Christian Bayer}
\address{
Department of Mathematics
University of Vienna
Nordbergstraße 15
1090 Wien, Austria}
\email{christian.bayer@univie.ac.at}
\author{H{\aa}kon Hoel}
\address{Department of Numerical Analysis,
  Kungl. Tekniska H\"ogskolan,
  100 44 Stockholm,
  Sweden}
\email{hhoel@kth.se}

\author{Petr Plech\'{a}\v{c}}
\address{Department of Mathematical Sciences, 
   University of Delaware, 
   Newark, DE 19716, 
   USA}
\email{plechac@math.udel.edu}
\author{Anders Szepessy}
\address{Department of Mathematics,
  Kungl. Tekniska H\"ogskolan,
  100 44 Stockholm,
  Sweden}
\email{szepessy@kth.se}
\author{Raul Tempone}
\address{Division of Mathematics,
King Abdullah University of Science and Technology, Thuwal 23955-6900, 
Kingdom of Saudi Arabia
}
\email{raul.tempone@kaust.edu.sa}

\maketitle
\tableofcontents

\section{Motivation for error estimates of molecular dynamics}

Molecular dynamics is a computational method
to study molecular systems in materials science, chemistry and molecular biology.
The simulations are used, for example, in designing and understanding
new materials or for
determining biochemical reactions in drug design, \REF{frenkel}.
The wide popularity of molecular dynamics simulations
relies on the fact that in many cases it agrees very well with experiments.
Indeed  when we have experimental data it is easy to verify correctness of the method
by comparing with experiments at certain parameter regimes.
However, if we want the simulation to predict something that has no comparing experiment,
we need a mathematical estimate of the accuracy of the computation.
In the case of molecular systems with few particles such studies are made
by directly solving the Schr\"odinger equation. 
A fundamental and still open question in classical molecular dynamics simulations 
is how to verify the accuracy computationally,
i.e., when the solution of the Schr\"odinger equation
is not a computational alternative.

The aim of this paper is to derive qualitative error estimates for molecular dynamics
and present new mathematical methods
which could be used also for a more demanding quantitive accuracy estimation, 
without solving the Schr\"odinger solution.
Having molecular dynamics error estimates opens, for instance, the possibility 
of systematically evaluating which density functionals or
empirical force fields are good approximations and under what conditions the approximation
properties hold.
Computations with such error estimates could also give improved understanding 
when quantum effects are important and when they are not, in particular 
in cases when the Schr\"odinger equation is too computational complex to solve.

\smallskip
{\it The first step to check the accuracy} of a molecular dynamics simulation is
to know what to compare with.
Here we compare with the value of any {\it observable} $g(X)$, of nuclei positions $X$,
for the {\it time-independent Schr\"odinger} eigenvalue equation
$\HOPER \Phi=E\Phi$, so that the approximation error we study is
\begin{equation}\label{g_def}
  \int_{\rset^{3(N+n)}} g(X)\Phi(x,X)^*\Phi(x,X) \, dx\,dX - 
  \lim_{T\to\infty}\frac{1}{T}\int_0^T g(X_t) dt\COMMA
\end{equation}
for a molecular dynamics path $X_t$, with total energy equal to the Schr\"odinger eigenvalue $E$.
The observable can be, for instance, the local potential energy, used in \REF{erik} to 
determine phase-field partial differential equations from molecular dynamics simulations, 
see Figure~\ref{two_phase_particles}.
The time-independent Schr\"odinger equation has a remarkable property of accurately predicting 
experiments in combination with no unknown data, thereby forming the foundation of computational chemistry.
However, the drawback is the high dimensional solution space for nuclei-electron systems with several particles, 
restricting numerical solution to small molecules.
In this paper we study the {\it time-independent} setting of the Schr\"odinger equation as the reference.
The proposed approach has the advantage of avoiding the  
difficulty of finding the initial data for the time-dependent Schr\"odinger equation.
\begin{figure}[htbp]
  \includegraphics[height=5cm]{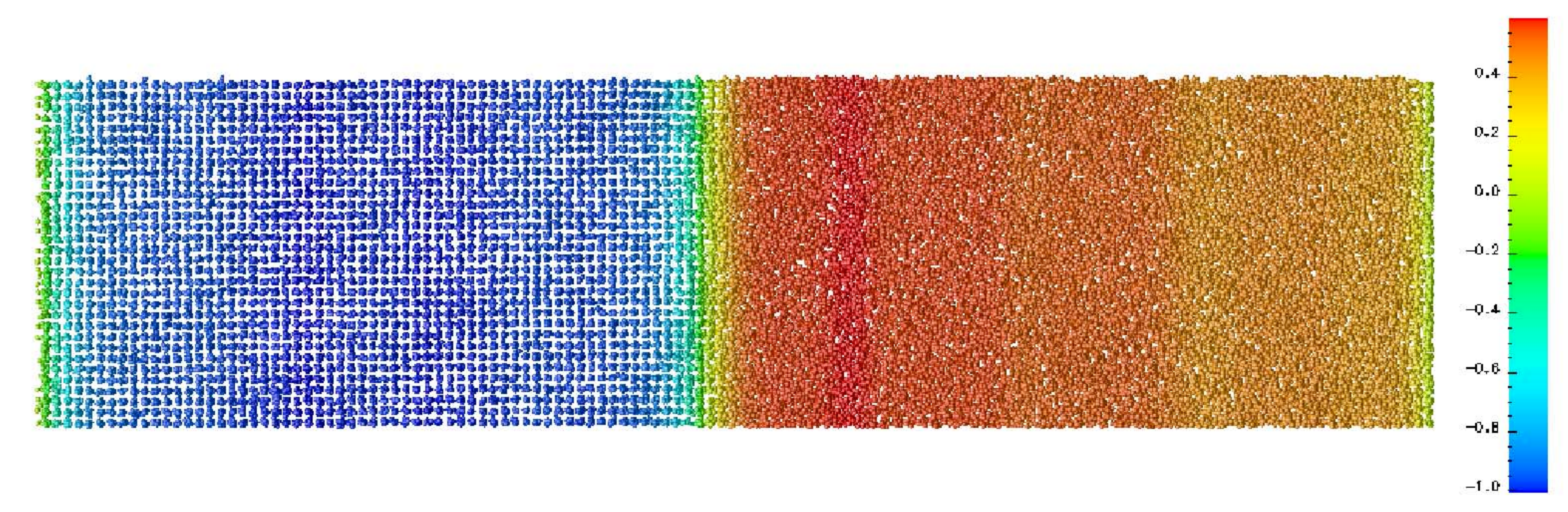}
  \caption{A Lennard-Jones molecular dynamics simulation of a phase transition with periodic boundary conditions, 
  from \REF{erik}.  The left part is solid and the right is liquid. The color measures the local potential energy.}
  \label{two_phase_particles}
\end{figure}

\smallskip
{\it The second step to check the accuracy} is to derive error estimates.  
We have three types of error: time discretization error, sampling error and modeling error.
The time discretization error comes from approximating the differential equation for molecular dynamics with a
numerical method, based on replacing time derivatives with difference quotients for a positive step size $\Delta t$.
The sampling error is due to truncating the infinite $T$ and using a finite value of $T$ in the integral 
in \eqref{g_def}.
The modeling error (also called coarse-graining error) originates from eliminating the electrons in
the Schr\"odinger nuclei-electron system and replacing the nuclei dynamics with their classical paths;
this approximation error was first analyzed by Born and Oppenheimer in their seminal paper \REF{BO}.

The time discretization and truncation error components 
are in some sense simple to handle by comparing simulations
with different choice of $\Delta t$ and $T$, although it can, of course, be difficult to know that 
the behavior does not change with even smaller $\Delta t$ and larger $T$.
The modeling error is more difficult to check since a direct approach would require to solve 
the Schr\"odinger equation.
Currently the Schr\"odinger partial differential equation can only be solved with few particles, therefore
it is not an option to solve the Schr\"odinger equation in general.
The reason to use molecular dynamics is precisely in avoiding solution of the Schr\"odinger equation.
Consequently the modeling error requires mathematical error analysis.   
In the literature there seems to be  no error analysis that
is precise, simple and constructive enough so that a molecular dynamics simulation can use it to
asses the modeling error. 
Our alternative error analysis presented  here is developed with the aim to 
allow the construction of
algorithms that estimate the modeling error in molecular 
dynamics computations.
Our analysis differs from previous ones by using 
\begin{itemize}
\item[-] the time-independent Schr\"odinger equation as the reference model to 
                 compare molecular dynamics with,
\item[-] an amplitude function in a WKB-Ansatz that depends only on position coordinates 
                 $(x,X)$ (and not on momentum coordinates $(p,P)$) for caustic states,
\item[-] actual solutions of the Schr\"odinger equation (and not only asymptotic solutions),
\item[-]  the theory of Hamilton-Jacobi partial differential equations to derive estimates 
                  for the corresponding Hamiltonian systems, i.e., the molecular dynamics systems.
\end{itemize}

Understanding both the exact Schr\"odinger model and the molecular dynamics model
through Hamiltonian systems allows us to
obtain bounds for the difference of the solutions by well-established
comparison results for the solutions of Hamilton-Jacobi equations, by
regarding the Schr\"odinger Hamiltonian and the molecular dynamics Hamiltonians 
as perturbations of each others. 
The Hamilton-Jacobi theory applied to
Hamiltonian systems is inspired by the error analysis of
symplectic methods for optimal control problems for partial differential equations, \REF{ms}.
The result is that the modeling error can be estimated based on the difference of the
Hamiltonians, for the molecular dynamics system and the Schr\"odinger system, 
along the same solution path, see Theorem~\ref{bo_thm} and Section~\ref{pert_ham}.

\section{The Schr\"odinger  and molecular dynamics models}

In deriving the approximation of the solutions to the full Schr\"odinger equation
the heavy particles are often treated within classical mechanics, i.e., by defining
the evolution of their positions and momenta by equations of motions of classical mechanics.
Therefore we denote $X_t:[0,\infty) \to \R^{3N}$ and $P_t:[0,\infty) \to \R^{3N}$ time-dependent
functions of positions and momenta with time derivatives denoted by
$$
  \DT{X_t} = \frac{d X_t}{dt}\COMMA\;\; \DDT{X_t} = \frac{d^2 X_t}{dt^2}\PERIOD
$$
We denote the Euclidean scalar product on $\R^{3N}$ by
$$
  \EPROD{X}{Y} = \sum_{i=1}^{3N} X^i Y^i\PERIOD 
$$
Furthermore, we use the notation $\GRADX \psi(x,X) = (\GRAD_{X^1}\psi(x,X),\dots,\GRAD_{X^N}\psi(x,X))$,
and as customary $\GRAD_{X^i}\psi = (\partial_{X_1^i}\psi, \partial_{X_2^i}\psi, \partial_{X_3^i}\psi)$.

On the other hand,  the light particles are treated within the quantum mechanical description
and the following complex valued bilinear map $\LPROD{\cdot\,}{\cdot}:\LTWO(\R^{3n}\times\R^{3N})\times\LTWO(\R^{3n}\times\R^{3N}) \to\LTWO(\R^{3N})$
will be used in the subsequent calculations
\begin{equation}\label{electron_scalar}
  \LPROD{\phi}{\psi} = \int_{\R^{3n}} \phi(x,X)^* \psi(x,X) \, dx\PERIOD
\end{equation}
The notation $\psi(x,X)=\BIGO(M^{-\alpha})$ is also used for complex valued functions, meaning that
$|\psi(x,X)|=\BIGO(M^{-\alpha})$ holds uniformly in $x$ and $X$.

The {\it time-independent  Schr\"odinger} equation 
\begin{equation}\label{schrodinger_stat}
  \HOPER(x,X) \Phi(x,X) = E\Phi(x,X)
\end{equation}
models many-body (nuclei-electron) quantum systems and  is
obtained from minimization of the energy in the solution space of wave functions, see 
\REFS{schrodinger,schiff,berezin,tanner,lebris_hand}.
It is an eigenvalue problem for the energy  $E\in\rset$ of the system in the solution space, 
described by wave functions, $\Phi:\rset^{3n}\times\rset^{3N}\to \C$, depending on electron coordinates 
$x=(x^1,\ldots, x^n)\in\rset^{3n}$,
nuclei coordinates $X=(X^1,\ldots, X^N)\in\rset^{3N}$,  and the Hamiltonian operator $\HOPER(x,X)$ 
\begin{equation}\label{V-definition}
   \HOPER(x,X)= \VOPER(x,X) - \frac{1}{2} M^{-1}\sum_{n=1}^N\Delta_{X^n}\PERIOD
\end{equation}
We assume that a quantum state of the system
is fully described by the wave function $\Phi:\R^{3n}\times\R^{3N}\to\C$ which
is an element of the Hilbert space %
of wave functions
with the standard complex valued scalar product
$$
  \langle\!\langle\Phi,\Psi\rangle\!\rangle = 
   \int_{\R^{3n}\times\R^{3N}} \Phi(x,X)^* \Psi(x,X)\, dx\,dX\COMMA
$$
and the operator $\HOPER$ is 
self-adjoint in this Hilbert space. The Hilbert space
is then a subset of $\LTWO(\rset^{3n}\times\rset^{3N})$ with symmetry conditions based on 
the Pauli exclusion principle for electrons, see \REFS{lebris_hand,lieb_matter}.

In computational chemistry the operator
$\VOPER$, the electron Hamiltonian, is independent of $M$ and it is
precisely determined by the sum of the
kinetic energy of electrons and the Coulomb interaction between nuclei and electrons.
We assume that the electron operator $\VOPER(\cdot, X)$ is self-adjoint
in the subspace with the inner product $\LPROD{\cdot}{\cdot}$ of functions in \eqref{electron_scalar} with fixed $X$ coordinate
and acts as a multiplication on functions that depend only on $X$. %
An essential feature of  the partial differential equation \VIZ{schrodinger_stat}
is the high computational complexity of finding  the solution in an antisymmetric/symmetric
subset of the Sobolev space $H^1(\R^{3n}\times\R^{3N})$.
The mass of the nuclei, which are much greater than one (electron mass),  
are the diagonal elements in the diagonal matrix $M$.

In contrast to the Schr\"odinger equation, 
a {\it molecular dynamics} model of $N$ nuclei $X:[0,T]\to\rset^{3N}$, 
with a given potential $V_p:\rset^{3N}\to \rset$,
can be computationally studied for large $N$  by solving the ordinary differential equations 
\begin{equation}\label{md_eq}
  \ddot X_t=- \GRADX V_p(X_t)\COMMA
\end{equation}
in the slow time scale, where the nuclei move $\mathcal O(1)$ in unit time.
This computational and conceptual simplification motivates the study to determine 
the potential and its implied accuracy compared with the
the Schr\"odinger equation,  as started  already
in the 1920's with the Born-Oppenheimer approximation \cite{BO}. 
The purpose of our work is to contribute to the current understanding
of such derivations by showing convergence rates under new assumptions.
The precise aim in this paper is to estimate the error 
\begin{equation}\label{approximation}
\frac{\int_{\rset^{3N+3n}} g(X) \Phi(x,X)^*\Phi(x,X) \,dx\,dX}{\int_{\rset^{3N+3n}} \Phi(x,X)^*\Phi(x,X) \,dx\,dX}
 - \lim_{T\to\infty}\frac{1}{T}\int_0^T g(X_t)\, dt
\end{equation}
for a position dependent observable $g(X)$ of the time-indepedent 
Schr\"odinger equation \eqref{schrodinger_stat} approximated 
by the corresponding molecular dynamics observable 
$\lim_{T\to\infty}T^{-1}\int_0^T g(X_t) \, dt$,
which is computationally cheaper to evaluate with several nuclei.
The Schr\"odinger eigenvalue problem may typically have multiple eigenvalues and the aim is to
find an eigenfunction $\Phi$ and a molecular dynamics system that can be compared.
There may be eigenfunctions that we cannot approximate, but with some assumptions on
the spectrum of $\VOPER(\cdot,X)$ the molecular dynamics  in fact approximates the observable corresponding to
one eigenfunction.

The main step to relate the Schr\"odinger wave function 
and the molecular dynamics solution is the so-called zero-order Born-Oppenheimer approximation, 
where $X_t$ solves  the classical
{\it ab initio} molecular dynamics (\ref{md_eq}) with
the potential $V_p:\rset^{3N}\to \rset$ determined as an eigenvalue of the electron Hamiltonian
$\VOPER(\cdot,X)$ for a given nuclei position $X$. That is  $V_p(X)=\lambda_{0}(X)$
and 
$$
  \VOPER(\cdot,X)\Psi_{\BO}(\cdot, X)=\lambda_{0}(X)\Psi_{\BO}(\cdot, X)\COMMA
$$
for an  electron eigenfunction $\Psi_{\BO}(\cdot, X)\in \LTWO(\rset^{3n})$, for instance,  the ground state.
The Born-Oppenheimer expansion \cite{BO} is an approximation of the solution to the time-independent
Schr\"odinger equation
which is shown in \REFS{hagedorn_egen,martinez} 
to solve the time-independent Schr\"odinger equation approximately. 
This expansion, analyzed by the methods of multiple scales, pseudo-differential operators
and spectral analysis in \REFS{hagedorn_egen,martinez, fefferman},
can be used to study the approximation error
\eqref{approximation}. %
However, in the literature, e.g.,\REF{martinez2}, it is easier to find precise statements      
on the error for the setting of the time-dependent Schr\"odinger equation,
since the stability issue is more subtle in the eigenvalue setting.  

Instead of an asymptotic expansion  we use a different method based on 
a Hamiltonian dynamics formulation of the {\it time-independent}
Schr\"odinger eigenfunction  and  the stability of the corresponding perturbed
Hamilton-Jacobi equations viewed as a hitting problem.
This approach makes it possible to reduce the error propagation on the infinite time interval 
to finite time excursions from a certain co-dimension one hitting set.
A motivation for our method  is 
that it forms a sub-step in trying to  estimate the approximation error
using only information available in molecular dynamics simulations.

The related problem of approximating 
observables to the time-dependent Schr\"odinger equation by the Born-Oppenheimer expansions
is well studied, theoretically in \REFS{robert, spohn_egorov} and computationally in 
\REF{lasser}
using the Egorov theorem.
The Egorov theorem shows that finite time
observables of the time-dependent Schr\"odinger equation are approximated with $\BIGO(M^{-1})$
accuracy by the zero-order Born-Oppenheimer dynamics with an electron eigenvalue gap.
In the special case of a position observable and no electrons (i.e., $\VOPER=V(X)$ in \eqref{V-definition}), 
the Egorov theorem states that
\begin{equation}\label{egorov}
  \left|\int_{\rset^{3N}} g(X) \Phi(X,t)^*\Phi(X,t) \, dX -
  \int_{\rset^{3N}} g(X_t) \Phi(X_0,0)^* \Phi(X_0,0) \, dX_0\right| \le C_t M^{-1}\COMMA
\end{equation}
where $\Phi(X,t)$  is a solution to the time-dependent
Schr\"odinger equation 
$$
  \Iunit \partial_t \Phi(\cdot,t)=\HOPER\Phi(\cdot,t)
$$ 
with the Hamiltonian \eqref{V-definition}
and the path $X_t$ is the nuclei coordinates for the dynamics with the Hamiltonian 
$\tfrac{1}{2}|\dot X|^2 + V(X)$.
If the initial wave function $\Phi(X,0)$ is the eigenfunction in \eqref{schrodinger_stat}
the first term in \eqref{egorov} reduces to the first term in \eqref{approximation}
and the second term can also become the same in an ergodic limit. However,
since we do not know that the parameter $C_t$ (bounding an integral over $(0,t)$) is bounded for all time
we cannot directly conclude an estimate for \eqref{approximation} from \eqref{egorov}.

In our perspective studying the time-independent instead of the time-dependent Schr\"odinger equation has
the important differences that
\begin{itemize}
\item[-] the infinite time study of  the Born-Oppenheimer dynamics can be reduced
      to a finite time hitting problem, 
\item[-] the computational and theoretical problem of specifying initial data for the Schr\"odinger equation 
      is avoided, and
\item[-] computationally cheap evaluation of the position observable $g(X)$ 
      is possible using the time average $\lim_{T\to\infty}\tfrac{1}{T}\int_0^T g(X_t)\,dt$ 
      along the solution path $X_t$.
\end{itemize}

In this paper we derive
the Born-Oppenheimer approximation
from the time-independent Schr\"odinger equation (\ref{schrodinger_stat})
and we establish
convergence rates for molecular dynamics approximations
to  time-independent Schr\"odinger observables under simple assumptions including 
the so-called {\it caustic} points, where
the Jacobian determinant $\DET J(X_t) \equiv \DET(\partial X_t/\partial X_0)$ of 
the Eulerian-Lagrangian transformation of $X$-paths vanish.
As mentioned previously, the main new analytical idea
is  an interpretation of the time-independent Schr\"odinger equation \eqref{schrodinger_stat}
as a Hamiltonian system  
and the subsequent analysis of the approximations  by comparing Hamiltonians. This analysis
employs the theory of Hamilton-Jacobi partial differential equations.
The problematic infinite-time evolution of perturbations in  the dynamics
is solved by viewing it as a finite-time hitting  problem for the Hamilton-Jacobi equation, 
with a particular hitting set.
In contrast to the traditional rigorous and formal asymptotic expansions  
we analyze the transport equation as a time-dependent Schr\"odinger equation.

The main inspiration for this paper
are works \REFS{Mott,briggs,briggs2} and the semi-classical WKB analysis in 
\REF{maslov}:
the works \REFS{Mott, briggs,briggs2} 
derive the time-dependent Schr\"odinger dynamics of an $x$-system, 
$
  \Iunit \DT{\Psi}=\HOPER_1 \Psi,
$
from the time-independent Schr\"odinger equation (with the Hamiltonian $\HOPER_1(x)+ \epsilon\, \HOPER(x,X)$)
by a classical limit for the environment variable $X$, as the coupling parameter $\epsilon$ vanishes
and the mass $M$ tends to infinity;  in particular \REFS{Mott, briggs,briggs2} show that the time 
derivative enters through the coupling of $\Psi$ with the classical velocity. 
Here we refine
the use of characteristics to study  classical  {\it ab initio} molecular dynamics where the coupling
does not vanish, and we establish error estimates for Born-Oppenheimer 
approximations of  Schr\"odinger observables. The small scale, introduced by the perturbation  
\[
  -(2M)^{-1}\sum_k\Delta_{X^k}
\] 
of the potential $\VOPER$,
is identified in a modified WKB eikonal equation
and analyzed through the corresponding transport equation as a time-dependent Schr\"odinger
equation along the eikonal characteristics.
This modified WKB formulation reduces to the standard semi-classical approximation, see \REF{maslov},
in the case of the potential function $\VOPER=V(X)\in\rset$, depending
only on nuclei coordinates, but becomes different in the case of operator-valued potentials studied here. 
The global analysis of WKB functions was initiated by Maslov in the 1960',  \REF{maslov}, 
and lead to the subject Geometry of Quantization, relating global classical paths to eigenfunctions of
the Schr\"odinger equation, see \REF{duistermaat}. The analysis presented in this paper is based
on a Hamiltonian system interpretation of the time-independent Schr\"odinger equation.
Stability of the corresponding Hamilton-Jacobi equation,
bypasses the usual separation 
of nuclei and electron wave functions in the time-dependent self-consistent field  equations,
\cite{schutte,marx,tully}.

Theorem~\ref{bo_thm}  demonstrates that observables from the zero-order Born-Oppenheimer dynamics
approximate observables for the Schr\"odinger eigenvalue problem
with the error of order $\BIGO(M^{-1+\delta})$, for any $\delta>0$,
assuming that the electron eigenvalues satisfy a spectral gap condition.
The result is based on the Hamiltonian \VIZ{V-definition} with any potential $\VOPER$ that is smooth in $X$,
e.g., a regularized version of the Coulomb potential. %
The derivation does not  assume that the nuclei are supported on small domains;
in contrast derivations based on the time-dependent self-consistent field equations 
require nuclei to be supported on small domains.
The reason that the small support is not needed here comes from the combination
of the characteristics and sampling from an equilibrium density. 
In other words,  the nuclei paths behave classically although they may not be supported 
on small domains.
Section~\ref{sec:observables} shows that caustics couple the WKB modes,
as is well-known from geometric optics, see \REFS{keller,maslov}, and generate non-orthogonal
WKB modes that are coupled in the Schr\"odinger density.
On the other hand, with a spectral gap and without caustics the Schr\"odinger
density is asymptotically decoupled into a simple sum of individual WKB densities.
Section~\ref{sec:caustics} constructs a WKB-Fourier integral Schr\"odinger solution for caustic states.
Section~\ref{md_sim} relates the approximation results to 
the accuracy of symplectic numerical methods for molecular dynamics.

A unique property of the time-independent Schr\"odinger equation we use
is the interpretation that the dynamics $X_t\in \rset^{3N}$ can return to a co-dimension one 
surface $I$ which then can  reduce %
the dynamics to a hitting time problem
with finite-time excursions from  $I$.
Another advantage of viewing the molecular dynamics as an approximation of the
eigenvalue problem is that stochastic perturbations of the electron ground state
can be interpreted as a Gibbs distribution of degenerate nuclei-electron eigenstates of the
Schr\"odinger eigenvalue problem \VIZ{schrodinger_stat}, see \REF{ASz2}.
The time-independent eigenvalue setting 
also avoids  the issue on ``wave function collapse'' to an eigenstate, 
present in the time-dependent Schr\"odinger equation.

We believe that these ideas can be further developed to better understanding of
molecular dynamics simulations. 
For example, it would be desirable to have more precise
conditions on the data (i.e. molecular dynamics initial data and potential $\VOPER$)
 instead of our implicit assumption on finite hitting time and
convergence of the Born-Oppenheimer power series approximation in Lemma \ref{born_oppen_lemma}.
\section{A time-independent Schr\"odinger WKB-solution}

\subsection{Exact Schr\"odinger dynamics}\label{sec:exact_schrod}
For the sake of simplicity we assume that all nuclei have the same mass.
If this is not the case, we can introduce new coordinates
$M_1^{1/2}\tilde X^k=M^{1/2}_kX^k$, which transform the Hamiltonian to the form we want
${\VOPER(x,M_1^{1/2}M^{-1/2}\tilde X)}- (2M_1)^{-1}\sum_{k=1}^N\Delta_{\tilde X^k}$.
The singular perturbation $-(2M)^{-1}\sum_k\Delta_{X^k}$ of the potential $\VOPER$ introduces
an additional  small scale $M^{-1/2}$ of high frequency oscillations,
as shown by a WKB-expansion, see   \REFS{Rayleigh,jeffreys,helfer,sjostrand}.
We shall construct  solutions to (\ref{schrodinger_stat}) in  such a WKB-form 
\begin{equation}\label{wkb_form}
  \Phi(x,X)=\PSIA(x,X)\EXP{\Iunit M^{1/2}\theta(X)}\COMMA
\end{equation}
where the amplitude function $\PSIA: \rset^{3n}\times \rset^{3N} \to \C$ 
is complex valued, the phase $\theta: \rset^{3N}\to \rset$ 
is real valued, and the factor $M^{1/2}$
is introduced in order to have  well-defined limits of $\PSIA$ and $\theta$  
as  $M\to\infty$.
Note that it is trivially always possible to find funtions $\PSIA$ and
$\theta$ satisfying  \VIZ{wkb_form}, even in the sense of a true equality.  Of
course, the ansatz only makes sense if $\PSIA$ and $\theta$ do not have strong
oscillations for large $M$.
The standard WKB-construction,  \cite{maslov,helfer}, is
based on a series expansion in powers of $M^{1/2}$ which solves
the Schr\"odinger equation with arbitrary high accuracy. 
Instead of an asymptotic solution,
we introduce  an actual solution based on
a time-dependent Schr\"odinger transport equation. This transport equation
reduces to the formulation
in \REF{maslov} for the case of a potential function
$\VOPER=V(X)\in \rset$, depending only on nuclei coordinates $X\in\rset^{3N}$,
and modifies it for the case of a self-adjoint potential operator $\VOPER(\cdot,X)$ 
on the electron space $\LTWO(\rset^{3n})$ which is the primary focus of our work here.
In Sections \ref{sec:observables} and \ref{sec:caustics} we use 
a linear combination of WKB-eigensolutions, but first we study the simplest case of a single WKB-eigensolution
as motivated by the following subsection.

\subsubsection{Molecular dynamics from a piecewise constant electron operator on a simplex mesh} 
The purpose of this section is to convey a first formal understanding of the relation between ab initio molecular dynamics
$\ddot X_t= -\GRADX\lambda_0(X_t)$ and the Schr\"odinger eigenvalue problem \eqref{schrodinger_stat} and motivate the WKB ansatz \eqref{wkb_form}.
In subsequent sections we will describe precise analysis of error estimates for the WKB-method. 
The idea behind this first study
is to approximate the electron operator $\VOPER$ by a finite dimensional matrix $\VOPER^h$, which is piecewise constant on
a  simplex mesh in the variable $X$, with the mesh size $h$. Furthermore, we introduce the change of variables
\[
\Phi = \sum_{j=0}^J \varphi_j \Psi_j =: \Psi\varphi
\]
based on the piecewise constant electron eigenvalues and  eigenvectors 
$\VOPER^h\Psi_j=\lambda_j^h\Psi_j, \ \langle\Psi_j,\Psi_j\rangle=1, \ j=0,\ldots J$, normalized and ordered with respect to increasing eigenvalues. Then the Schr\"odinger  equation \eqref{schrodinger_stat} becomes
\[
-\frac{1}{2M} \Delta_X (\Psi\varphi) + \VOPER^h \Psi\varphi = E\Psi\varphi\COMMA
\]
with the notation $\Delta_X=\sum_j\Delta_{X_j}$,
so that on each simplex
\[
-\frac{1}{2M} \Delta_X \varphi_j + \lambda_j^h\varphi_j = E\varphi_j\COMMA
\]
which by separation of variables, for each $j=0,1,2,\ldots ,J$, implies 
\begin{equation}\label{phi_sum}
\varphi_j = \sum_{P^j} a({P^j})e^{iM^{1/2} P^j\cdot X}
\end{equation}
for any $P^j\in \C^{3N}$ that satisfies the eikonal equation
\[ 
\frac{1}{2}\EPROD{P^j}{P^j} + \lambda_j^h =E\COMMA
\] 
for any $a(P^j)\in \C$, if all components of $P^j$ are non zero. If $P^j_k=0$ we have
$a(P^j)=\prod_{\{k\,:\, P^j_k=0\}} (A_{k}X_k+B_k)$ for any 
$A_k\in\C, B_k\in\C$, since $e^{\pm i M^{1/2}P^j_k X_k}=1$  in this case. 
 The solution $\Phi$, to \eqref{schrodinger_stat}, and its normal derivative are continuous at the interfaces of the simplices.
On the intersection of the faces the normal derivative is not defined but this set is of measure zero and thus
negligible as seen from the $H^1(\rset^{3N})$ solution concept of \eqref{schrodinger_stat}. 

We investigate a simpler, one-dimensional case,
$X\in \rset$, first. 
Then the solution $\varphi$ simplifies to
\[
\varphi_j= a_j e^{\Iunit M^{1/2} P^j\cdot X} + b_j e^{-\Iunit M^{1/2} P^j\cdot X}
\]
for $a_j, b_j, P^j \in \mathbb C$ and $(P^j)^2/2 +\lambda_j =E$\PERIOD
\ The continuity conditions
\begin{equation}\label{contin}
\begin{split}
\lim_{X\rightarrow X_0+} \Phi(X) &= \lim_{X\rightarrow X_0-} \Phi(X)\\
\lim_{X\rightarrow X_0+} \partial_X\Phi(X) &= \lim_{X\rightarrow X_0-} \partial_X\Phi(X)\\
\end{split}
\end{equation}
hold for any $X_0\in \rset$, in particular, at the interval boundary where for $X_0=0$
\begin{equation}\label{reflect-transmitt}
\begin{split}
\lim_{X\rightarrow X_0\pm} \Phi(X) &= \sum_{j} (a_{j\pm}\Psi_{j\pm} + b_{j\pm}\Psi_{j\pm})\\
\lim_{X\rightarrow X_0\pm} \partial_X\Phi(X) &
= iM^{1/2}\sum_{j} (a_{j\pm}P^j_\pm\Psi_{j\pm} - b_{j\pm}P^j_\pm\Psi_{j\pm}) \PERIOD\\
\end{split}
\end{equation}
It is clear that given  $a_-$ and $b_-$ we can determine $a_+$ and $b_+$ so that
\eqref{contin} holds. In order to prepare for the multi-dimensional case it is convenient to consider each incoming wave $a_-$ and $b_+$ separately: 
the incoming $a_-$ wave is split into  a refracted $a_+$ and reflected $b_-$ wave
\begin{equation}\label{a-b}
 \sum_{j} a_{j-}\Psi_{j-}P^j_- =
 \sum_{j} (a_{j+}\Psi_{j+}P^j_+ + b_{j-}\Psi_{j-}P^j_-)
 \end{equation}
 and similarly the incoming $b_+$ wave is split into a refracted $b_-$ wave  and a  reflected $a_+$ wave, see Figure \ref{reflect}.
 The jump conditions at the different interfaces are coupled by the oscillatory functions $e^{\pm iM^{1/2}P^j\cdot X}$.
 The global construction of $\varphi$ and $\Psi$ in one dimension
 follows by marching in the positive $X$-direction to successive intervals, creating in each interval  both
 a $e^{iM^{1/2}P^j\cdot X}\Psi_j$ and a $e^{-iM^{1/2}P^j\cdot X}\Psi_j$ wave.

%
In general each interface condition \eqref{reflect-transmitt} also couples all eigenvectors $\Psi_j$.
However, we shall see that if $M$ is large, $\VOPER$ smooth and there is a spectral gap $\lambda_1-\lambda_0>c>0$
 then, in the limit of the simplex size $h$ tending to zero, there is  an asymptotically uncoupled WKB-solution
 $\Phi(x,X)= \phi(x,X)e^{iM^{1/2}\theta(X)}$, where
 $\theta:\rset^{3N}\rightarrow \rset, \ \phi:\rset^{3n}\times \rset^{3N}\rightarrow \mathbb C$.
Under these assumptions the Born-Oppenheimer approximation in Lemma~\ref{born_oppen_lemma} 
 shows that $\phi$ is asymptotically parallel, in $L^2(dx)$, to the electron eigenfunction $\Psi_0$ as $M\rightarrow\infty$.
The gradient $\GRADX\theta(X)=P^0$ 
is obtained from the differential $\theta(X)=\theta(X_0) + \GRADX\theta(X_0)\cdot (X-X_0) + o(|X-X_0|)$.

  In the case of electron eigenvalue crossing, i.e., $\lambda_1(X)=\lambda_0(X)$ for some $X$, or so called
  avoided crossings (meaning that the eigenvalue gap $c\ll 1$ is small and dependent on $M$),
  a refraction will, in general, include all components
  $a_je^{iM^{1/2} \EPROD{P^j}{X}} \Psi_j, \ j=1,\ldots, J$ 
  and consequently the Born-Oppenheimer approximation fails.

 The construction of a solution to the Schr\"odinger equation with a piecewise constant
 potential is more involved in the multi-dimensional case for two reasons: each reflection at an interface generates, in general,
 an additional path in a new direction, so that many paths are needed.
 Furthermore, the construction of a solution to the eikonal equation is more complicated 
 since the jump condition \eqref{reflect-transmitt}
 implies that the tangential component $P^j_t$ of $P^j$ must be continuous across a simplex face
 and only the normal component $P^j_n=P^j-P^j_t$ may have a jump.
%
%
 In multi-dimensional cases it is still possible to construct
 a solution of the form \eqref{phi_sum} by following 
 the characteristic paths $\dot X_t=P^j(X_t)$ and using the jump conditions
 \eqref{reflect-transmitt}: 
 when the path $X_t$ hits a simplex face, 
 the tangential part $P^j_t$ of $P^j$ is continuous and the normal component $P^j_n$ of $P^j$ may jump.
 At a simplex face the new value of the $P^j_n$ is determined by $(\EPROD{P^j_n}{P^j_n} +\EPROD{P^j_t}{P^j_t})/2 + \lambda_j^h=E$.
 Analogously 
 to the one dimensional case we treat the pair $e^{iM^{1/2}\EPROD{(P^j_t+P^j_n)}{X}}$ and $e^{iM^{1/2}\EPROD{(P^j_t-P^j_n)}{X}}$
 together. However, each collision with $e^{iM^{1/2}\EPROD{(P^j_t+P^j_n)}{X}}$ on an interface now creates a reflected wave in another direction, 
 in particular, $e^{iM^{1/2}\EPROD{(P^j_t-P^j_n)}{X}}\Psi_j$,
 and we  get many paths to follow.  
 Therefore each mode $e^{iM^{1/2}\EPROD{P^j}{X}}$ follows its characteristic
 $X_t$, where $\dot X_t=P^j$, through the simplex to the adjacent simplicial faces, which the characteristic pass through
 when they leave the simplex, and at these outflow faces a reflected mode is created and a refracted mode
 continues into the adjacent simplices, see Figure \ref{reflect}. In this way we can formally construct a solution of the form $ \sum_{P^j} a({P^j})e^{iM^{1/2} P^j\cdot X}\Psi_j$
 to the Schr\"odinger equation \eqref{schrodinger_stat}, with possibly several different characteristic paths in each simplex.
%
%
 \begin{figure}[htbp]
\includegraphics[height=3cm]{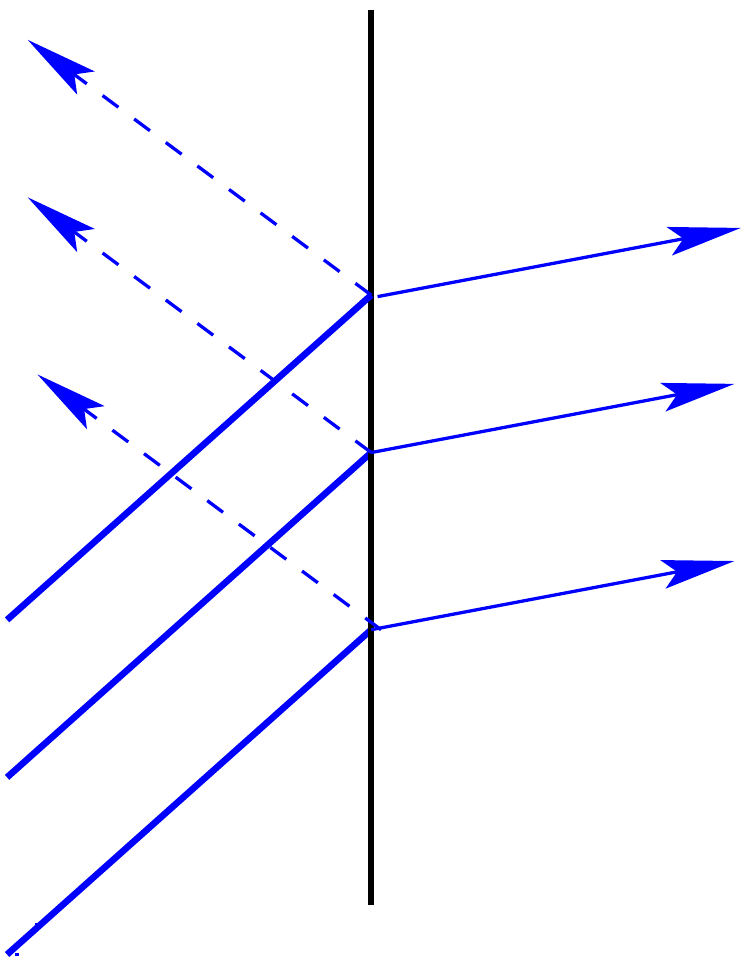}
\includegraphics[height=3cm]{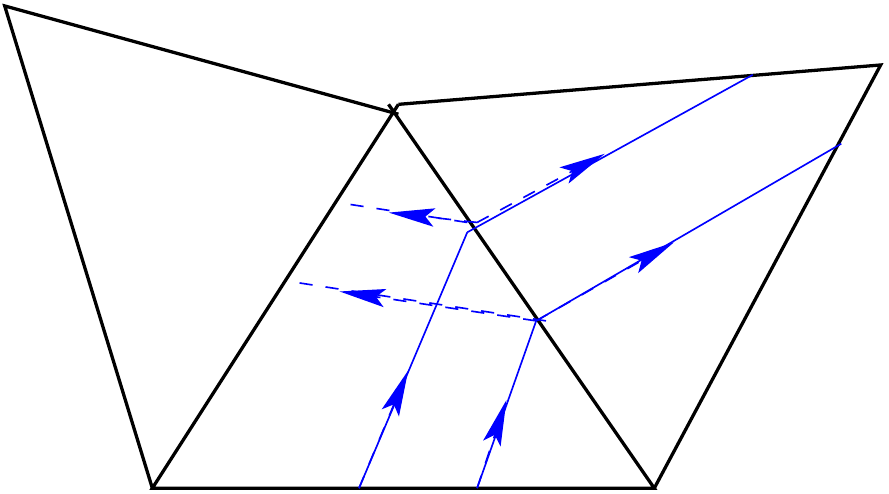}
\includegraphics[height=3cm]{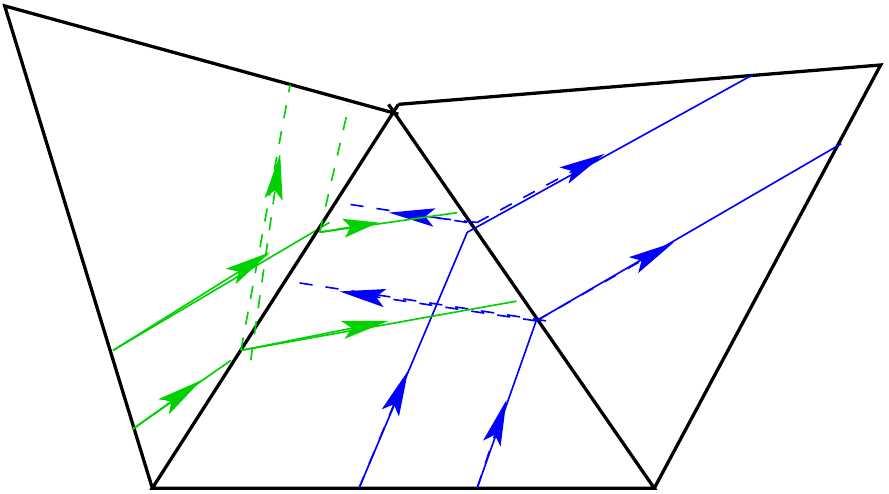}
\caption{The value of $P^j$ is constructed by following the characteristic paths $X_t$ (the blue and green curves), based on $\dot X_t=P^j$, with a reflection-refraction at each simplex face (left) following the path through simplices (middle) and each simplex may have several $P^j$ (right).}
\label{reflect}
\end{figure}

In conclusion, the piecewise constant electron operator shows that the solution to the Schr\"odinger equation
\eqref{schrodinger_stat} is composed of a linear combination of highly oscillatory function modes  $a_je^{iM^{1/2} P^j\cdot X} \Psi_j$ based on the electron eigenvectors $\Psi_j$ and eigenvalues $\lambda_j$,  
where $P^j$ satisfies the eikonal equation $P^j\cdot P^j/2 + \lambda_j(X)=E$. 
These modes can be followed be characteristics $\dot X=P^j$ from simplex to simplex. 
In this paper we show that observables based on the related WKB Schr\"odinger solutions
can be approximated by molecular dynamics time averages,
when there is a spectral gap
around $\lambda_0$. 

%
%
%
\subsubsection{A first WKB-solution}\label{first_WKB}
The WKB-solution satisfies the Schr\"odinger equation \VIZ{schrodinger_stat} provided that
\begin{equation}\label{wkb_eq}
  \begin{split}
     0 &=(\HOPER-E)\PSIA \, \EXPWKB\\
       &= \left( ( \frac{1}{2}|\GRADX\theta|^2 + \VOPER - E) \PSIA \right.
            - \frac{1}{2M} \Delta_{X}\PSIA
              - \frac{\Iunit}{M^{1/2}}   (\EPROD{\GRADX\PSIA}{\GRADX\theta}
       \left. + \frac{1}{2}\PSIA\, \Delta_{X}\theta)\right) \EXPWKB \PERIOD 
  \end{split}
\end{equation}
We shall see that only eigensolutions $\Phi$ that correspond to
dynamics without caustics correspond to such a single WKB-mode,
as for instance when the eigenvalue $E$  is inside an electron eigenvalue gap.
Solutions in the presence of caustics use a Fourier integral of such WKB-modes,
and we treat this case in detail in Section~\ref{sec:caustics}.
To understand the behavior of $\theta$, we multiply \VIZ{wkb_eq} by 
$\PSIA^* \EXP{-\Iunit M^{1/2}\theta(X)}$ and integrate over
$\rset^{3n}$. Similarly we take the complex conjugate of \VIZ{wkb_eq}, and
multiply by $\PSIA\EXPWKB $
and integrate over $\rset^{3n}$. By adding these two expressions we obtain
\begin{equation}\label{theta_ekvation}
 \begin{split}
      0 & = 2\big( \frac{1}{2}|\GRADX\theta|^2  - E \big) \, \LPROD{\PSIA}{\PSIA}
            +\underbrace{\LPROD{\PSIA}{\VOPER\PSIA} + \LPROD{\VOPER\PSIA}{\PSIA}}_{=2\LPROD{\PSIA}{\VOPER\PSIA}} 
            -\frac{1}{2M}\left( \LPROD{\PSIA}{ \LAP_{X}\PSIA}
                + \LPROD{ \LAP_{X}\PSIA}{\PSIA}\right) \\
        &{}-\frac{\Iunit}{M^{1/2}}
            \underbrace{\big(\LPROD{\PSIA}{\EPROD{\GRADX\PSIA}{\GRADX\theta}}
           -\LPROD{\EPROD{\GRADX\PSIA}{\GRADX\theta}}{\PSIA}\big)}_{%
            =2\Iunit\IMAG\LPROD{\PSIA}{\EPROD{\GRADX\PSIA}{\GRADX\theta}}}
           +\frac{\Iunit}{2M^{1/2}}
            \underbrace{\big(\LPROD{\PSIA}{\PSIA} - \LPROD{\PSIA}{\PSIA}\big)}_{=0}\LAP_{X}\theta\PERIOD
  \end{split}
\end{equation}
The purpose of the phase function $\theta$ is to generate an accurate
approximation in the limit as $M\to\infty$.
A possible and natural definition of $\theta$ would be the formal limit of 
\VIZ{theta_ekvation} as $M\to\infty$, which is
the {\it Hamilton-Jacobi equation}, also called the {\it eikonal equation}
\begin{equation}\label{theta_eq}
    \frac{1}{2}|\GRADX\theta|^2 =E - V_0 \COMMA
\end{equation}
where the function $V_0:\rset^{3N}\to \rset$ is  
\begin{equation}\label{V_0_definition}
   V_0:=\frac{\LPROD{\PSIA}{\VOPER\PSIA}}{\LPROD{\PSIA}{\PSIA}}\PERIOD
\end{equation}
The solution to the Hamilton-Jacobi eikonal equation can be constructed from the
associated Hamiltonian system
\begin{equation}\label{hj_first}
   \begin{split}
     \DT{X}_t &= P_t\\
     \DT{P}_t &= -\GRADX V_0(X_t)
     \end{split}
\end{equation}
through  the characteristics path $(X_t,P_t)$ satisfying $\GRADX\theta(X_t)=:P_t$.
The amplitude function $\phi$ can be determined by requiring the ansatz \eqref{wkb_eq} to be a solution, which gives
\begin{eqnarray*}
     0 &=    &  (\HOPER - E)\PSIA \EXPWKB\\
       &=    & \Big( \underbrace{( \frac{1}{2} |\GRADX\theta|^2 + V_0 - E) }_{=0} \PSIA \\
       &\quad&  %
       { - \frac{1}{2M} \LAP_{X}\PSIA +(\VOPER - V_0)\PSIA
                - \frac{\Iunit}{M^{1/2}} ( \EPROD{\GRADX\PSIA}{\GRADX\theta}
                + \frac{1}{2}\PSIA\LAP_X\theta)\Big)} %
                \EXPWKB\COMMA
\end{eqnarray*}
so that by using  \eqref{theta_eq} we have
\[
   - \frac{1}{2M} \LAP_{X}\PSIA +(\VOPER - V_0)\PSIA
   - \frac{\Iunit}{M^{1/2}} ( \EPROD{\GRADX\PSIA}{\GRADX\theta}
   + \frac{1}{2}\PSIA\LAP_X\theta) =0 \PERIOD
\]
The usual method for determining $\PSIA$ from this so-called {\it transport equation}
uses an asymptotic expansion $\PSIA\simeq \sum_{k=0}^K M^{-k/2} \PSIA_k$, 
see \REFS{hagedorn_egen,martinez} and the beginning of Section~\ref{sec:wkb_analysis}.
An alternative is to write it as a Schr\"odinger equation, similar to work in \REF{maslov}:
we apply the characteristics in \eqref{hj_first} to write
\[
  \frac{d}{dt} \PSIA(X_t)= \EPROD{\GRADX\PSIA}{\DT{X}_t}=\EPROD{\GRADX\PSIA}{\GRADX\theta}\COMMA
\]
and define the weight function $G$ by
\begin{equation}\label{G_first}
  \frac{d}{dt} \log G_t= \frac{1}{2} \LAP_X\theta(X_t)\COMMA
\end{equation}
and the  variable $\tpsi_t:=\PSIA(X_t) G_t$.
We use the notation $\phi(X)$ instead of the more precise $\phi(\cdot, X)$, so that e.g. 
$\psi_t=\psi_t(x)=\phi(x,X_t)G_t$.
Then the transport equation becomes a Schr\"odinger equation
\begin{equation}\label{schrod_first}
  \Iunit M^{-1/2} \DT{\tpsi}_t =(\VOPER - V_0)\tpsi_t 
                               - \frac{G_t}{2M} \LAP_X\left(\frac{\tpsi_t}{G_t}\right)\PERIOD
\end{equation}
In conclusion, equations \eqref{theta_eq}-\eqref{schrod_first} determine the WKB-ansatz
\eqref{wkb_form} to be a solution to the Schr\"odinger equation \eqref{schrodinger_stat}.
\begin{theorem}\label{thr:schrodinger-hamiltonian_first}
Assume the Hamilton-Jacobi equation, with the  corresponding Hamiltonian,
   \[
     H_{\SCH}(X,P):=
     \frac{1}{2}|P|^2 
     + \underbrace{ \frac{\LPROD{\tpsi(X)}{\VOPER(X)\tpsi(X)}}{\LPROD{\tpsi(X)}{\tpsi(X)}} }_{=:V_0(X)} - E=0 \COMMA 
   \]
   based on  the primal variable $X$ and the dual variable
   $P = P(X)=\GRADX\theta(X)$, has a smooth solution $\theta(X)$, then $\theta$
   generates a solution to the time-independent Schr\"odinger equation $(\HOPER - E)\Phi=0$,
   in the sense that 
   \[
      \Phi(X_t,x)
        = \hat G^{-1}(X_t) \hat\psi(x,X_t) e^{\Iunit M^{1/2}\theta(X_t)}\COMMA
   \]
    solves the equation \eqref{schrodinger_stat},
    where $ \hat \psi(X_t) :=\tpsi_t$ satisfies the transport equation \eqref{schrod_first} and
   \[
   \begin{split}
      & \hat G(X_t) =G_t\COMMA\\
      & \frac{d}{dt}\log G_t =\frac{1}{2}\Delta_{X}\theta(X_t)\COMMA\\
      & \mbox{$(X_t,P_t)$  solves the Hamiltonian system \eqref{hj_first} corresponding to $H_{\SCH}$.}
   \end{split}
   \]
\end{theorem}
It is well know that Hamilton-Jacobi equations in general do not have smooth solutions,
due to $X$-paths that collide, as seen by \eqref{q_sym} generating blow up in $\partial_{XX}\theta(X)$.
However if the domain is small enough, the data on the boundary is smooth  and $V_0$ is smooth, then the characteristics generate a smooth solution, see Ref.~\cite{evans}. In Section~\ref{global_sec} we describe Maslov's method
to find a global solution by patching together local solutions.

Note that the nuclei density, using $\hat G$,  can be written
\begin{equation}\label{rho_definition}
  \rho := \frac{\LPROD{\PSIA}{\PSIA}}{\int_{\rset^{3N}} \LPROD{\PSIA}{\PSIA} \dX}=
  \frac{\LPROD{\hat\psi}{\hat\psi}\, \hat G^{-2}}{\int_{\rset^{3N}} \LPROD{\hat\psi}{\hat\psi} \,\hat G^{-2}\dX}\COMMA
\end{equation}
and since each time $t$ determines a unique point $(X_t,P_t)=(X_t,\GRADX\theta(X_t))$ in the phase space  
the functions $\hat G$ and $\hat\psi$ are well defined. 

\subsubsection{Liouville's Formula}\label{liouville}
In this section we verify Liouville's formula
\begin{equation}\label{liouville_form}
   \frac{G^2_0}{G_t^2}=\EXP{-\int_0^t \TRACE\left(\GRADX P(X_t)\right) \,dt}= 
   \left|\DET \frac{\partial (X_0)}{\partial (X_t)}\right| \COMMA
\end{equation}
given in \REF{maslov}.
The characteristic $\DT{X}_t= P(X_t)$ implies $\tfrac{d}{dt}\JAC{X_t} = \GRADX P\, \JAC{X}$, 
where $\JAC{X}_{ij}=\partial X^i_t/\partial X^j_0$ denotes the first variation with 
respect to perturbations of the initial data. The logarithmic derivative then satisfies
$d/dt\big(\log \JAC{X}\big)_{ij} = \partial_{X^j} P^i(X_t)=\partial_{X^iX^j}\theta(X)$ which implies that 
$\log \JAC{X_t}$ is symmetric and shows that  \eqref{liouville_form} holds
\[
  \DIV P=\TRACE \GRADX P = \frac{d}{dt} \TRACE \log\JAC{X} = \frac{d}{dt}  \log\DET\JAC{X}\PERIOD
\]
The last step uses that $\JAC{X}$ can be diagonalized by an orthogonal transformation
and that the trace is invariant under orthogonal transformations.

\subsubsection{Data for the Hamiltonian system}\label{data_hj}
For the energy $E$ chosen larger than the potential energy, that is such that $E\ge V_0$,  
the Hamiltonian system \eqref{hj_first}
yields a solution $(X,P): [0,T]\rightarrow U\times\rset^{3N}$ to the eikonal equation 
(\ref{theta_eq}) locally  in a neighborhood  $U\subseteq \rset^{3N}$,  for 
regular compatible data $(X_0,P_0)$ given on a $3N-1$ dimensional "inflow"-domain 
$ I\subset \overline{U}$.
Typically, the domain $I$ and the data $(X_0,P_0)|_I$ 
are not given (except that its total energy is $E$), unless it is really an
inflow domain and characteristic paths do not return to $I$ as in a scattering problem.
If paths leaving from $I$ return to $I$, there is an additional compatibility
of data on $I$: 
assume $X_0\in I$ and $X_t\in I$,
then the values $P_t$ are determined from $P_0$;
continuing the path to subsequent hitting points $X_{t_j}\in I$, $j=1,2,\ldots$ determines 
$P_{t_j}$ from $P_0$. 
The characteristic path $(X_t,P_t)$, $t>0$, generates a manifold in the phase space $(X,P)$, 
which is smooth under our assumptions. This manifold is in general only locally of the form $(X,P(X))$, 
but in the case of no caustics it is globally of this form
and then there is a phase function $X\mapsto \theta(X)$ such that $P(X)=\GRADX\theta(X)$ globally. 
In Section \ref{sec:caustics} we study phase space manifolds with caustics.

\begin{remark}\label{G_remark}
The integrating factor $G$ and its derivative $\partial_{X^i} G$ can be determined from 
$(P,\partial_{X^i} P,\partial_{X^iX^j}P)$
along the characteristics by
the following characteristic equations obtained from \eqref{theta_eq} by differentiation with respect to $X$ 
\begin{equation}\label{p_xx_eq}
\begin{split}
   \frac{d}{dt}{\partial_{X^r} P^k}   &=\left[ \sum_j P^j\partial_{X^jX^r}P^k=\sum_jP^j\partial_{X^rX^k}P^j\right]\\
                                      &= -\sum_j\partial_{X^r} P^j\partial_{X^k} P^j -\partial_{X^rX^k}V_0 \COMMA \\
   \frac{d}{dt}{\partial_{X^rX^q} P^k}&=\left[\sum_j P^j\partial_{X^jX^rX^q}P^k+\sum_jP^j\partial_{X^rX^kX^q}P^j\right]\\
                                      &= -\sum_j\partial_{X^r} P^j\partial_{X^kX^q} P^j 
                                         -\sum_j\partial_{X^rX^q} P^j\partial_{X^k} P^j -\partial_{X^rX^kX^q}V_0\COMMA\\
\end{split}
\end{equation}
and similarly $\partial_{X^iX^j}G$ can be determined from  
$(P,\partial_{X^i} P,\partial_{X^iX^j}P,\partial_{X^iX^jX^k}P)$.
\end{remark}

\subsection{Born-Oppenheimer dynamics}
The Born-Oppenheimer approximation leads to the 
standard formulation of {\it ab initio} molecular dynamics, in the micro-canonical ensemble
with the constant number of particles, volume and energy, for the nuclei positions  $X=X_{\BO}$,
\begin{equation}\label{bo_stated}
\begin{split}
  \dot{X}_t &=  P_t \COMMA \\
  \dot{P}_t &= -\GRADX \lambda_0(X_t) \COMMA
\end{split}
\end{equation}
by using that the electrons are in the eigenstate  $\tpsi=\Psi_{\BO}$ with eigenvalue $\lambda_0$ to $\VOPER$, in $L^2(dx)$
for fixed $X$, i.e., $\VOPER(X)\Psi_{\BO}=\lambda_0(X)\Psi_{\BO}$. 
The corresponding Hamiltonian is $H_{\BO}( X, P):= |P|^2/2 + \lambda_0(X)$ with
the eikonal equation
\begin{equation}\label{bo_ham_first}
\frac{1}{2}|\GRADX\theta_{\BO}(X)|^2 + \lambda_0(X) = E\PERIOD
\end{equation}

\subsection{Equations for the density}
We note that
\[
  \PSIA = \hat G^{-1} \hat\tpsi=\left(
         \frac{\rho}{\LPROD{\hat\tpsi}{\hat\tpsi}/
                     \int\LPROD{\hat\tpsi}{\hat\tpsi} \hat G^{-2}\, dX}
                       \right)^{1/2}\, \hat \tpsi\COMMA
\]
shows that $G$ and $\tpsi$ determine the density 
\begin{equation}\label{r_s_dens}
   \rho_{\SCH}=\rho=\frac{\LPROD{\hat\tpsi}{\hat\tpsi} |\hat G|^{-2}}{\int\LPROD{\hat\tpsi}{\hat\tpsi}\, |\hat G|^{-2} dX} \COMMA %
\end{equation} 
defined in \eqref{rho_definition}.
Using the Born-Oppenheimer approximation in Lemma~\ref{born_oppen_lemma} we have
$\LPROD{\hat\tpsi}{\hat\tpsi}= 1+ \BIGO(M^{-1})$ in the case of a spectral gap.
Therefore the weight function $|\hat G|^{-2}$ approximates the density and
we know from Theorem~\ref{thr:schrodinger-hamiltonian_first} 
that $|\hat G|^{-2}$ is determined by the phase function $\theta$.

The Born-Oppenheimer dynamics generates an approximate solution
$\Psi_{\BO}\hat G_{\BO}^{-1}\EXP{\Iunit M^{1/2}\theta_{\BO}}$ which yields the density
\begin{equation}\label{r_bo_dens}
  \rho_{\BO}=|\hat G_{\BO}|^{-2},
\end{equation} where
\[
  \frac{d}{dt} \log |\hat G_{\BO}|^{-2} = -\Delta_{X} \theta_{\BO}(X)\PERIOD
\]
This representation can also be obtained from the conservation of mass
\begin{equation}\label{mass_cons}
  0 = \DIV(\rho_{\BO}\GRADX\theta_{\BO})
\end{equation}
implying
\begin{equation}\label{approx_dens}
  \frac{d}{dt}{\rho}_{\BO}(X_t) =\EPROD{\GRADX \rho_{\BO}(X_t)}{\dot{X}_t}
                          = -\rho_{\BO}(X_t)\ \DIV \GRADX\theta_{\BO}\COMMA\\
\end{equation}
with the solution
\begin{equation}\label{density_fact}
  \rho_{\BO}( X_t)=\frac{C}{|\hat G_{\BO}(X_t)|^2}\COMMA
\end{equation}
where $C$ is a positive  constant for each characteristic. Note that the derivation of this classical density 
does not need a corresponding WKB equation but
uses only the conservation of mass 
that holds for classical paths
satisfying a Hamiltonian system. 
The classical density corresponds precisely to
the Eulerian-Lagrangian change of coordinates $ |G_t|^2/ |G_0|^2=\DET (\partial  X_t/\partial X_0)$
in \eqref{liouville_form}.

\subsection{Construction of the solution operator}\label{start_sec}
The WKB Ansatz  \eqref{wkb_form}  is meaningful when $\tpsi$ does not include the full small scale.
In Lemma~\ref{born_oppen_lemma} we present conditions for $\tpsi$ to be smooth.

To construct the solution operator it is convenient  to include
a non interacting particle in the system, i.e., a particle without charge,
and assume that this particle moves with a constant,  high speed $dX_1^1/dt=P_1^1\gg 1$ (or equivalently
with the unit speed and a large mass).
Such a non interacting particle does not affect the other particles. The additional new coordinate $X^1_1$
is helpful in order to simply relate the  time-coordinate $t$ and $X^1_1$. 
We add the corresponding kinetic energy $(P^1_1)^2/2$ to $E$ in order not to change the original problem
\eqref{schrodinger_stat} and
write the equation \eqref{schrod_first} in the fast time scale $\tau=M^{1/2}t$
\[
   \Iunit \frac{d}{d\tau}\tpsi=(\VOPER-V_0)\tpsi - \frac{1}{2M} G\sum_j\Delta_{X^j}(G^{-1}\tpsi)\PERIOD
\]
Furthermore, we change to the coordinates  
\[
   (\tau,X_*):=(\tau,X^1_2,X^1_3,X^2,\ldots,X^N)\in [0,\infty)\times I\COMMA\;\; 
    \mbox{instead of $(X^1,X^2,\ldots, X^N)\in\rset^{3N}$,}
\]
where $X^j=(X^j_1,X^j_2,X^j_3)\in \rset^3$. 
Hence we obtain
\begin{equation}\label{new_psi_eq}
 \Iunit \dot{\tpsi} +\frac{1}{2(P_1^1)^2}\ddot{\tpsi}
     =(\VOPER-V_0)\tpsi - \frac{1}{2M} G\sum_j\Delta_{X_*^j}(G^{-1}\tpsi) =: \tilde\VOPER\tpsi\COMMA
\end{equation}
using the notation $\dot w= dw/d\tau$ in this section. 
In Section \ref{start_sec2} we show that the left hand side can be reduced to $i\dot\tpsi$ as $P_1^1\rightarrow \infty$, by choosing special initial data.
Note also that $G$ is independent of the first component in $X^1$.
We see that the operator
\[
  \bar \VOPER:=G^{-1} \tilde \VOPER G=\underbrace{G^{-1} (\VOPER-V_0)G}_{=\VOPER-V_0} -\frac{1}{2M}\sum_{j}\Delta_{X^j_*}
\]
is symmetric on $L^2(\rset^{3n+3N-1})$. 
Assume now the data $(X_0,P_0,Z_0)$
for $X_0\in \rset^{3N-1}$ is $(L\mathbb Z)^{3N-1}$-periodic, then also $(X_\tau,P_\tau,Z_\tau)$
is $(L\mathbb Z)^{3N-1}$-periodic, for $Z_t=\theta(X_t)$ and $P_t=\GRADX\theta(X_t)$.
To simplify the notation for such periodic functions, define the periodic circle 
\[
   \tset:=\rset/(L\mathbb Z)\PERIOD
\]

We seek a solution $\Phi$ of \eqref{schrodinger_stat}
which is $(L\mathbb Z)^{3(n+N)-1}$-periodic in the $(x,X_*)$-variable.
The  Schr\"odinger operator $\bar \VOPER(\cdot, X_\tau)$ 
has, for each $\tau$, real eigenvalues $\{\lambda_m(\tau)\}$
with a complete set of eigenvectors $\{\zeta^m(x,X_*,\tau)\}$ orthogonal in the space
of $x$-anti-symmetric functions in $L^2(\mathbb T^{3n+3N-1})$, see \REF{berezin}.
The proof uses
that the operator $\bar \VOPER_\tau+\gamma I$ generates a compact solution operator
in the Hilbert space
of $x$-anti-symmetric functions in $L^2(\mathbb T^{3n+3N-1})$, for the constant $\gamma\in (0,\infty)$
chosen sufficiently large. The discrete spectrum and the compactness comes from 
Fredholm theory for compact operators and the fact
that the bilinear form $\int_{\tset^{3(n+N)-1}} v \bar \VOPER_\tau w +\gamma v w \, dx\,dX_*$
is continuous and coercive on $H^1(\tset^{3(n+N)-1})$, see \REF{evans}. %
We see that  $\tilde \VOPER$ has the same eigenvalues $\{\lambda_m(\tau)\}$ and the eigenvectors 
$\{G_\tau \zeta^m(\tau)\}$, 
orthogonal in the weighted $L^2$-scalar product 
\[
  \int_{\mathbb T^{3N-1}}\LPROD{v}{  w} \ \hat G^{-2}\, dX_*\PERIOD
\]
The construction and analysis of the solution operator continues in Section ~\ref{start_sec2}
based on the spectrum.

\begin{remark}[{Boundary conditions}]
The eigenvalue problem \eqref{schrodinger_stat} makes sense not only
in the periodic setting but also with alternative {boundary conditions}
from interaction with an external environment, e.g., for scattering problems.
\end{remark}

\section{Computation of observables}\label{sec:observables}
Suppose the goal is to compute a real-valued {\it observable} 
\[\int_{  \tset^{3N}}\LPROD{\Phi}{ A\Phi}\, dX\] for a given bounded linear multiplication operator $A=A(X)$
on $L^{2}(  \tset^{3N})$ and a solution $\Phi=\sum_k \phi_k e^{iM^{1/2}\theta_k}$ of \eqref{schrodinger_stat}.
We have
\begin{equation}\label{observ_expand}
  \begin{array}{ll}
    \int_{ \tset^{3N}}\LPROD{\Phi}{ A\Phi} dX
    & = \sum_{k,l} \int_{  \tset^{3N}}\LPROD{A\phi_k \EXP{\Iunit M^{1/2}\theta_k(X)}}{ \phi_l \EXP{\Iunit M^{1/2}\theta_l(X)}}\, dX\\
    & = \sum_{k,l}\int_{  \tset^{3N}}A \EXP{\Iunit M^{1/2}\left(\theta_l(X)-\theta_k(X)\right)} 
              \LPROD{\phi_k}{\phi_l} \, dX\PERIOD
  \end{array}
\end{equation}
The integrand is oscillatory for $k\ne l$, hence critical points (or near critical points) of the phase difference give the main contribution.
The stationary phase method, see \REFS{duistermaat,maslov} and 
Section~\ref{stat_phase_sec}, 
shows that these integrals  are small, 
bounded by $\BIGO(M^{-3N/4})$,
in the case when the phase  difference has non degenerate critical points, or no critical point, 
and the functions $A \LPROD{\phi_k}{ \phi_l}$ and $\theta_l$  are sufficiently smooth.
A critical point $X_c\in\R^{3N}$  satisfies $\GRADX \theta_l(X_c)-\GRAD_X\theta_k(X_c)=0$, 
which means that the two different paths, generated by $\theta_l$ and $\theta_k$,
passing through $X=X_c$ also have the same momentum $P$ at this point. 
That the critical point is degenerate means that
the Hessian matrix $\partial_{X^iX^j}(\theta_k-\theta_l)(X_c)$ is singular  (or asymptotically singular for $M\rightarrow\infty$ as for avoided crossings when the electron eigenvalues have a vanishing spectral gap depending on $M$).
Therefore caustics, crossing or avoided crossing electron eigenvalues may generate coupling between the
WKB terms. On the other hand, without such coupling
%
the density of a linear combination of WKB terms
separates asymptotically to a sum of densities of the individual WKB terms
\begin{equation}\label{obser_decay}
  \int_{  \tset^{3N}}\LPROD{\Phi}{ A\Phi} dX
  =\sum_{k=1}^{\bar k}\int_{  \tset^{3N}}A \underbrace{\LPROD{\phi_k}{\phi_k}}_{=\rho_k}\, dX +\BIGO(M^{-1})\COMMA
\end{equation}
in the case of multiple eigenstates, $\bar k>1$, and 
\[
  \int_{  \tset^{3N}}\LPROD{\Phi}{ A\Phi}\, dX
  =\int_{  \tset^{3N}}A\LPROD{\phi_1}{\phi_1}\, dX
\]
for a single eigenstate.
In the next section we will study molecular dynamics approximations of a single state
\begin{equation}\label{observable_sum}
  \int_{  \tset^{3N} }A\LPROD{\phi_k}{\phi_k}\, dX = \int_{  \tset^{3N}} A(X) \rho_k(X)\, dX\PERIOD 
\end{equation}

In the presence of a
caustic, the WKB terms can be asymptotically non orthogonal,  since their coefficients and phases typically are
not smooth enough to allow the integration by parts  to gain powers of $M^{-1/2}$.  
Non-orthogonal WKB functions tell how the caustic couples the WKB modes.

Regarding the inflow density  $\rho_k\big|_I$ there are two situations: either the characteristics
return often to the inflow domain or not. If they do not return we have a scattering problem and
it is reasonable to define the inflow-density  
$ \rho_k\big|_I$ as an initial condition. If characteristics return,
the dynamics can be used to estimate the return-density  $\rho_k\big|_I$ as follows:
Assume that the following limits exist
\begin{equation}\label{ergod_limit}
  \lim_{T\rightarrow\infty}
      \frac{1}{T}\int_0^{T} A(X_t) \,dt  = \int_{ \tset^{3N}} A(X) \rho_k(X) \, dX
\end{equation}
which bypasses   the need to find $\rho_k\big|_I$ and the quadrature in the number of characteristics.
A way to think about this limit  is to sample the return points $ X_t\in I$ and from these samples 
construct an {\it empirical} return-density, 
converging to $\rho_k\big|_I$ as the number of return iterations tends to infinity. 
We shall use this perspective to view the eikonal equation \eqref{theta_eq} as a hitting problem
on $I$, with hitting times $\tau$ (i.e., return times).
The property having $\rho|_I$ constant as a function
of $X_0$ is called {\it ergodicity}, which we will use. 
We could allow the density $\rho|_I$ to depend on the initial
position $X_0$ and momentum $P_0$, but then our observables need to conditional expected values.
 An example of a hitting surface is the co-dimension one surface
where the first component $X_{11}$ in $X_1=(X_{11},X_{12},X_{13})$  is equal to its initial value 
$X_{11}(0)$. The dynamics does not always have such a hitting surface: for instance if  all
particles are close initially and then are scattered away from each other, as in an explosion, no co-dimension
one hitting surface exists.

\section{Molecular dynamics approximation of Schr\"odinger observables}\label{sec_hj}
A numerical computation of an approximation to 
$\sum_{k}\int_{  \tset^{3N} }\LPROD{\phi_k}{A\phi_k}\, dX$ has the main ingredients:
\begin{itemize}
  \item[(1)] to approximate the exact characteristics by molecular dynamics characteristics \eqref{hj_first},
  \item[(2)] to discretize the  molecular dynamics equations, and 
  \item[(3a)] if $\rho\big|_I$ is  an inflow-density, to introduce quadrature in the number of characteristics, or
  \item[(3b)] if $\rho\big|_I$ is  a return-density, to replace the ensemble average by a time average 
              using the property \eqref{ergod_limit}. 
\end{itemize}
This section presents a derivation of the approximation error in the step (1) in the case of
a return density and comments on the time-discretization of step (2) treated in Section~\ref{md_sim}.
The third and fourth
discretization steps, which are not described here, are studied,
for instance, in \REFS{cances,lebris_hand,lebris}.

\subsection{The Born-Oppenheimer approximation error}\label{BO_stycke}
This section states our main result of molecular dynamics approximating Schr\"odinger
observables. We formulate it using the assumption of the Born-Oppenheimer property
\begin{equation}\label{bo_assump}
   \|\tpsi_t-\Psi_{\BO}(X_t)\|_{L^2(dx)}=\BIGO(M^{-1/2})\COMMA \quad \mbox{uniformly in $t$.}
\end{equation}
This assumption is then proved in Lemma~\ref{born_oppen_lemma}
based on a setting with a spectral gap.

\noindent
{\it The spectral gap condition.}
The electron eigenvalues $\{\lambda_k\}$
satisfy,  for some positive $c$, the spectral gap condition
\begin{equation}\label{gap_c1}
   \inf_{k\ne 0,\ Y\in D} |\lambda_k(Y)-\lambda_0(Y)|>c\COMMA
\end{equation}
where $D:=\{X_{\SCH}(t) \SEP t\ge 0\}\cup\{X_{\BO}(t) \SEP t\ge 0\}$ is the set of all nuclei positions
obtained from the Schr\"odinger characteristics $X=X_{\SCH}$ in Theorem~\ref{thr:schrodinger-hamiltonian_first} and from
the Born-Oppenheimer dynamics $X=X_{\BO}$ in \eqref{bo_stated}, for all considered initial data.

\begin{theorem}\label{bo_thm}
  Assume that the phase functions $\theta_{\SCH}$ and $\theta_{\BO}$
  are smooth solutions to the eikonal equations \eqref{theta_eq} and \eqref{bo_ham_first}
  and that the Born-Oppenheimer property \eqref{bo_assump} holds,
  then the zero-order Born-Oppenheimer dynamics \eqref{bo_stated}, 
  assumed to have the ergodic limit \eqref{ergod_limit} and
   bounded hitting times $\tau$ in \eqref{hitting}, \eqref{hitting_remark} and
  \eqref{hitting_remarken}, approximates time-independent
  Schr\"odinger observables, generated by Theorem~\ref{thr:schrodinger-hamiltonian_first} or the caustic case  
  in Section~\ref{sec:caustic_general},
  with error bounded by $\BIGO(M^{-1+\delta})$
  \begin{equation}\label{bo_estimate}
          \int_{\tset^{3N}} g(X)\rho_{\BO}(X)\, dX
               =\int_{\tset^{3N}} g(X)\rho_{\SCH}(X)\, dX + \BIGO(M^{-1+\delta})\COMMA\; \mbox{ for any $\delta>0$.}
\end{equation}
\end{theorem}

The proof is given in Sections~\ref{born_oppen} and \ref{sec:caustic_general}. 

\subsection{Why do symplectic numerical simulations of molecular dynamics work?}\label{md_sim}
The derivation of the approximation error for the Born-Oppenheimer dynamics, 
in Theorem~\ref{bo_thm}, also allows
to study perturbed systems.  For instance, the perturbed 
Born-Oppenheimer dynamics 
\[
   \begin{split}
      \dot X_t &= P_t + \GRADP H^{\epsilon}(X_t,P_t)\\
      \dot P_t &=-\GRADX\lambda_0(X_t) -\GRADX H^\epsilon(X_t,P_t)\COMMA
   \end{split}
\]
generated
from a perturbed Hamiltonian $H_{\BO}(X,P)+ H^\epsilon(X,P)=E$,
with the perturbation satisfying
\begin{equation}\label{stab_h_d}
   \|H^\epsilon\|_{L^\infty} \le \epsilon \quad \mbox{ for some $\epsilon\in (0,\infty)$}
\end{equation}
yields through \eqref{H_felet} and \eqref{g_felet} an additional error term $\BIGO(\epsilon)$
to the approximation of observables in \eqref{bo_estimate}. So called symplectic numerical methods are 
precisely those that can
be written as perturbed Hamiltonian systems, see \REF{ms}, and consequently we have a method 
to precisely analyze their numerical error
by combining an explicit construction of $H^\epsilon$ 
with the stability condition \eqref{stab_h_d}
to obtain $\BIGO\big((M^{-1} +\epsilon)^{1-\delta}\big)$ accurate approximations, provided
the corresponding phase function has bounded second difference quotients.
The popular St\"ormer-Verlet method  is symplectic and the positions $X$ coincides with those of
the symplectic Euler method,  for which $H^\epsilon$ is explicitly constructed in \REF{ms}
with  $\epsilon$ proportional to the time step. The construction in \REF{ms} is 
not using the modified equation
and formal asymptotics, instead a piecewise linear extension of the solution generates $H^\epsilon$.

\section{Analysis of the molecular dynamics approximation}\label{born_oppen}\label{sec:wkb_analysis}

Before we proceed with the analysis of the approximation error we motivate our results by a significantly simpler
case of a system {\it without electrons}.  We  use the densities \eqref{r_s_dens} and \eqref{r_bo_dens}
and we show heuristically how the characteristics can be used 
to estimate the difference $\rho_{\SCH}-\rho_{\BO}$, leading to
$\BIGO(M^{-1})$ accurate Born-Oppenheimer approximations of Schr\"odinger observables
\[
  \int g(X)\underbrace{\rho_{\SCH}(X)}_{\LPROD{\Phi}{\Phi}}\,dX=\int g(X)\rho_{BO}(X)\,dX +\BIGO(M^{-1})\PERIOD
\]

In the special case of no electrons,  the dynamics of $X$ does not depend on $\tpsi$
and therefore $X_{\BO}=X_{\SCH}=X$ and consequently $G_{\BO}=G_{\SCH}$.
The difference $\tpsi_{\SCH}-\tpsi_{\BO} $ can be
understood from iterative approximations of \eqref{schrod_first}
\begin{equation}\label{hat_noel}
  \frac{\Iunit}{M^{1/2}} \dot{\tpsi}_{k+1} -(\VOPER-V_0)\tpsi_{k+1}=
              \frac{1}{2M} G\Delta_{X}(G^{-1}\tpsi_k)
\end{equation}
with $\psi_0=0$.
Then $\tpsi_{\BO}=\tpsi_1$ is the Born-Oppenheimer approximation
and formally we have the iterations approaching the
full Schr\"odinger solution $\tpsi_k\rightarrow \tpsi_{\SCH}$ as $k\rightarrow \infty$.

In the special case of no electrons, there holds $\VOPER=V_0$, thus the transport equation
$\Iunit\dot{{\tpsi}}_1=0$ has constant solutions. We  let $\tpsi_1=1$ and
then $\tpsi_2-\tpsi_1$ is imaginary with its absolute value bounded by $\BIGO(M^{-1/2})$.
We write the iterations of $\tpsi_k$ by integrating \eqref{hat_noel} as the linear mapping
\[
  \tpsi_{k+1}=1 +\Iunit M^{-1/2}\hat\SOPER(\tpsi_k)
             =\sum_{l=0}^{k} \Iunit^l M^{-l/2}\hat\SOPER^l(\tpsi_1)\COMMA
\]
which formally shows that 
\[  
  |\tpsi_{\SCH} |^2=|\tpsi_1|^2 + 2\REAL\LPROD{\tpsi_{\SCH}-\psi_1}{\tpsi_1} + |\tpsi_{\SCH}-\tpsi_1|^2
                   =1+\BIGO(M^{-1})\PERIOD
\]
Consequently this special Born-Oppenheimer density satisfies 
\begin{equation}\label{g_dens}
   \rho_{\BO}=\underbrace{G_{\SCH}^{-2} \LPROD{\tpsi_{\SCH}}{\tpsi_{\SCH}}}_{=\rho_{\SCH}}+\BIGO(M^{-1})\COMMA
\end{equation}
since $G_{\BO}=G_{\SCH}$ and  $X$ do not depend on $\tpsi$. 

In the general case with electrons and a spectral gap,  we show in  Lemma~\ref{born_oppen_lemma}
that there is a solution $\tpsi_{\SCH}$  satisfying 
\begin{equation}\label{alfa_ekv}
   \|\tpsi_{\SCH}-\Psi_{\BO}\|_{L^2(dx)}= \BIGO(M^{-1/2})\COMMA
\end{equation}
for the electron eigenfunction $\Psi_{\BO}$, satisfying 
\[ 
   \VOPER(\cdot,X)\Psi_{\BO}(\cdot,X)=\lambda_{0}(X)\Psi_{\BO}(\cdot,X)
\] 
and the eigenvalue $\lambda_{0}(X)\in\rset$ with a (fixed) nuclei position $X$. Then 
the state $\psi_1$ equal to a constant, in the case of no electrons,
corresponds to the electron eigenfunction $\Psi_{\BO}$ in the case with electrons present.
In the general case the $X$ dynamics for the Schr\"odinger and the Born-Oppenheimer dynamics are not
the same, but we will show that \eqref{alfa_ekv}
implies that the Hamiltonians $H_{\SCH}$ and $H_{\BO}$ 
are $\BIGO(M^{-1})$ close.
Using stability of Hamilton-Jacobi equations, the phase functions
$\theta_{\SCH}$ and $\theta_{\BO}$ are then also close in the
maximum norm, which, combined with an assumption of
smooth phase functions, show that $|G_{\SCH}-G_{\BO}|=\BIGO(M^{-1+\delta})$ for any $\delta>0$.
Lemma~\ref{born_oppen_lemma} also shows that 
$|\LPROD{\tpsi_{\SCH}}{\tpsi_{\SCH}} -1|= \BIGO(M^{-1})$ and consequently the density bound 
$|\rho_{\SCH}-\rho_{\BO}|=\BIGO(M^{-1+\delta})$ holds.
To obtain the estimate \eqref{alfa_ekv} the important new property, compared to no electrons, 
is to use oscillatory cancellation in directions orthogonal to $\Psi_{\BO}$.

\subsection{Continuation of the construction of the solution operator}\label{start_sec2}

This section continues the construction of the solution operator started in Section~\ref{start_sec}.
Assume for a moment that $\tilde \VOPER$ is  independent of $\tau$. 
Then the solution to \eqref{new_psi_eq}
can be written
as a linear combination of the two exponentials
\[
   a\EXP{\Iunit\tau\AOPER_+} + b\EXP{\Iunit\tau \AOPER_-}
\]
where the two characteristic roots are the operators
\[
  \AOPER_\pm=(P_1^1)^2\left(-1\pm (1-2(P_1^1)^{-2}\tilde \VOPER)^{1/2}\right)\PERIOD
\]
We see that $\EXP{\Iunit\tau\AOPER_-}$ is a highly oscillatory solution on the fast $\tau$-scale with 
\[
  \lim_{P_1^1\to\infty} \frac{1}{(P_1^1)^2}\AOPER_- = -2\mathrm{Id} \COMMA
\]
while 
\begin{equation}\label{alfa_1_bound}
  \lim_{P_1^1\to\infty}\AOPER_+= -\tilde \VOPER \PERIOD
\end{equation}
Therefore we chose initial data 
\begin{equation}\label{tau_const}
   \Iunit\dot{\tpsi}|_{\tau=0}=-\AOPER_+\tpsi|_{\tau=0}
\end{equation}
to have $b=0$, which eliminates the fast scale, and  the
limit $P_1^1\to\infty$  determines the solution by
the Schr\"odinger equation
\[
  \Iunit\dot{ \tpsi}=\tilde \VOPER\tpsi\PERIOD
\]
The next section presents an analogous  construction for the
slowly, in $\tau$, varying  operator $\tilde \VOPER$.

\subsubsection{Spectral decomposition}\label{spekral_decomp}
Write \eqref{new_psi_eq} as the first order system
\[
  \begin{split}
    \Iunit \dot{\tpsi} &= \pi \\
    \Iunit \dot\pi     &=-2(P_1^1)^2 (\tilde \VOPER\tpsi-\pi) \COMMA
  \end{split}
\]
which for $\bar\psi:=(\tpsi,\pi)$ takes the form
\[
  \dot{\bar\psi}=\Iunit \BOPER\bar\psi\COMMA\;\;\; 
  \BOPER:= \left(\begin{array}{cc}
           0                       & -1\\
           2(P_1^1)^2\tilde \VOPER & -2(P_1^1)^2\\
                 \end{array}\right)\COMMA
\]
where the  eigenvalues $\Lambda_\pm$ , right eigenvectors $\QOPER_\pm$ and left eigenvectors $\QOPER^{-1}_\pm$ of 
the real ``matrix'' operator $\BOPER$ are
\[
\begin{split}
  \Lambda_\pm &:= (P_1^1)^2\left(-\ID \pm \left(\ID-2 (P_1^1)^{-2}\tilde\VOPER\right)^{1/2}\right)\COMMA\\
          \QOPER_+ &:=\left(\begin{array}{c}
                                 \ID \\
                        -\Lambda_+\\
                 \end{array}\right)\COMMA \;\;\;
          \QOPER_- :=\left(\begin{array}{c}
                         -\Lambda_-^{-1}\\
                                       \ID \\
                       \end{array}\right)\COMMA  \\
          \QOPER_+^{-1} &:={(\ID-\Lambda_+\Lambda_-^{-1})^{-1}}
                     \left(\begin{array}{c}
                                          \ID\\
                             \Lambda_-^{-1}\\
                     \end{array}\right)\COMMA\;\;\;
          \QOPER_-^{-1}  :={(\ID-\Lambda_+(\Lambda_-)^{-1})^{-1}}
                     \left(\begin{array}{c}
                             \Lambda_+\\
                                     \ID \\
                           \end{array}\right)\PERIOD
\end{split}
\]
We see that $\lim_{P_1^1\to\infty}\Lambda_+=-\tilde \VOPER$ and 
$\lim_{P_1^1\to\infty}(P_1^1)^{-2}\Lambda_-=-2\ID$.
The important property here is that the left eigenvector limit
$\lim_{P_1^1\to\infty} \QOPER_+^{-1}=(\ID,0) $ is constant, independent of $\tau$,
which implies that the $\QOPER_+$ component $\QOPER_+^{-1}\bar\psi=\tpsi$ decouples.
We obtain in the limit $P_1^1\to\infty$ the time-dependent Schr\"odinger equation
\[
\begin{split}
    \Iunit\dot{\tpsi}(\tau)&=\Iunit \frac{d}{d\tau}( \QOPER_+^{-1}{\bar\psi}_\tau)
                            =\Iunit \QOPER_+^{-1}\frac{d}{d\tau}{\bar\psi}_\tau  
                            ={}-\QOPER_+^{-1} \BOPER_\tau \bar\psi_\tau    \\
                           &={}-\Lambda_+(\tau) \QOPER_+^{-1} \bar\psi_\tau     
                            ={}-\Lambda_+(\tau) \tpsi(\tau)                  
                            =\tilde \VOPER_\tau \tpsi(\tau)\COMMA             \\ 
\end{split}
\]
where the operator $\tilde \VOPER_\tau$ depends on $\tau$ and $(x,X_0)$, and we define
the solution operator $\SOPER$
\begin{equation}\label{psi_evolution}
   \tpsi(\tau)=\SOPER_{\tau,0}\tpsi(0)\PERIOD
\end{equation}
As in \eqref{tau_const} we can view this as choosing special initial data for $\tpsi(0)$.
From now on we only consider such data.

The operator $\tilde \VOPER$ can be symmetrized 
\begin{equation}\label{v_sym}
  \bar\VOPER_\tau:=G_{\tau}^{-1} \tilde \VOPER_\tau G_{\tau}
   = (\VOPER-V_0)_\tau - \frac{1}{2M} \sum_{j} \Delta_{X^j_*},
\end{equation}
with real eigenvalues  $\{\check\lambda_m\}$ and orthonormal eigenvectors $\{\zeta^m\}$ in $L^2(dx\,dX_*)$,
satisfying 
\[
  \bar \VOPER_\tau \zeta^m(\tau)=\check\lambda_m(\tau) \zeta^m(\tau)\PERIOD
\]
Therefore $\tilde \VOPER_\tau$ has the same eigenvalues and the eigenvectors
$\bar\zeta^m:=G_\tau \zeta^m$, which establishes  the spectral representation
\begin{equation}\label{spectrum}
  \tilde \VOPER_\tau \tpsi(\cdot,\tau,\cdot)=\sum_m \check\lambda_m(\tau)
         \int_{\tset^{3N-1} }\LPROD{\tpsi(\cdot,\tau,\cdot)}{ \bar \zeta^m}  G_\tau^{-2} dX_*\,\bar \zeta^m(\tau)\PERIOD
\end{equation}
We note that the weight $G^{-2}$ on the co-dimension one surface $\tset^{3N-1}$
appears  precisely because the operator $\tilde \VOPER$ is symmetrized by $G^{-2}$
and the weight $G^{-2}$ corresponds to the Eulerian-Lagrangian change of coordinates \eqref{liouville_form}
\begin{equation}\label{g_vikt}
    \int_{\tset^{3N-1} }\LPROD{\tpsi}{ \bar \zeta^m } G_\tau^{-2}\, dX_*
   =\int_{\tset^{3N-1} }\LPROD{\tpsi}{\bar \zeta^m}\, dX_0\PERIOD
\end{equation}
The existence of the orthonormal set of eigenvectors and real eigenvalues
makes the operator $\tilde \VOPER$  self-adjoint 
in the Lagrangian coordinates and hence the solution operator $\SOPER$
becomes unitary in the Lagrangian coordinates.

 \subsection{Stability from perturbed Hamiltonians}\label{pert_ham}
In this section we derive error estimates
of the weight functions $G$
when the corresponding Hamiltonian system is perturbed.
To derive the stability estimate
we consider the 
Hamilton-Jacobi equation  \[H(\GRADX\theta(X),X)=0\] in an optimal control
perspective with the corresponding Hamiltonian system
\[
  \begin{split}
    \dot X_t &=  \GRADP H(P_t,X_t)   \\ 
    \dot P_t &= -\GRADX H(P_t,X_t)\PERIOD
  \end{split}
\]
We define  the ``value'' function
\[
  \theta(X_0)= \theta(X_t) - \int_0^t h(P_s,X_s)\, ds\COMMA
\]
where the ``cost'' function defined by
\[
   h(P,X):=  \EPROD{P}{\GRADP  H(P,X)}-H(P,X)
\]
satisfies the Pontryagin principle (related to the Legendre transform)
\begin{equation}\label{pontry}
  H(P,X)= \sup_Q \big( \EPROD{P}{ \GRAD_Q H(Q,X)} - h(Q,X)\big)\PERIOD
\end{equation}
Let $\theta\Big|_I$ be defined by the hitting problem
\[
  \theta(X_0)=\theta(X_\tau) -\int_0^\tau  h(P_s,X_s)\, ds
\]
using the hitting time $\tau$ on the return surface $I$
\begin{equation}\label{hitting}
   \tau:= \inf\{ t \SEP  X_0\in I,\, X_t\in I\, \& \, t>0\}\PERIOD
\end{equation}
For a perturbed Hamiltonian $\tilde H$
and its dynamics $(\tilde X_t,\tilde P_t)$ we define analogously the value function $\tilde\theta$ and
the cost function $\tilde h$.

We can think of the difference $\theta-\tilde\theta$ as composed by a perturbation of the boundary
data (on the return surface $I$) and perturbations of the Hamiltonians. 
The difference  of the value functions due to the perturbed Hamiltonian satisfies the stability estimate
\begin{equation}\label{theta_alfa}
\begin{split}
   \theta(X_0)-\tilde\theta(X_0) & \ge  \theta(\tilde X_{\tilde\tau})-\tilde\theta(\tilde X_{\tilde\tau}) +
         \int_0^{\tilde\tau} (H-\tilde H)\left(\GRAD_X\theta(\tilde X_t),\tilde X_t\right)\, dt \\
   \theta(X_0)-\tilde\theta(X_0) &\le \theta( X_{\tau})-\tilde\theta( X_{\tau}) +
         \int_0^{\tau} (H-\tilde H)\left(\GRAD_X\tilde\theta(X_t), X_t\right)\, dt 
\end{split}
\end{equation}
with a difference of the Hamiltonians evaluated along the same solution path.
This result follows by differentiating the value function along a path and
using the Hamilton-Jacobi equations, see Remark~\ref{hj_stab} and \REF{css}.

We assume that
\begin{equation}\label{H_felet}
  \sup_{(P,X)=(\GRADX\tilde\theta(X_t), X_t),\,  (P,X)=(\GRADX\theta(\tilde X_t), \tilde X_t)}
           |(H-\tilde H)(P,X)|=\BIGO(M^{-1})\COMMA
\end{equation}
which is verified in \eqref{H_verif} for Schr\"odinger and Born-Oppenheimer Hamiltonians.
We choose the hitting set as 
\begin{equation}\label{hitting_remark}
   I:=\{X\in\tset^{3N}\SEP \theta(X)=\tilde\theta(X)\}
\end{equation}
on which the two phases coincide.
Now assume that $I$ forms a codimension one set in $\tset^{3N}$
and that the maximal hitting time $\tau$ for characteristics starting on $I$ is bounded;
the fact that $I$ is a codimension one set holds, for instance, locally if 
$|\GRADX(\theta-\tilde\theta)|$ is nonzero. In fact, 
it is sufficient to assume that there exists a function $\gamma:\tset^{3N}\to\rset$,
satisfying $\gamma=\BIGO(M^{-1})$, and
such that the set 
$I:=\{X\in\tset^{3N}\SEP\theta(X)-\tilde\theta(X)=\gamma(X)\}$
 is a codimension one set with bounded hitting times.
Then the representation \eqref{theta_alfa}, for any time $t$ replacing $\tau$ and $\tilde\tau$,
together with the stability of the Hamiltonians \eqref{H_felet}
and the initial data $(\theta-\tilde\theta)|_I=0$ obtained from \eqref{hitting_remark} imply that
\begin{equation}\label{teta_uppskatt}
  \|\theta -\tilde\theta\|_{L^\infty} = \BIGO(M^{-1})\COMMA
\end{equation}
provided the maximal hitting time $\tau$ is bounded,
which we assume.

When the value functions $\theta$ and $\tilde\theta$ are smoothly differentiable in $X$
with  derivatives bounded uniformly in $M$, 
the stability estimate \eqref{theta_alfa} implies that also the difference of the second derivatives
has the bound
\begin{equation}\label{theta_der_stab}
  \| \LAP_{X}\theta - \LAP_{X}\tilde\theta \|_{L^\infty}=\BIGO(M^{-1+\delta})\COMMA\;
   \mbox{ for any  $\delta>0$.}
\end{equation}

Our goal is to analyze the  density function $\rho=|G|^{-2}\LPROD{\tpsi}{\tpsi}$ with $G$ defined in \eqref{G_first}.
The Born-Oppenheimer approximation \eqref{bo_assump} yields 
$\LPROD{\tpsi}{\tpsi}=1+\BIGO(M^{-1})$
thus it remains to estimate the weight function $|G|^{-2}$. This weight function satisfies the 
Hamilton-Jacobi equation
\begin{equation}\label{G_HJ}
   H_G(\GRADX\log |G|^{-2}, X):=\EPROD{\GRADX\theta(X)}{ \GRADX \log |G|^{-2}} +
   \LAP_{X} \theta(X)=0\PERIOD
\end{equation}
The stability of Hamilton-Jacobi equations
can then be applied to \eqref{G_HJ}, as in \eqref{theta_alfa},
using now the hitting set
\begin{equation}\label{hitting_remarken}
   I:=\{X\in\tset^{3N}\SEP \log |G(X)|^{-2}= \log|\tilde G(X)|^{-2} \}
\end{equation}
and the assumption of bounded hitting times $\tau$ in the hitting problem, and we obtain
\begin{equation}\label{g_felet}
\|\log|G|^{-2}-\log|\tilde G|^{-2}\|_{L^\infty} \le C\| H_G- H_{\tilde G}\|_{L^\infty}=\BIGO(M^{-1+\delta})\PERIOD
\end{equation}
In this sense we will use that an $\BIGO(M^{-1})$ perturbation of the Hamiltonian yields an error
estimate of almost the same order for the difference of the corresponding densities $\rho-\tilde\rho$.

The Hamiltonians we use are
\[
\begin{split}
   H_{\SCH} &= \frac{|P|^2}{2} + \frac{\LPROD{\tpsi(X)}{ \VOPER(X)\tpsi(X)}}{\LPROD{\tpsi(X)}{\tpsi(X)}}-E\COMMA\\
   H_{\BO} & = \frac{|P|^2}{2} +  \lambda_0(X) -E\COMMA \\
\end{split}
\]
based on the cost functions
\[
\begin{split}
   h_{\SCH}  & = E + \frac{|P|^2}{2}  -\frac{\LPROD{\tpsi(X)}{\VOPER(X)\tpsi(X)}}{\LPROD{\tpsi(X)}{\tpsi(X)}}\COMMA\\
   h_{\BO}   & = E + \frac{|P|^2}{2} - \lambda_0(X)\PERIOD
\end{split}
\]
For the Born-Oppenheimer case the electron wave function is the eigenstate $\Psi_{\BO}$. The Born-Oppenheimer
approximation \eqref{bo_assump}, proved in Lemma \ref{born_oppen_lemma}, implies that 
\begin{equation}\label{H_verif}
\|  H_S- H_{BO}\|_{L^\infty}=\mathcal O(M^{-1})\COMMA
\end{equation} which verifies \eqref{H_felet}.

\begin{remark}\label{hj_stab}
This remark derives the stability estimate \eqref{theta_alfa}.
The definitions of the value functions imply
\begin{equation}\label{eq:err2}
  \begin{split}
    &\underbrace{\tilde \theta(\tilde X_{\tilde\tau}) 
    -\int_0^{\tilde\tau} \tilde h(\tilde P_t,\tilde X_t)\,dt}_{\tilde\theta(\tilde X_0)} - 
    \underbrace{ \left(\theta( X_\tau) -\int_0^\tau h(P_t, X_t)\,dt \right) }_{\theta(X_0)}\\
    &= -\int_0^{\tilde\tau} \tilde h(\tilde P_t,\tilde X_t) \,dt  +\theta(\tilde X_{\tilde \tau})
       -\underbrace{ \theta(X_0)}_{ \theta(\tilde X_0)} 
       +\tilde\theta(\tilde X_{\tilde \tau}) - \theta(\tilde X_{\tilde \tau})\\
    &= -\int_0^{\tilde\tau} \tilde h(\tilde P_t,\tilde X_t) \,dt  
       +\int_0^{\tilde\tau} d\theta(\tilde X_t)
       +\tilde\theta(\tilde X_{\tilde \tau}) - \theta(\tilde X_{\tilde \tau})\\
    &=\int_0^{\tilde\tau}\underbrace{-\tilde h(\tilde P_t,\tilde X_t) +
        \EPROD{\GRADX\theta (\tilde X_t)}{ \GRADP\tilde H(\tilde P_t,\tilde X_t)} }_{\le 
        \tilde H\left(\GRADX\theta(\tilde X_t), \tilde X_t\right) }\,dt
       +\tilde\theta(\tilde X_{\tilde \tau}) - \theta(\tilde X_{\tilde \tau})\\
    &\le \int_0^{\tilde\tau} (\tilde H-H)\left(\GRADX\theta(\tilde X_t), \tilde X_t\right)\,dt
       +\tilde\theta(\tilde X_{\tilde \tau}) - \theta(\tilde X_{\tilde \tau})\COMMA
  \end{split}
\end{equation}
where the Pontryagin principle \eqref{pontry} yields the inequality and
we use the Hamilton-Jacobi equation 
\[
  H(\GRADX \theta(\tilde X_t), \tilde X_t)=0\PERIOD
\]
To establish the lower bound we replace $\theta$ along with $\tilde X_t$ by
$\tilde\theta$ and $X_t$ and repeat the derivation above.
\end{remark}

\subsection{The Born-Oppenheimer approximation}\label{sec:bo}
The purpose of this section is to present a case when the Born-Oppenheimer approximation holds in
the sense that $\|\tpsi-\Psi_{\BO}\|_{L^2(dx)}$ is small.

 We know from Section~\ref{spekral_decomp} that the solution $\tpsi_t=\SOPER_{t,0}\tpsi_0$ is bounded
 in $L^2(dx\,dX)$, since $\SOPER$ is unitary in the Lagrangian coordinates.
 This unitary $\SOPER$ implies that the integral in the Lagrangian
 coordinates $\int_{\tset^{3N-1}} \LPROD{\tpsi_t}{\tpsi_t } \, d X_0$ is constant in time.
 We consider the co-dimension one set 
 \[
    I_{\tpsi}:= \{ X\in \rset^{3N} \SEP  \LPROD{\tpsi(X)}{\tpsi(X)}
              =\int_{\tset^{3N-1}} \LPROD{\tpsi(t,X_0)}{\tpsi(t,X_0)}\,  d X_0
                        /\int_{\tset^{3N-1}}\, d X_0\}\COMMA
\]
where the  point values of $\LPROD{\tpsi(X)}{\tpsi(X)}$
coincides with its $L^2$ average.  
We choose a time $t$ such that $X_t\in I_{\tpsi}$ and assume that the time $\tau^*$
it takes to hit $I_{\tpsi}$ the next time  is bounded, i.e.,
 \[
   \tau^*:=\inf\{ \tau\SEP X_{t}\in I_{\tpsi},\, \tau>0\  \& \ X_{t+\tau}\in I_{\tpsi}\}=\BIGO(1)\PERIOD
 \]
 We also assume that all functions of $X$ are smooth.

\begin{lemma}\label{born_oppen_lemma}
  Assume that $\Iunit\dot{\tpsi}=M^{1/2}\tilde \VOPER\tpsi$ holds,
  then there exists initial data for $\tpsi$
  such that the $L^2(dx)$ orthogonal decomposition $\tpsi=\bar\psi_0 \oplus\psi_0^\perp$, 
  where $\bar\psi_0=\alpha\Psi_{\BO}$
  for some $\alpha\in\mathbb C$ satisfies 
  \begin{equation}\label{omega_estimate}
    \begin{split}
       \frac{\|\psi_0^\perp(t)\|_{L^2(dx)}}{\|\bar\psi_0(t)\|_{L^2(dx)}}&=\BIGO(M^{-1/2})\\
       |\LPROD{\tpsi_t}{ \tpsi_t} - 1|&=\BIGO(M^{-1})\\
       \|\tpsi_t -\Psi_{\BO}(X_t)\|_{L^2(dx)}&=\BIGO(M^{-1/2})\\
    \end{split}
  \end{equation}
  uniformly in time $t$, provided the spectral gap condition \eqref{gap_c1} holds,
  the smoothness estimate \eqref{beta_est} is satisfied
  and the hitting time $\tau^*$ is bounded.
\end{lemma}

\begin{proof}
  We consider the decomposition
  $\tpsi=\bar\psi_0\oplus\psi_0^\perp$,
  where $\bar\psi_0(\tau)$ is an eigenfunction of $\VOPER(X_\tau)$ in $L^2(dx)$,  
  satisfying $\VOPER(X_\tau)\bar \psi_0(\tau)=\lambda_0(\tau)\bar\psi_0(\tau)$
  for the eigenvalue $\lambda_0(\tau)\in\rset$. 
  This {\it ansatz} is motivated by the zero residual
  \begin{equation}\label{R_residual}
     \ROPER\tpsi:= \dot\tpsi +\Iunit M^{1/2}\tilde \VOPER\tpsi= 0 
  \end{equation}
  and the small residual for the eigenfunction
 \[\begin{split} %
     \LPROD{\PSHARP{(\dot{\bar\psi}_0)}}{\bar\psi_0} &=0 \\ 
     M^{1/2}\tilde \VOPER\bar\psi_0                  &=\BIGO(M^{-1/2})\COMMA 
  \end{split}\] %
  where 
  \begin{equation}\label{natural}
      w(X)=\LPROD{\Psi_{\BO}(X)}{w(X)} \Psi_{\BO}(X)\oplus \PSHARP{w(X)}
  \end{equation}
  denotes the orthogonal decomposition in the eigenfunction direction $\Psi_{\BO}$ and
  its orthogonal complement in $L^2(dx)$. We consider first  the linear operator $\ROPER$ 
  in \eqref{R_residual} with a given function $V_0$ and then we use a contraction 
  setting to show that $V_0=\LPROD{\tpsi}{ \VOPER\tpsi}/\LPROD{\tpsi}{\tpsi}$ also works since 
  $\|\bar\psi_0^\perp\|_{L^2(dx)}$ is small. 
  The orthogonal splitting  $\tpsi=\bar\psi_0\oplus\psi_0^\perp$ and the projection 
  $\PSHARP{(\cdot)}$ in \eqref{natural} imply
  \[
    \begin{split}
        0 &= \PSHARP{\left(\ROPER(\bar\psi_0+\psi_0^\bot)\right)}\\
          &= \PSHARP{\left(\ROPER(\bar\psi_0)\right)} + \PSHARP{\left(\ROPER(\psi_0^\bot)\right)}\\
          &= \PSHARP{(\ROPER\bar\psi_0)} + \dot{\psi_0}^\bot +\Iunit M^{1/2}(\VOPER-V_0)\psi_0^\bot + 
                                     \Iunit\PSHARP{\left(\frac{G M^{-1/2}}{2}\LAP_X(G^{-1}\psi_0^\bot)\right)}\COMMA
    \end{split}
  \] 
  where the last step  follows from the orthogonal splitting 
  \[
      \PSHARP{\left((\VOPER-V_0)\psi_0^\bot\right)}=(\VOPER-V_0)\psi_0^\bot
  \]
  together with the second order change in the subspace projection
  \[
     \psi_0^\bot(\tau+\Delta\tau)=\PSHARPT{\tau+\Delta\tau}{\left(\psi_0^\bot(\tau+\Delta \tau)\right)}= 
     \PSHARPT{\tau}{\left(\psi_0^\bot(\tau+\Delta \tau)\right)} + \BIGO(\Delta \tau^2)
  \]
  which yields $\PSHARP{(\dot\psi_0^\bot)}=\dot\psi_0^\bot$;
   here $\PSHARPT{\tau}{\cdot}$ denotes the projection on the orthogonal complement to the eigenvector $\bar\psi_0(\tau)$.
  To explain the second order change start with a function $v$ satisfying $\langle v,\Psi_{BO}(X_\tau)\rangle = 0$
  and $\Psi_{BO}(X_\sigma)=\Psi_{BO}(X_\tau) + \mathcal O(\Delta \tau)$ for $\sigma\in [\tau,\tau+\Delta \tau]$
  to obtain
  \[
  \begin{split}
  \Pi(\sigma)\big(\Pi(\tau+\Delta\tau)v-\Pi(\tau)v\big)&= 
  \Pi(\sigma)\Big(\langle v,\Psi_{BO}(X_{\tau})\rangle\Psi_{BO}(X_{\tau})
  -  \langle v,\Psi_{BO}(X_{\tau+\Delta\tau})\rangle\Psi_{BO}(X_{\tau+\Delta\tau})\Big)\\
  &=\Pi(\sigma) \mathcal O(\Delta\tau^2) 
  + \Pi(\sigma)\Big(\langle v,\mathcal O(\Delta\tau)\rangle \Psi_{BO}(X_\tau)\Big)\\
  & =\mathcal O(\Delta\tau^2) 
  + \mathcal O(\Delta\tau)\Big(\Psi_{BO}(X_\tau)-\langle\Psi_{BO}(X_\tau),\Psi_{BO}(X_\sigma)\rangle \Psi_{BO}(X_\sigma)\Big)\\
  &=\mathcal O(\Delta\tau^2).
  \end{split}
  \]
   Let $\widetilde \SOPER_{\tau,\sigma}$ be the solution operator from time $\sigma$ to $\tau$ for the generator 
  \[
     v\mapsto \Iunit M^{1/2}(\VOPER-V_0)v + \Iunit\PSHARP{\left(\frac{GM^{-1/2}}{2}\LAP_X(G^{-1}v)\right)}=:
      \Iunit M^{1/2}\hat \VOPER v\PERIOD
  \]
  Consequently, the perturbation $\psi_0^\bot$ can be determined from the projected residual
  \[
     \dot{\psi}_0^\bot =-\Iunit M^{1/2} \hat \VOPER \psi_0^\bot - \PSHARP{(\ROPER\bar\psi_0)}
  \]
  and we have the solution representation 
  \begin{equation}\label{r_int}
    \psi_0^\bot(\tau)= \widetilde \SOPER_{\tau,0} \psi_0^\bot(0)
                       -\int_0^\tau  \widetilde \SOPER_{\tau,\sigma} 
                         \PSHARP{\left(\ROPER\bar\psi_0(\sigma)\right)}\, d\sigma\PERIOD
  \end{equation}
  Integration by parts introduces the factor $M^{-1/2}$ we seek
  \begin{equation}\label{s_int_def}
  \begin{split}
          \int_0^\tau \widetilde \SOPER_{\tau,\sigma} \PSHARP\ROPER{\bar\psi_0(\sigma)}\, d\sigma&=
          \int_0^\tau \Iunit M^{-1/2}\frac{d}{d\sigma} 
                      (\widetilde \SOPER_{\tau,\sigma})\hat \VOPER^{-1}  \PSHARP\ROPER{\bar\psi_0(\sigma)}\, d\sigma\\
             &=\int_0^\tau \Iunit M^{-1/2}\frac{d}{d\sigma} \left(\widetilde \SOPER_{\tau,\sigma}\hat \VOPER^{-1}  
                    \PSHARP\ROPER{\bar\psi_0(\sigma)}\right)\, d\sigma\\
             &\qquad -\int_0^\tau \Iunit M^{-1/2} \widetilde \SOPER_{\tau,\sigma}\frac{d}{d\sigma}
                 \left(\hat \VOPER^{-1}(X_\sigma)  \PSHARP\ROPER{\bar\psi_0(\sigma)}\right)\, d\sigma \\
             &= \Iunit M^{-1/2} \hat \VOPER^{-1}  \PSHARP\ROPER{\bar\psi_0(\tau)}
                - \Iunit M^{-1/2} \widetilde \SOPER_{\tau,0}\hat \VOPER^{-1}  \PSHARP\ROPER{\bar\psi_0(0)}\\
             &\qquad -\int_0^t \Iunit M^{-1/2} \widetilde \SOPER_{\tau,\sigma}\frac{d}{d\sigma}
                \left(\hat \VOPER^{-1}(X_\sigma)  \PSHARP\ROPER{\bar\psi_0(\sigma)}\right)\, d\sigma\PERIOD\\
  \end{split}
  \end{equation}
  To analyze the integral in the right hand side we will use  the fact 
  \[
  \hat \VOPER^{-1}=\left(I +(\VOPER-V_0)^{-1}\left[\hat \VOPER-(\VOPER-V_0)\right]\right)^{-1}
                      (\VOPER-V_0)^{-1},\]
 which can be verified by multiplying both sides
 from the left by $I +(\VOPER-V_0)^{-1}\left[\hat \VOPER-(\VOPER-V_0)\right]$.
 A spectral decomposition in $L^2(dx)$, based on the electron 
  eigenpairs $\{\lambda_k,\bar\psi_k\}_{k=1}^\infty$ and satisfying 
  $\VOPER\bar\psi_k=\lambda_k\bar\psi_k$, 
  then implies
  \begin{equation}\label{vtilde}
  \begin{split}
      \hat \VOPER^{-1}\PSHARP(\ROPER\bar\psi_{0})
             &=\left(I +(\VOPER-V_0)^{-1}\left[\hat \VOPER-(\VOPER-V_0)\right]\right)^{-1}
                      (\VOPER-V_0)^{-1}\PSHARP(\ROPER\bar\psi_{0})\\
             &=\sum_{k\ne 0} \left(I +(\VOPER-V_0)^{-1}\left[\hat \VOPER-(\VOPER-V_0)\right]\right)^{-1}
                    (\lambda_k -V_0)^{-1}\psi_k \LPROD{\PSHARP(\ROPER\bar\psi_{0})}{ \psi_k}\\
             &=\sum_{k\ne 0}  (\lambda_k -V_0)^{-1} \psi_k \LPROD{\PSHARP(\ROPER\bar\psi_{0})}{\psi_k}+\BIGO(M^{-1})
  \end{split}
  \end{equation}
  which applied to the integral in the right hand side of \eqref{s_int_def}
  shows that $\|\bar\psi_0^\bot\|_{L^2(dx)}=\BIGO(M^{-1/2})$ on a bounded time interval,
  when the spectral gap condition holds and $\psi_k$ are smooth.

  The evolution on longer times requires an additional idea:
  one can integrate by parts recursively in \eqref{s_int_def} to obtain
\[
  \begin{split}
       \int_0^\tau \widetilde \SOPER_{\tau,\sigma} \PSHARP\ROPER{\bar\psi_0(\sigma)}\, d\sigma
       &=\left[\widetilde \SOPER_{\tau,\sigma}\Big(\BTOPER \tilde\ROPER -\BTOPER \frac{d}{d\sigma}(\BTOPER\tilde\ROPER) + 
               \BTOPER \frac{d}{d\sigma}
               \big(\BTOPER \frac{d}{d\sigma}(\BTOPER\tilde\ROPER)\big)-\ldots\Big)\right]_{\sigma=0}^{\sigma=\tau}\COMMA \\
                      \BTOPER &:= \Iunit M^{-1/2}\hat \VOPER^{-1}\COMMA\;\;\;
                 \tilde\ROPER := \PSHARP\ROPER{\bar\psi_0(\sigma)}\COMMA
  \end{split} 
\]
so that by \eqref{r_int} we have
  \[
      \psi_0^\bot(\tau)= \widetilde \SOPER_{\tau,0} \psi_0^\bot(0)
     - \left[\widetilde \SOPER_{\tau,\sigma}\Big(\BTOPER \tilde\ROPER -\BTOPER \frac{d}{d\sigma}(\BTOPER\tilde\ROPER) 
     + \BTOPER \frac{d}{d\sigma}
               \big(\BTOPER \frac{d}{d\sigma}(\BTOPER\tilde\ROPER)\big)-\ldots\Big)\right]_{\sigma=0}^{\sigma=\tau}\PERIOD
\]
By choosing   
\[
  \bar\psi_0^\bot(\sigma)\Big|_{\sigma=0}
  =-\Big(\BTOPER \tilde\ROPER(\sigma) -\BTOPER \frac{d}{d\sigma}(\BTOPER\tilde\ROPER)(\sigma) + \BTOPER \frac{d}{d\sigma}
               \big(\BTOPER \frac{d}{d\sigma}(\BTOPER\tilde\ROPER)\big)(\sigma)-\ldots\Big)\Big|_{\sigma=0}
\] 
we  get
\begin{equation}\label{beta_exp}
  \bar\psi_0^\bot(\tau) =  -\sum_{n=0}^\infty \BTOPER_0^n \ROPER_0(\tau)\COMMA
\end{equation}
where $\BTOPER_0:=-\Iunit M^{-1/2}\hat\VOPER^{-1} \tfrac{d}{d\tau}$ and $\ROPER_0:=\Iunit M^{-1/2}\hat\VOPER^{-1}\tilde\ROPER$.
We assume this expansion \eqref{beta_exp} is convergent in $L^2(dx)$ for each $\tau$, which follows from 
the smoothness estimate
\begin{equation}\label{beta_est}
   \|\BTOPER_0^n\ROPER_0(\tau)\|_{L^2(dx)}\rightarrow 0 \mbox{ as $n\rightarrow\infty$}
\end{equation} 
and \eqref{vtilde}.
 
The next step, verifying that also the non linear problem for $V_0$ works, is based on
the contraction obtained from
\[
    V_0- \lambda_0=\frac{\LPROD{\tpsi}{(\VOPER-\lambda_0)\tpsi}}{\LPROD{\tpsi}{\tpsi}}= \BIGO(\|\psi_0^\perp\|_{L^2(dx)})
\]
and that $\psi_0^\perp$ depends on $V_0$  in \eqref{r_int}, \eqref{s_int_def}  and \eqref{vtilde}
with a multiplicative factor $\BIGO(M^{-1/2})$.

Finally, to conclude that $|\langle\tpsi,\tpsi\rangle -1|=\BIGO(M^{-1})$, we use the evolution equation
\[
    \frac{d}{dt} \LPROD{\tpsi}{\tpsi} = M^{-1/2}|G|^2 \IMAG \LPROD{\Delta\frac{\tpsi}{G}}{ \frac{\tpsi}{G}}
           =\BIGO(M^{-1})
\]
where the last equality uses the obtained bound of $\psi_0^\perp$ in the first part of \eqref{omega_estimate}. 
The assumption of a finite hitting time $\tau^*$ then implies
that $|\LPROD{\tpsi}{\tpsi}-1|=\BIGO(\tau^*M^{-1})=\BIGO(M^{-1})$,
since we may assume that $\LPROD{\tpsi}{\tpsi}=1$ on $I_{\tpsi}$.
\end{proof}

\begin{remark}[Error estimates for the densities]
We have the densities
\begin{align}
   &\rho_{\SCH} = G^{-2}_{\SCH}\LPROD{\tpsi}{\tpsi} & \mbox{for the Schr\"odinger equation,}\\
   &\rho_{\BO}  = G_{\BO}^{-2}                      & \mbox{for the Born-Oppenheimer dynamics.}
\end{align}
From the stability of the Hamilton-Jacobi equation
for $\log(|G|^{-2})$  and the estimate $\|\partial_{X^iX^j}(\theta-\tilde\theta)\|_{L^\infty}=\BIGO(M^{-1+\delta})$ 
in \eqref{theta_der_stab}
we have
\[
    G_{\SCH}^{-2}= G_{\BO}^{-2} +\BIGO(M^{-1 +\delta})\COMMA
\]
and Lemma~\ref{born_oppen_lemma} implies 
\begin{equation}\label{tilde_psi_m}
   \LPROD{\tpsi}{\tpsi}= 1 + \BIGO(M^{-1})\COMMA
\end{equation}
which proves
\[
   \rho_{\SCH} =\rho_{\BO} +\BIGO(M^{-1+\delta})\PERIOD
\]
\end{remark}

\section{Fourier integral WKB states including caustics}\label{sec:caustics}

\subsection{A preparatory example with the simplest caustic}
As an example of a caustic, we study first the simplest example of a fold caustic based on the Airy function
$\AIRY:\rset\rightarrow\rset$ which solves
\begin{equation}\label{eq:Airy_ODE}
  -\partial_{xx}{\AIRY}(x)  +x \AIRY(x)=0\PERIOD
\end{equation}
The scaled Airy function
\[
  u(x)= C\,\AIRY(M^{1/3} x)
\]
solves the Schr\"odinger equation
\begin{equation}\label{eq:airy_m}
-\frac{1}{M}\partial_{xx}{u}(x) + x u(x)=0\COMMA
\end{equation}
for any constant $C$. In our context an important property of the Airy function is the fact
that it is the inverse Fourier transform of the function
\[
 \hat\AIRY(p) = \sqrt{\frac{2}{\pi}}\EXP{\Iunit p^3/3}\COMMA
\]
i.e.,
\begin{equation}\label{airy_fourier}
  \AIRY(x)= \frac{1}{\pi} \int_\rset \EXP{\Iunit (xp+p^3/3)}\, dp\PERIOD
\end{equation}
In the next section, we will consider a general Schr\"odinger equation and determine
a WKB Fourier integral corresponding to \eqref{airy_fourier} for the
Airy function; as an introduction to the general case we show how the derive \eqref{airy_fourier}:  by taking the Fourier transform of the ordinary differential equation~\eqref{eq:Airy_ODE}
\begin{equation}
    0 = \int_\rset \left(-\partial_{xx}  +x\right){\AIRY}(x) \EXP{-\Iunit xp}\, dx
      = (p^2+\Iunit \partial_p) \hat\AIRY(p)\COMMA %
\end{equation}
we obtain an ordinary differential equation for the Fourier transform $ \hat\AIRY(p)$ %
with the solution $ \HATAIRY(p)= C \EXP{\Iunit p^3}$, for any constant $C$.
Then, by differentiation, it is clear that the scaled Airy function $u$ solves~\eqref{eq:airy_m}.
Furthermore, the stationary phase method, cf. Section~\ref{stat_phase_sec}, shows that to
the leading order $u$ is approximated by 
$$
u(x)\simeq C\left(-x M^{1/3}\right)^{-1/4} \cos \big(M^{1/2} (-x)^{3/2} -\pi/4\big)\COMMA\;\;
  \mbox{ for } x < 0\COMMA
$$ 
and $u(x)\simeq 0$ to any order (i.e., $\BIGO(M^{-K})$ for any positive $K$) when $x>0$. 
The behaviour of the Airy function is illustrated in Figure~\ref{airy_fig}.

\begin{figure}[htbp]
\includegraphics[height=10cm]{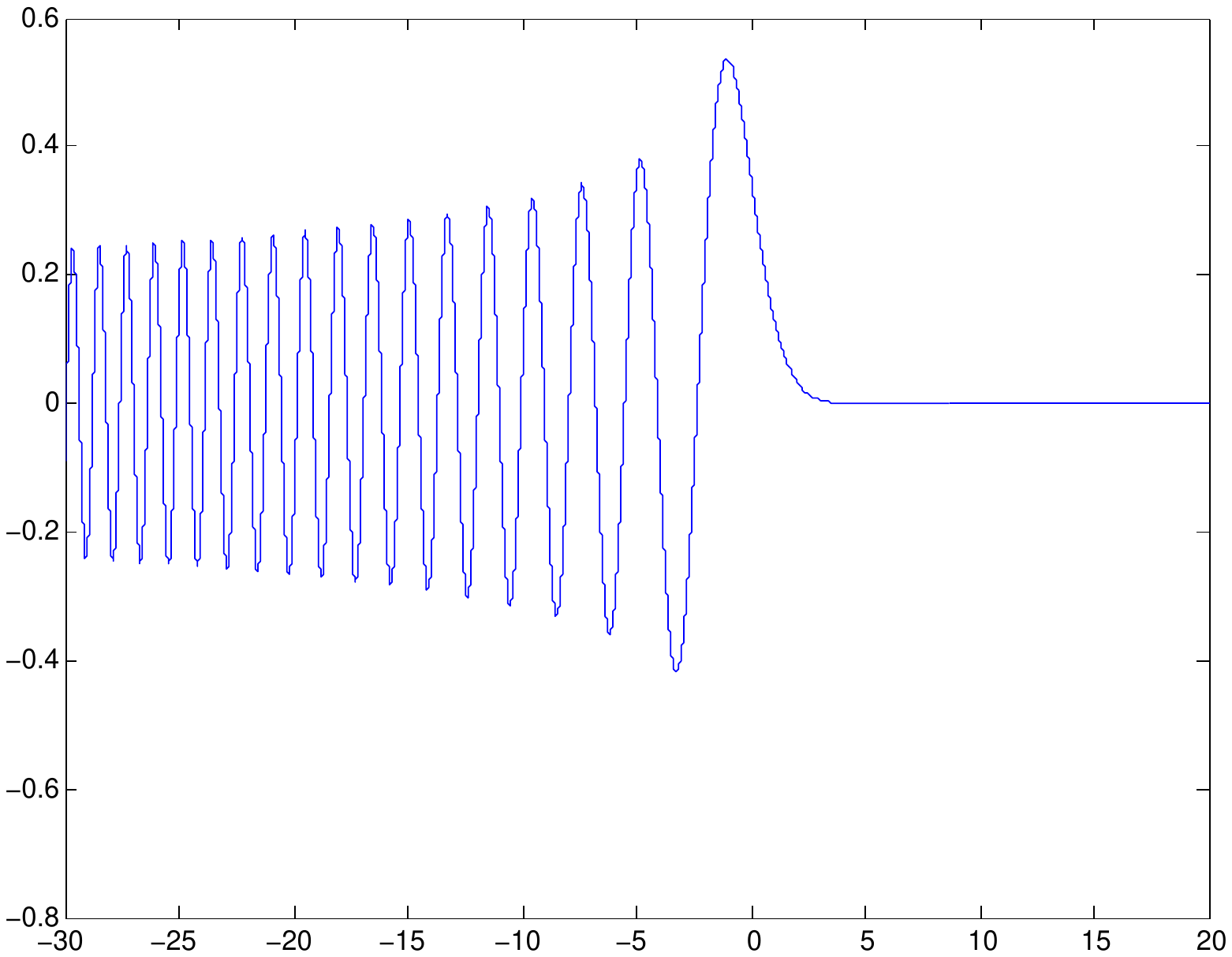}
\caption{The Airy function.}
\label{airy_fig}
\end{figure}

\subsubsection{Molecular dynamics for the Airy function}
The eikonal equation  corresponding to \eqref{eq:airy_m} is 
\[
  p^2+x=0
\]
with solutions for $x\le 0$, which leads to the phase
\begin{equation}\label{legendre_airy}
  p=\theta'(x)=\pm (-x)^{1/2}\COMMA\;\mbox{and}\;\;\theta(x)= \mp \frac{2}{3} (-x)^{3/2}\PERIOD
\end{equation}
We compute the Legendre transform
\[
  \LFT{\theta}(p)= x p -\theta(x)
\]
where by \eqref{legendre_airy} and  $-x=p^2$ we obtain
\[
   \LFT{\theta}(p)=- p^2 p +\frac{2}{3} p^3 = -\frac{p^3}{3}\PERIOD
\]
We note that this solution is also obtained from the eikonal equation
\[
  p^2 + \partial_p\LFT{\theta}(p)=0\COMMA
\]
which is solved by 
\[
   \LFT{\theta}(p)=-p^3/3\PERIOD
\]
Thus we recover the relation for the Legendre transform 
$-xp +\LFT{\theta}(p)=-\theta(x)$.

\subsubsection{Observables for the Airy function}\label{airy_obs} 
The primary object of our analysis is an observable (a functional depending on $u$)
rather than the solution $u(x)$ itself. Thus we first compute the observable evaluated on
the solution obtained from the Airy function. In the following calculation we denote 
by $C$ a generic constant not necessarily the same at each occurrence, 
\begin{equation}\label{g_A_m}
\begin{split}
   \int_\rset g(x) |u(x)|^2 dx 
          &= C \int_\rset g(x) \int_\rset\int_\rset \EXP{-\Iunit M^{1/2}(xp+p^3/3)} \EXP{\Iunit M^{1/2}(xq+q^3/3)}\,dq\, dp\, dx\\
          &= C \int_\rset\int_\rset \hat g\left(M^{1/2}(p-q)\right)\EXP{\Iunit M^{1/2}(q^3/3-p^3/3)}\, dq\, dp \\
          &= C \int_\rset\int_\rset \hat g\left(M^{1/2}(p-q)\right)
                \EXP{\Iunit M^{1/2}\left((q-p)^3/12+(q-p)(p+q)^2/4\right)}\, dq\, dp\\
          &= C \int_\rset\int_\rset \hat g(\underbrace{-M^{1/2}\bar q}_{=t})\,
                \EXP{\Iunit M^{1/2}\left(\overline q^3/12+\overline q\, \overline{p}^2/4\right)} 
                \overbrace{d\overline q d\overline p}^{\overline q = q-p, \, \overline p = p+q}  \\
          &= C \int_\rset\int_\rset \hat g(t) \EXP{-\Iunit \left( t^3/(12M)+t\, \overline{p}^2/4\right)}\, dt\, d\overline p\\
          &= C \int_\rset g * \AIRY_M( \underbrace{-\overline p^2}_{=\partial_p\LFT\theta(\overline p)=x})\, d \overline p\\
          &= C \int_{-\infty}^0 g * \AIRY_M (x) |\partial_x p(x)|\, dx\COMMA
\end{split}
\end{equation}
where 
\begin{equation}\label{airy_fourier_m}
  \AIRY_M(x):= \left(\frac{M}{4}\right)^{1/3} \AIRY\left(\left(\frac{M}{4}\right)^{1/3}x\right)
\ \mbox{is the Fourier transform of $\EXP{-\Iunit t^3/(12M)}$.}
\end{equation}
\begin{lemma}\label{lem_airy}
  The scaled Airy function $\AIRY_M$ is an approximate identity in the following sense
  \begin{equation}\label{airy_accur}
      \| g * \AIRY_M -g \|_{L^2(\rset)}\le \frac{1}{12M}\| \partial^3_x g\|_{L^2(\rset)}\PERIOD
  \end{equation}
\end{lemma}
\begin{proof}
Plancherel's Theorem implies 
\[
\begin{split}
   M\| g* \AIRY_M -g \|_{L^2} &= M\|\hat g\HATAIRY_M -\hat g\|_{L^2} =\|\hat g(\EXP{\Iunit p^3/(12M)} -1)M\|_{L^2}\\
            &\le \frac{1}{12}\||p|^3\hat g\|_{L^2} =\frac{1}{12}\|\partial^3_x g \|_{L^2}\PERIOD
\end{split}
\]
The inequality follows from $|\EXP{\Iunit y}-1|\le |y|$ which holds for all $y\in\rset$.
\end{proof}

The classical molecular dynamics approximation corresponding to the Schr\"odinger
equation \eqref{eq:airy_m} is the Hamiltonian system 
\[
  \dot X =p\COMMA\;\;\;\;\dot p =-\frac{1}{2}
\]
with a solution $X_t=-t^2/4$ and the corresponding approximation of the observable 
\[
   \frac{1}{T} \int_0^T g(X_t)\, dt = \frac{1}{T} \int_0^T g(X_t) \,\frac{d X_t}{\dot X_t}
        =\frac{1}{T} \int^{0}_{-T^2/4} g(x)\, \frac{dx}{|p(x)|}\PERIOD
\]
In this specific case the phase satisfies $|p(x)|=|x|^{1/2}$ and $|\partial_x p|=|x|^{-1/2}/2$, and hence 
the non-normalized density $|p|^{-1}$ is in this case equal to $2|\partial_x p|$. Equation \eqref{g_A_m} and Lemma \ref{lem_airy}
imply 
\[
|\int_\rset g |u|^2\, dx-\int_\rset g \partial_x p(x)\, dx|=\mathcal O(M^{-1})
\]
and consequently
for two different observables $g_1$ and $g_2$ we have that Schr\"odinger observables
are approximated by the classical observables with the error $\BIGO(M^{-1})$
\begin{equation}\label{A_m}
  \frac{\int_\rset g_1 |u|^2\, dx}{\int_\rset g_2 |u|^2\, dx} 
- \frac{\int_\rset g_1 |\partial_{xx}\theta|\, dx}{ \int_\rset g_2 |\partial_{xx}\theta|\, dx}= \BIGO(M^{-1})\COMMA
\end{equation}
using $\partial_x p(x)= \partial_{xx}\theta(x)$.
The reason we compare two different observables
with a compact support
is that $\int_\rset u^2(x)\,dx=\infty$ in the case of the Airy function. %

We note that in \eqref{g_A_m}  we used
\[
  \frac{1}{3}(q^3-p^3) = \LFT{\theta}(p)-\LFT{\theta}(q)
               =(p-q)\partial_p\LFT{\theta}\left(\frac{1}{2}(p+q)\right) 
               +\frac{1}{3}{\partial^3\LFT{\theta}}\left(\frac{1}{2}(p+q)\right)\left(\frac{1}{2}(p-q)\right)^3
\]
which in the next section is  generalized to other caustics. For the Airy function there holds
\[
  \frac{1}{3}{\partial^3\LFT{\theta}\left(\frac{1}{2}(p+q)\right)} = -\frac{2}{3}\PERIOD
\]

\subsection{A general Fourier integral ansatz}\label{sec:caustic_general}
In order to treat a more general case with a caustic of the dimension $d$ we use the Fourier integral ansatz
\begin{equation}\label{caustic_ansats}
  \Phi(X,x)= \int_{\rset^d} \hpsi(X,x) 
             \EXP{-\Iunit M^{1/2}\PHASE}\, d\CHECKP
\end{equation}
and we write
\[
\begin{split}
  X &= (\HATX, \CHECKX)\COMMA\;\;\; P = (\HATP,\CHECKP)\\
  \EPROD{\CHECKX}{\CHECKP} &=\sum_{j=1}^d\CHECKX^j\CHECKP^j\COMMA\;\;\;
        \EPROD{\HATX}{\HATP}      =\sum_{j=d+1}^N\HATX^j\HATP^j\\
  \PHASE & = \EPROD{\CHECKX}{\CHECKP} -\LFT{\theta}(\HATX,\CHECKP)\COMMA
\end{split}
\]
based on the Legendre transform 
\[
  \LFT{\theta}(\HATX,\CHECKP)= \min_{\CHECKX}\left(\EPROD{\CHECKX}{\CHECKP} - \theta(\HATX, \CHECKX)\right)\PERIOD
\]
If the function $\LFT{\theta}(\HATX,\CHECKP)$ is not defined for all $\CHECKP\in\rset^d$,
but only for $\CHECKP\in \NEIGH\subset \rset^d$ we replace the integral over $\rset^d$ by integration over $\NEIGH$
using a smooth cut-off function $\chi(\CHECKP)$.
The cut-off function is zero outside $\NEIGH$ and equal to one in a large part
of the interior of $\NEIGH$, see Section~\ref{p_comp}.
The  ansatz \eqref{caustic_ansats} is inspired by Maslov's work \REF{maslov}, although
it is not the same since our amplitude function $\hpsi$
depends on $(\HATX,\CHECKX,x)$ but not on $\CHECKP$.
We emphasize that our  modification consisting in having an amplitude function that is not dependent on $\CHECKP$
is essential in the construction of  the solution and for determining  the accuracy of observables based on this solution.

\subsubsection{Making the ansatz for a Schr\"odinger solution}
In this section we construct a solution to the Schr\"odinger equation from the ansatz \eqref{caustic_ansats}.
The constructed solution will be an {\it actual} solution and not only an asymptotic solution as in 
\REF{maslov}.
We consider first the case when the integration is over $\rset^d$ and
then conclude in the end that the cut-off function $\chi(\CHECKP)$ can be included in all
integrals without changing the property of the Fourier integral ansatz being a solution
in the $\CHECKX$-domain  where $\CHECKX=\GRAD_{\CHECKP}\LFT{\theta}(\HATX,\CHECKP)$ for some $\CHECKP$ satisfying
$\chi(\CHECKP)=1$.

The requirement to be a solution means that there should hold
\begin{equation}\label{H_ekv_caustic}
\begin{split}
0 &=(\HOPER-E)\Phi \\
  &=\int_{\rset^d} \left(\frac{1}{2}|\nabla_{\HATX}\LFT{\theta}(\HATX,\CHECKP)|^2 
       +\frac{1}{2}|\CHECKP|^2 +V_0(X) -E\right)\hpsi(X,x) \EXP{-\Iunit M^{1/2}\PHASE}\, d\CHECKP\\
  &\quad {}-\int_{\rset^d}\left(\Iunit M^{-1/2}(\EPROD{\nabla_{\HATX}\hpsi}{\nabla_{\HATX}\LFT{\theta}} 
                - \EPROD{\nabla_{\CHECKX}\hpsi}{\CHECKP}+\frac{1}{2} \hpsi\LAP_{\HATX}\LFT{\theta})
                -(\VOPER-V_0)\hpsi +\frac{1}{2M}\LAP_{X}\hpsi\right)
            \EXP{-\Iunit M^{1/2}\PHASE}\, d\CHECKP\PERIOD
\end{split}
\end{equation}
Comparing this expression to the previously discussed case of a single WKB-mode we see that the zero order term is
now $\LAP_{\HATX}\LFT{\theta}$ instead of $\LAP_{X}\theta$
and that we have $-\EPROD{\nabla_{\CHECKX}\hpsi}{\CHECKP}$ instead of 
$\EPROD{\nabla_{\CHECKX}\hpsi}{\nabla_{\CHECKX}\theta}$.
However, the main difference is that the first integral is 
not zero (only the leading order term of its
stationary phase expansion is zero, cf. \eqref{caustic_expansion}). Therefore, the first integral contributes to the second integral. The
goal is now to determine a function $\SFUN(\HATX,\CHECKX, \CHECKP)$ satisfying
\begin{equation}\label{v-v_int}
\begin{split}
   \int_{\rset^d} & \left(\frac{1}{2}|\nabla_{\HATX}\LFT{\theta}|^2 +\frac{1}{2}|\CHECKP|^2
        +V_0(X) -E\right)\EXP{-\Iunit M^{1/2}\PHASE}\, d\CHECKP\\
        & =\Iunit M^{-1/2}\int_{\rset^d}  \SFUN(\HATX,\CHECKX, \CHECKP)\,\EXP{-\Iunit M^{1/2}\PHASE}\, d\CHECKP\COMMA
\end{split}
\end{equation}
and verify that it is bounded.
\begin{lemma}\label{F_lem}
There holds
$F=F_0+F_1$
where
\[
\begin{split}
F_0 &=\frac{1}{2}\sum_{i,j}  \partial_{\CHECKX^i\CHECKX^j } V_0\left(\DPLFT{\theta}(\CHECKP)\right) 
                    \partial_{\CHECKP^j\CHECKP^i}\LFT{\theta}(\CHECKP)\COMMA \\
F_1 &= \Iunit M^{-1/2} \int_0^1\int_0^1\int_{\rset^d}  \sum_{i,j,k} t(1-t)
     \partial_{\CHECKP^k} \left[\partial_{\CHECKX^i\CHECKX^j\CHECKX^k} 
          V_0\left(\DPLFT{\theta}(\CHECKP)  + s\, t\,\DTHETA\right)
           \partial_{\CHECKP_j\CHECKP_i}\DPLFT{\theta}(\CHECKP)\right]\,dt\,ds\PERIOD\\
\end{split}
\]
\end{lemma}
\begin{proof}
 The function
$\LFT{\theta}(\HATX,\CHECKP)$ is defined as a solution to
the Hamilton-Jacobi (eikonal) equation
\begin{equation}\label{caustic_hj}
    \frac{1}{2}|\nabla_{\HATX}\LFT{\theta}(\HATX,\CHECKP)|^2 +\frac{1}{2}|\CHECKP|^2 
       + V_0\left(\HATX,\nabla_{\CHECKP}\LFT{\theta} (\HATX,\CHECKP)\right) - E = 0
\end{equation}
for all $(\HATX,\CHECKP)$.
Consequently, the integral on the left hand side of \eqref{v-v_int} is 
\[
    \int_{\rset^d} \left(V_0(\HATX,\CHECKX) -V_0(\HATX,\nabla_{\CHECKP}\LFT{\theta}(\HATX,\CHECKP)\right)
    \EXP{-\Iunit M^{1/2}\left(\EPROD{\CHECKX}{\CHECKP} -\LFT{\theta}(\HATX,\CHECKP)\right)}\, d\CHECKP\PERIOD
\]
Let  $\CHECKP_0(\CHECKX)$ be any solution to the stationary phase equation
$\CHECKX=\nabla_{\CHECKP}\LFT{\theta}(\HATX,\CHECKP_0)$
and introduce the notation 
\[
 \PHASED  :=\EPROD{\GRAD_{\CHECKP}\LFT{\theta}(\HATX,\CHECKP_0)}{\CHECKP} 
           -\LFT{\theta}(\HATX,\CHECKP)\PERIOD
\]
Then by writing a difference as $V(y_1)-V(y_2)=\int_0^1\partial_y V(y_2+t(y_1-y_2))dt\cdot(y_1-y_2)$, 
identifying a derivative $\partial_{\CHECKP_i} $ and integrating by parts
the integral can be written
\[
\begin{split}
    &\int_{\rset^d} \left(V_0(\HATX,\nabla_{\CHECKP}\LFT{\theta}(\HATX,\CHECKP_0))
           -V_0(\HATX,\nabla_{\CHECKP}\LFT{\theta}(\HATX,\CHECKP)\right)
             \EXP{-\Iunit M^{1/2}\PHASED}\, d\CHECKP\\
    &=  \int_0^1\int_{\rset^d}  \sum_i\partial_{\CHECKX^i} V_0\left(\DPLFT{\theta}(\CHECKP) 
                    +t\,\left[\DPLFT{\theta}(\CHECKP_0)-\DPLFT{\theta}(\CHECKP)\right]\right)\times\\
    & \qquad\times        \left(\partial_{\CHECKP^i}\LFT{\theta}(\CHECKP_0)-\partial_{\CHECKP^i}\LFT{\theta}(\CHECKP)\right) 
        \EXP{-\Iunit M^{1/2}\PHASED}\,d\CHECKP\, dt \\ 
    &= -\Iunit M^{-1/2} \int_0^1\int_{\rset^d}  \sum_i\partial_{\CHECKX^i} V_0
           \left(\DPLFT{\theta}(\CHECKP) +t\,\left[\DPLFT{\theta}(\CHECKP_0)-\DPLFT{\theta}(\CHECKP)\right]\right)
           \partial_{\CHECKP_i} \EXP{-\Iunit M^{1/2}\PHASED} \,d\CHECKP\, dt \\ 
   &= \Iunit M^{-1/2} \int_0^1\int_{\rset^d}  \sum_i\partial_{\CHECKP_i} \partial_{\CHECKX^i} 
                V_0\left(\DPLFT{\theta}(\CHECKP) +t\,\left[\DPLFT{\theta}(\CHECKP_0)-\DPLFT{\theta}(\CHECKP)\right]\right)
            \EXP{-\Iunit M^{1/2}\PHASED}\,d\CHECKP\, dt\PERIOD
\end{split}
\]
Therefore the  leading order term  in $\SFUN=: \SFUN_0 + \SFUN_1$ is
\[
\begin{split}
  \SFUN_0&:=\int_0^1\sum_{i,j}  (1-t) \partial_{\CHECKX^i\CHECKX^j } V_0\left(\DPLFT{\theta}(\CHECKP)\right)
                                 \partial_{\CHECKP^j\CHECKP^i}\LFT{\theta}(\CHECKP)\, dt\\
     &=\frac{1}{2}\sum_{i,j}  \partial_{\CHECKX^i\CHECKX^j } V_0\left(\DPLFT{\theta}(\CHECKP)\right) 
                    \partial_{\CHECKP^j\CHECKP^i}\LFT{\theta}(\CHECKP)\PERIOD
\end{split}
\]
Denoting $\DTHETA = \DPLFT{\theta}(\CHECKP_0)-\DPLFT{\theta}(\CHECKP)$  the remainder becomes
\[
\begin{split}
    & -\Iunit M^{-1/2} \int_0^1\int_{\rset^d}  \sum_{i,j} \left[
           \partial_{\CHECKX^i\CHECKX^j} V_0\left(\DPLFT{\theta}(\CHECKP) \right)
         - \partial_{\CHECKX^i\CHECKX^j} V_0\left(\DPLFT{\theta}(\CHECKP)  + t\,\DTHETA\right) \right]\\
    &\qquad \times (1-t)\partial_{\CHECKP^j\CHECKP^i}\LFT{\theta}(\CHECKP) \,\EXP{-\Iunit M^{1/2}\PHASED} \,d\CHECKP\, dt\\
    &= \Iunit M^{-1/2} \int_0^1\int_0^1\int_{\rset^d}  \sum_{i,j,k} t(1-t)
          \partial_{\CHECKX^i\CHECKX^j\CHECKX^k} V_0\left(\DPLFT{\theta}(\CHECKP)  +  s\, t\,\DTHETA\right) 
           \partial_{\CHECKP^j\CHECKP^i}\LFT{\theta}(\CHECKP) \\
    &\qquad \times \left(\partial_{\CHECKP^k}\LFT{\theta}(\CHECKP_0)-\partial_{\CHECKP^k}\LFT{\theta}(\CHECKP)\right)
      \EXP{-\Iunit M^{1/2}\PHASED} \,d\CHECKP\,dt\, ds\\
    &= - \frac{1}{M} \int_0^1\int_0^1\int_{\rset^d}  \sum_{i,j,k} t(1-t)
           \partial_{\CHECKP^k} \left[\partial_{\CHECKX^i\CHECKX^j\CHECKX^k} 
             V_0\left(\DPLFT{\theta}(\CHECKP)  + s\, t\, \DTHETA\right)   
             \partial_{\CHECKP^j\CHECKP^i}\LFT{\theta}(\CHECKP)\right] \\
    &\qquad\times \EXP{-\Iunit M^{1/2}\PHASED}\,d\CHECKP\,dt\,ds\COMMA
\end{split}
\]
hence the function $\SFUN_1$ is purely imaginary and  small
\[
\begin{split}
  &\SFUN_1= \Iunit M^{-1/2} \int_0^1\int_0^1\int_{\rset^d}  \sum_{i,j,k} t(1-t)
     \partial_{\CHECKP^k} \left[\partial_{\CHECKX^i\CHECKX^j\CHECKX^k} 
          V_0\left(\DPLFT{\theta}(\CHECKP)  + s\, t\,\DTHETA\right)
           \partial_{\CHECKP_j\CHECKP_i}\DPLFT{\theta}(\CHECKP)\right]\,dt\,ds\COMMA
\end{split}
\]
and
\begin{equation}\label{Re_s}
2\REAL{\SFUN}=\sum_{i,j}  \partial_{\CHECKX^i\CHECKX^j } V_0\left(\DPLFT{\theta}(\CHECKP)\right) 
            \partial_{\CHECKP^j\CHECKP^i}\LFT{\theta}(\CHECKP)\PERIOD
\end{equation}
\end{proof}

The eikonal equation \eqref{caustic_hj} and the requirement that
$(\HOPER-E)\Phi =0$ in \eqref{H_ekv_caustic} then imply that
\begin{equation}\label{before_hj}
\begin{split}
   0&= \int_{\rset^d}\left[\Iunit M^{-1/2}\left(\EPROD{\nabla_{\HATX}\hpsi}{\nabla_{\HATX}\LFT{\theta}} 
         - \EPROD{\nabla_{\CHECKX}\hpsi}{\CHECKP}
         +\frac{1}{2} \hpsi \left( \LAP_{\HATX}\LFT{\theta} - 2 \SFUN(X,\CHECKP) \right)\right)\right.\\
    &\qquad \left. -(\VOPER-V_0)\hpsi +\frac{1}{2M}\LAP_{X}\hpsi\right]
            \EXP{-\Iunit M^{1/2}\PHASE}\, d\CHECKP\PERIOD
\end{split}
\end{equation}
The Hamilton-Jacobi eikonal equation \eqref{caustic_hj}, in the primal variable $(\HATX,\CHECKP)$
with the corresponding dual variable $\hat P, \check X)$,
can be 
solved by  the characteristics
\begin{equation}\label{caustic_char}
\begin{split}
   \dot{\HATX} & = \HATP\\
   \dot{\HATP} &= -\GRAD_{\HATX} V_0(\HATX,\CHECKX)\\
   \dot{\CHECKX} & = -\CHECKP\\
   \dot{\CHECKP} &= \GRAD_{\CHECKX} V_0(\HATX,\CHECKX)\COMMA
\end{split}
\end{equation}
using the definition
\[
\begin{split}
    \GRAD_{\HATX}\LFT{\theta}(\HATX,\CHECKP)   &= \HATP         \\
    \GRAD_{\CHECKP}\LFT{\theta}(\HATX,\CHECKP) &= \CHECKX\PERIOD
\end{split}
\]
The characteristics give 
\[
   \frac{d}{dt} \hpsi = \EPROD{\GRAD_{\HATX}\hpsi}{\GRAD_{\HATX}\LFT{\theta}} - 
                       \EPROD{\GRAD_{\CHECKX}\hpsi}{\CHECKP}\COMMA
\]
so that the Schr\"odinger transport equation becomes, as in  \eqref{schrod_first},
\begin{equation}\label{caustic_transport}
    \Iunit M^{-1/2} \left(\dot{\hpsi} +  \hpsi \frac{\dot G}{G}\right) 
                = (\VOPER - V_0) \hpsi - \frac{1}{2M} \LAP_{X}\hpsi
\end{equation}
and for $\tpsi=G\phi$
\begin{equation}\label{caustic_transport_psi}
    \Iunit M^{-1/2} \dot{\tpsi}= (\VOPER - V_0) \tpsi - \frac{G}{2M} \LAP_{X}\frac{\tpsi}{G}
\end{equation}
with  the complex valued weight function $G$ defined by
\begin{equation}\label{g_2}
  \frac{d}{dt} \log G_t= \frac{1}{2}\LAP_{\HATX}\LFT{\theta}(\hat X_t,\CHECKP_t) - \SFUN(\hat X_t,\CHECKP_t)\PERIOD
\end{equation}
This transport equation is of the same form as the transport equation  for a single WKB-mode, with
a modification of the weight function $G$.

Differentiation of the second equation in the 
Hamiltonian system \eqref{caustic_char} implies that the first variation 
$\partial \CHECKP_t/\partial \CHECKX_0$
satisfies
\[
  \frac{d}{dt}\left(\frac{\partial\CHECKP_t^i}{\partial \CHECKX_0}\right) = 
      \sum_{j,k}\partial_{\CHECKX^i\CHECKX^j} V_0(\HATX,\CHECKX_t) 
           \partial_{\CHECKP^j\CHECKP^k}\LFT{\theta}(\CHECKP) \frac{\partial\CHECKP_t^k}{\partial \CHECKX_0} \COMMA
\]
which by the Liouville formula \eqref{liouville_form} and the equality 
\[
2\REAL{\SFUN}=\sum_{i,j} \partial_{\check X^i\check X^j}V_0\partial_{\CHECKP^j\CHECKP^i}\LFT{\theta}
=\TRACE(\sum_j\partial_{\CHECKX^i\CHECKX^j}V_0\partial_{\CHECKP^j\CHECKP^k}\LFT{\theta})\]
in \eqref{Re_s} yields the relation, 
\begin{equation}\label{re_s}
   \EXP{-2\int_0^t\REAL{\SFUN}\, dt'} = C\, \left|\DET \frac{\partial \CHECKP_t}{\partial \CHECKX_0}\right|\COMMA
\end{equation}
for the constant $C:=|\DET \frac{\partial \CHECKX_0}{\partial \CHECKP_0}|$.  We use relation
\eqref{re_s} %
to study the density in the next section.

\begin{remark}\label{chi_rem} 
  The conclusion in this section holds also when all integrals over $d\CHECKP$ in $\rset^d$
  are replaced by integrals with the measure $\chi(\CHECKP)\, d\CHECKP$. Then there holds
  $2\REAL{\SFUN}=\sum_{ij} \partial_{\CHECKX^i\CHECKX^j}\VOPER \partial_{\CHECKP^i}(\chi \partial_{\CHECKP^j}\LFT{\theta})$. We  use
  that the observable $g$ is zero when the cut-off function $\chi_j$ is not one, see Section~\ref{p_comp}. 
  In Section~\ref{global_sec} we show how to construct a global solution
  by connecting the Fourier integral solutions, valid in a neighborhood where 
  $\DET\partial(X)/\partial(P)$ vanishes
  (and $\chi(\CHECKP)=1$), 
  to a sum of WKB-modes, valid in neighborhoods where $\DET\partial(P)/\partial(X)$ vanishes
  (and $\chi(\CHECKP)<1$).
\end{remark}

\subsubsection{The Schr\"odinger density for caustics.} 
In this section we study the density generated by the solution
\[
  \Phi(X,x)= \int_{\rset^d} \hpsi(X,x)\, \EXP{-\Iunit M^{1/2}\left(
        \EPROD{\CHECKX}{\CHECKP} -\LFT{\theta}(\HATX,\CHECKP)\right)}\,d\CHECKP\PERIOD
\]
The analysis of the density generalizes the calculations for the Airy function in Section~\ref{airy_obs}.
We have, using the notation $\hat{\TLDG}$ for the Fourier transform of $\TLDG$ with respect to the $\CHECKX$ variable,
and by introducing the notation $\CHECKR = \tfrac{1}{2}(\CHECKP + \CHECKQ)$ and $\CHECKS = \CHECKP - \CHECKQ$ 

\begin{equation}\label{caustic_density}
\begin{split}
 \int g(X)|\Phi(x,X)|^2 \, dxdX  &= \int \underbrace{g(X) \LPROD{\hpsi}{\hpsi}}_{=:\TLDG(X)}\, 
      \EXP{\Iunit M^{1/2}(\EPROD{\CHECKX}{\CHECKP} -\LFT{\theta}(\HATX,\CHECKP))}\,
      \EXP{-\Iunit M^{1/2}(\EPROD{\CHECKX}{\CHECKQ} -\LFT{\theta}(\HATX,\CHECKQ))}\, d\CHECKP\, d\CHECKQ\, dX\\
    &=\int \hat{\TLDG}(\HATX, M^{1/2}\CHECKS)\,  
      \EXP{\Iunit M^{1/2}(\LFT{\theta}(\HATX,\CHECKQ) -\LFT{\theta}(\HATX,\CHECKP))}\, d\CHECKP\, d\CHECKQ\, d\HATX\\
    &=\int \hat{\TLDG}(\HATX, M^{1/2}\CHECKS)\,  
      \EXP{\Iunit M^{1/2}\frac{1}{6}(\EPROD{\CHECKS}{\nabla_{\CHECKP})^3 \LFT{\theta}(\HATX,\CHECKR} 
            +\gamma \CHECKS/2)} \times\\
    &\qquad\times\EXP{\Iunit M^{1/2}\EPROD{\CHECKS}{\nabla_{\CHECKP}\LFT{\theta}(\HATX,\CHECKR)}}\, 
                 d\CHECKS\, d\CHECKR\ d\HATX \\ %
    &= \left(\frac{1}{2\pi}\right)^{d/2} M^{-1/2}\int \TLDG * \AIRY_M \big(\HATX, \underbrace{\nabla_{\CHECKP} 
            \LFT{\theta} (\HATX, \CHECKR}_{=\CHECKX})\big)\, d\CHECKR d\HATX \\
    &= \left(\frac{1}{2\pi}\right)^{d/2} M^{-1/2}\int \TLDG * \AIRY_M (\HATX, \CHECKX)\, 
         \left|\DET \frac{\partial(\CHECKP)}{\partial (\CHECKX)}\right|\,dX  \PERIOD
\end{split}
\end{equation}
In the convolution $\TLDG * \AIRY_M$, the function $\AIRY_M$, analogous to \eqref{airy_fourier_m},
 is the Fourier transform of 
\[
   \EXP{\Iunit \frac{1}{M}(\EPROD{\omega}{\nabla_{\CHECKP})^3 \LFT{\theta}(\HATX,\CHECKP)}}
        \Big|_{\CHECKP= \CHECKR +\gamma \omega}
\]
with respect to $\omega\in\rset^d$ and the integration in $\CHECKX$ is with respect to the range of 
$\nabla_{\CHECKP}\LFT{\theta}(\HATX,\cdot)$.
As a next step we evaluate the Fourier transform and its derivatives at zero and obtain
\[
\begin{split}
  & \int_{\rset^d} \AIRY_M(\CHECKX)\, d\CHECKX                  = 1\COMMA\;\;\;
       \int_{\rset^d} \CHECKX^i \AIRY_M(\CHECKX)\, d\CHECKX     = 0\COMMA      \\
  & \int_{\rset^d} \CHECKX^i\CHECKX^j \AIRY_M(\CHECKX)\, d\CHECKX    = 0\COMMA\;\;\;
       M\int_{\rset^d} \CHECKX^i \CHECKX^j \CHECKX^k  \AIRY_M(\CHECKX)\, d\CHECKX    = \BIGO(1).\\
\end{split}
\]
Here we use that both differentiation with respect to $(\EPROD{\omega}{\nabla_{\CHECKP}})^3$
and $\LFT{\theta}(\HATX,\CHECKR + \gamma\omega)$
yield factors of $\omega$ which vanish.
The vanishing moments of $\AIRY_M$ imply that
\begin{equation}\label{A_m_felet}
   \|\TLDG * \AIRY_M -\TLDG\|_{L^2(d\CHECKX)} = \BIGO(M^{-1})
\end{equation}
as in \eqref{airy_accur}, so that up to $\BIGO(M^{-1})$ error the convolution with $\AIRY_M$ can be neglected. 

\subsubsection{Integration over a compact set in $\CHECKP$}\label{p_comp}
In the case when the integration is over $\NEIGH\subset\rset^d$ instead of $\rset^d$,
we use a smooth cut-off function $\chi(\CHECKP)$, which is zero outside $\NEIGH$ and restrict our analysis
to the case when the smooth observable mapping  
$\CHECKP\mapsto  g(\HATX,\nabla_{\CHECKP}\LFT{\theta}(\HATX,\CHECKP))$ is compactly supported in
the domain where $\chi$ is one. In this way $g(\HATX,\nabla_{\CHECKP}\LFT{\theta}(\HATX,\CHECKP))$ 
is zero when $\nabla_{\CHECKP}\chi(\CHECKP)$ is non zero.
The integrand is thus equal to
\[
   (g(X)\ \LPROD{\hpsi}{\hpsi}) \chi(\CHECKP)\chi(\CHECKQ)
\]
and we use the convergent Taylor expansion
\[
   \chi(\underbrace{\CHECKR+M^{-1/2}\omega}_{\CHECKP})
   \chi(\underbrace{\CHECKR-M^{-1/2}\omega}_{\CHECKQ})= \sum_{k=0}^\infty \frac{|\omega|^{2k}}{M^k} 
    a_k(\CHECKR)\PERIOD
\]
Then the observable becomes
\[
    (2\pi)^{-d/2}M^{-1/2}\sum_{k=0}^\infty \int \left(a_k\, (M^{-1}\LAP_{\CHECKX})^k
       \TLDG\right) * \AIRY_M \left(\HATX, \nabla_{\CHECKP} \LFT{\theta}(\HATX,\CHECKR)\right)\,
       d\CHECKR\, d\HATX\PERIOD
\]
As in \eqref{A_m_felet} we can remove the convolution with $\AIRY_M$ by introducing an error $\BIGO(M^{-1})$
and since for $k>0$ we have $a_k(\CHECKR)
g(\HATX,\nabla_{\CHECKP}\LFT{\theta}(\HATX,\CHECKR)) =0$ and  $a_0=1$,
we obtain the same observable as before
\[
\begin{split}
   &\sum_{k=0}^\infty \int \left(a_k\, (M^{-1}\LAP_{\CHECKX})^k \TLDG\right) * \AIRY_M\big(\HATX,\nabla_{\CHECKP} 
               \LFT{\theta} (\HATX,\CHECKR )\big)\, d\CHECKR\, d\HATX\\
   &=\sum_{k=0}^\infty \int \left(a_k\, (M^{-1}\LAP_{\CHECKX})^k \TLDG\right) 
              \left(\HATX,\nabla_{\CHECKP} \LFT{\theta}(\HATX,\CHECKR)\right)\,
                   d\CHECKR\, d\HATX\, +\BIGO(M^{-1})\\
   &=\int \TLDG \left(\HATX, \nabla_{\CHECKP} \LFT{\theta}(\HATX,\CHECKR )\right)\, 
                   d\CHECKR\, d\HATX\, +\BIGO(M^{-1})\\
   &=\int \TLDG (\HATX, \CHECKX)\, \left|\DET\frac{\partial(\CHECKP)}{\partial (\CHECKX)}\right|\, dX\, +\BIGO(M^{-1})\PERIOD
\end{split}
\]

\subsubsection{Comparing the Schr\"odinger and molecular dynamics densities}
We compare the Schr\"odinger density to the molecular dynamics density generated by  the continuity
equation
\[
  0= \DIV (\rho \GRAD\theta) = \EPROD{\GRAD\rho}{\GRAD\theta} + \rho\, \DIV(\GRAD{\theta})
   = \dot\rho + \rho\, \DIV(\GRAD\theta)
\]
which yields the density
\[
  \EXP{-\int\DIV(\GRAD\theta)\, dt}\PERIOD
\]
We have $P=\GRAD\theta$, so that $\tfrac{\partial(P)}{\partial(X)}=\partial_{XX}\theta$.
The Liouville formula \eqref{liouville_form} implies the molecular dynamics density
\begin{equation}\label{rho_md_caustic}
 \begin{split}
   \rho_{\BO}=\EXP{-\int_0^t\DIV(\GRAD\theta) \,dt'} &= 
              \DET \frac{\partial X_{0,\BO}}{\partial X_{t,\BO}}\PERIOD
 \end{split}
\end{equation}

The observable for the Schr\"odinger equation has, by \eqref{caustic_density},  the density
\[
   (g\LPROD{\hpsi}{\hpsi})*\AIRY_M \left|\DET \frac{\partial(\CHECKP)}{\partial(\CHECKX)}\right|\PERIOD
\]
We want to compare it with the molecular dynamics density $\rho_{\BO}$. 
The convolution with $\AIRY_M$ gives an error term of the order $\BIGO(M^{-1})$, as in \eqref{airy_accur},
and  following the proof of Theorem~\ref{bo_thm} for a single WKB-state in Section~\ref{sec:wkb_analysis}  (now
based on the Hamilton-Jacobi equation \eqref{caustic_hj}, the
Schr\"odinger transport equation \eqref{caustic_transport} and the definition of the weight $G$ in \eqref{g_2}),
the amplitude function satisfies, by \eqref{caustic_transport_psi}  and \eqref{g_2} and the Born-Oppenheimer approximation Lemma~\ref{born_oppen_lemma},
\[
\LPROD{\hpsi}{\hpsi}=|G|^2\langle\tpsi,\tpsi\rangle
=\EXP{\int 2\REAL\,\SFUN - \LAP_{\HATX} \LFT{\theta}\, dt}
+ \BIGO(M^{-1}),
\]
so that by \eqref{re_s}
\begin{equation}\label{dens_caustic}
\begin{split}
       (g\LPROD{\hpsi}{\hpsi})*\AIRY_M |\DET \frac{\partial(\CHECKP)}{\partial(\CHECKX)}|
    &= (g\LPROD{\hpsi}{\hpsi})|\DET \frac{\partial(\CHECKP)}{\partial(\CHECKX)}| +\BIGO(M^{-1})\\
    &=  g\, \EXP{\int 2\REAL{\SFUN} - \LAP_{\HATX}\LFT\theta \, dt}\, |\DET \frac{\partial(\CHECKP)}{\partial(\CHECKX)}|
        +\BIGO(M^{-1})\\
    &=  g\, |\DET \frac{\partial(\CHECKX_0)}{\partial(\CHECKP)}|\, |\DET \frac{\partial(\HATX_0)}{\partial(\HATX)}|\,
          |\DET \frac{\partial(\CHECKP)}{\partial(\CHECKX)}| + \BIGO(M^{-1})\COMMA \\
    &=  g\, |\DET \frac{\partial(\CHECKX_0)}{\partial(\CHECKX)}|\, |\DET \frac{\partial(\HATX_0)}{\partial(\HATX)}|
        +\BIGO(M^{-1})\COMMA \\
    &=  g\, |\DET \frac{\partial( X_0)}{\partial( X)}| + \BIGO(M^{-1})\PERIOD
\end{split}
\end{equation}
When we restrict the domain to $\NEIGH$ with the cut-off function $\chi$ as in Remark~\ref{chi_rem} we 
use the fact that $g(\HATX,\nabla_{\CHECKP}\LFT{\theta}(\HATX,\CHECKP))$ 
is zero when $\nabla_{\CHECKP}\chi(\CHECKP)$ is non zero and obtain the same.
The representations \eqref{dens_caustic} and \eqref{rho_md_caustic}
show that the density generated in the caustic case with a Fourier integral also takes
the same form, to the leading order, as the molecular dynamics density and the remaining discrepancy  
is only due to $\LFT{\theta}=\LFT\theta_{\SCH}$ and $\LFT\theta=\LFT\theta_{\BO}$ being different.
This difference is, as in the single mode WKB expansion, of size $\BIGO(M^{-1})$ which is 
estimated by the difference in Hamiltonians
of the Schr\"odinger and molecular dynamics eikonal equations.  
The estimate of the  difference of the
phase functions uses the Hamilton-Jacobi equation \eqref{caustic_hj} for $\LFT{\theta}_{\SCH}(\HATX,\CHECKP)$
and a similar Hamilton-Jacobi equation for $\LFT{\theta}_{\BO}(\HATX,\CHECKP)$ 
with $V_0=\lambda_{\BO}+\BIGO(M^{-1})$ replaced by $\lambda_{\BO}$.
The difference in the weight functions $\log(|G(\HATX,\CHECKP)|^{-2})$ is estimated by the
Hamilton-Jacobi equation 
\[
   \left(\EPROD{\nabla_{\HATX}\LFT{\theta}_{\SCH}(\HATX,\CHECKP)}{\nabla_{\HATX}} 
        -\EPROD{\nabla_{\CHECKX}V_0(\HATX,\CHECKX)}{\nabla_{\CHECKP}}\right)
         \log|G_{\SCH}(\HATX,\CHECKP)|^{-2} - \LAP_{\HATX}\LFT{\theta}_{\SCH}(X,\CHECKP) + \REAL{\SFUN}(X,\CHECKP)=0\COMMA
\]
 where $\REAL{\SFUN}$ is given in \eqref{Re_s}, and by the similar Hamilton-Jacobi equation
 with $V_0=\lambda_{\BO}+\BIGO(M^{-1})$ replaced by $\lambda_{\BO}$
 and $\LFT{\theta}_{\SCH}$ by $\LFT{\theta}_{\BO}$.

\subsubsection{A global construction coupling caustics with single WKB-modes}\label{global_sec}

We use a Hamiltonian system to construct solutions to the Schr\"odinger equation. Given a set of initial points $X_0\in \rset^{3N}$ 
the solution paths $\{(X_t,P_t)\in \rset^{6N}\SEP 0\le t< \infty,\,  H(P_0,X_0)=E\}$ of the Hamiltonian system
\[
  \begin{split}
    \dot X_t  &= \GRADP H(P_t,X_t)\\
    \dot P_t &= -\GRADX H(P_t,X_t)
  \end{split}
\]
with a smooth and bounded Hamiltonian $H(P,X)$ generate 
a $3N$-dimensional manifold called Lagrangian manifold.
The Lagrangian manifold defined by the tube of trajectories is defined by the phase
function $\theta(X)$ that plays the role of a generating function of the Lagrangian manifold.
Thus we seek a function $\theta:U\subset\R^{3N}\to \mathbb{R}$ 
such that  $P_t = \GRADX \theta(X_t)$.
We show that there exists a potential function $\theta$ 
 by determining an equation that preserves the symmetry for the matrix $Q_t$, defined as
$Q^{ij}(X):=\partial_{X^j}P^i(X)$ and $Q^{ij}_t:=Q^{ij}(X_t)$.
The relations $P^i_t=P^i(X_t)$ and $Q_t^{ij}:=\partial_{X^j}P^i(X_t)$ imply 
\[
    \dot P^i_t = \frac{d}{dt} P^i(X_t)=\sum_j \dot X^j_t \partial_{X^j}P^i_t 
    = \sum_{j}\dot{X}^j_t Q^{ij}_t = \sum_{j} \partial_{P^j} H\big(P(X_t),X_t\big)Q_t^{ij}\COMMA
\]
so that 
\[
\begin{split}
\partial_{X^k}\dot P^i_t &= \partial_{X^k}\Big( \sum_{j} \partial_{P^j} H\big(P(X_t),X_t\big)Q_t^{ij}\Big)\\
&=\underbrace{\sum_j \dot X^j_t \partial_{X^k}Q^{ij}_t}_{=\sum_j \dot X^j_t\partial_{X^kX^j}P^i_t 
=\sum_j \dot X^j_t\partial_{X^jX^k}P^i_t =\dot Q^{ik}_t} 
+\sum_j \partial_{P^jP^l} H\big(P(X_t),X_t\big)\underbrace{\partial_{X^k}P^l}_{=Q^{lk}} Q^{ij} \\
&+ \sum_j \partial_{P^jX^k} H\big(P(X_t),X_t\big) Q^{ij}_t
\end{split}
\]
and 
\[
\partial_{X^k}\dot P^i_t= -\partial_{X^k}\Big(\partial_{X^i}H\big(P(X_t),X_t\big)\Big)
= -\partial_{X^iX^k}H\big(P(X_t),X_t\big)\Big)
- \sum_j \partial_{X^iP^j} H\big(P(X_t),X_t\big)\underbrace{\partial_{X^k}P^j}_{=Q^{jk}_t}
\] 
together with the symmetry of $Q_t$ show that
%
\begin{equation}\label{q_sym}
\begin{split}
  \dot{Q}_t^{ik} &=  
  - \partial_{X^iX^k} H(P_t,X_t) -\sum_{j,l}\partial_{P^j P^l} H(P_t,X_t) Q_t^{kl}Q_t^{ij}\\
 &\qquad - \sum_j \partial_{P^jX^k} H(P_t,X_t) Q^{ij}_t
  - \sum_j \partial_{P^jX^i} H(P_t,X_t) Q^{kj}_t\PERIOD
  \end{split}
\end{equation}
Since the Hamiltonian is assumed to be smooth it follows that the right hand side in \VIZ{q_sym} is symmetric and thus the matrix $Q_t$ 
remains symmetric if it is initially symmetric. Hence there exists a potential function $\theta(X)$ such that
$P(X) = \GRADX \theta(X)$ in simple connected domains where  $Q$ is smooth. The function $Q$
may become unbounded due to the term $\partial_{P^jP^l} H \, Q^{kl} Q^{ij}$, even though $H$ has bounded third derivatives.
Points $X_t$ at which 
$\left|\TRACE(Q_t)\right|=\infty$ 
satisfy, by Liouville's theorem (see Section~\ref{liouville}),
$\left|\DET \frac{\partial X_0}{\partial X_t}\right|=\infty$ and such points are called {\it caustic} points. 

The same construction of a potential works for the local chart expressed as $X=X(P)$ instead of $P=P(X)$. 
In fact any new variable $\hat X$ (not including both $X^i$ and $P^i$ for any $i$),
based on $3N$ of the $6N$ variables $(X,P)$, and the remaining variables $3N$ variables, $\hat P$,
represent the same Hamiltonian system  with the Hamiltonian $\hat H(\hat P,\hat X):=H(P,X)$.
The Lagrangian manifold is defined by $\hat P = \GRAD_{\hat X}\hat\theta(\hat X)$ in the local chart of
$\hat P$-coordinates with the generating (potential) function $\hat\theta(\hat X)$ defined
in domains excluding caustics, i.e., where $\DET \left|\tfrac{\partial \hat X_0}{\partial \hat X_t}\right|<\infty$.
Maslov, \REF{maslov}, 
realized that a Lagrangian manifold can be partitioned, by changing coordinates in the 
neighborhood of a caustic,
into  domains where $\hat P=\GRAD_{\hat X}\hat\theta(\hat X)$ is smooth.
He used the generating (potential) functions $\hat\theta$  to construct asymptotic WKB solutions of Fourier integral type.
A sketch of this general situation is depicted in Figure~\ref{phasespacecaustics}.
In previous sections we have described global construction of solutions in a simpler case without caustics, i.e., $P_t=\GRADX\theta(X_t)$ holds everywhere.
In this section we describe the global construction of WKB solutions in the general case when  caustics are present.

\begin{figure}[ht]
 

\includegraphics[scale=0.4]{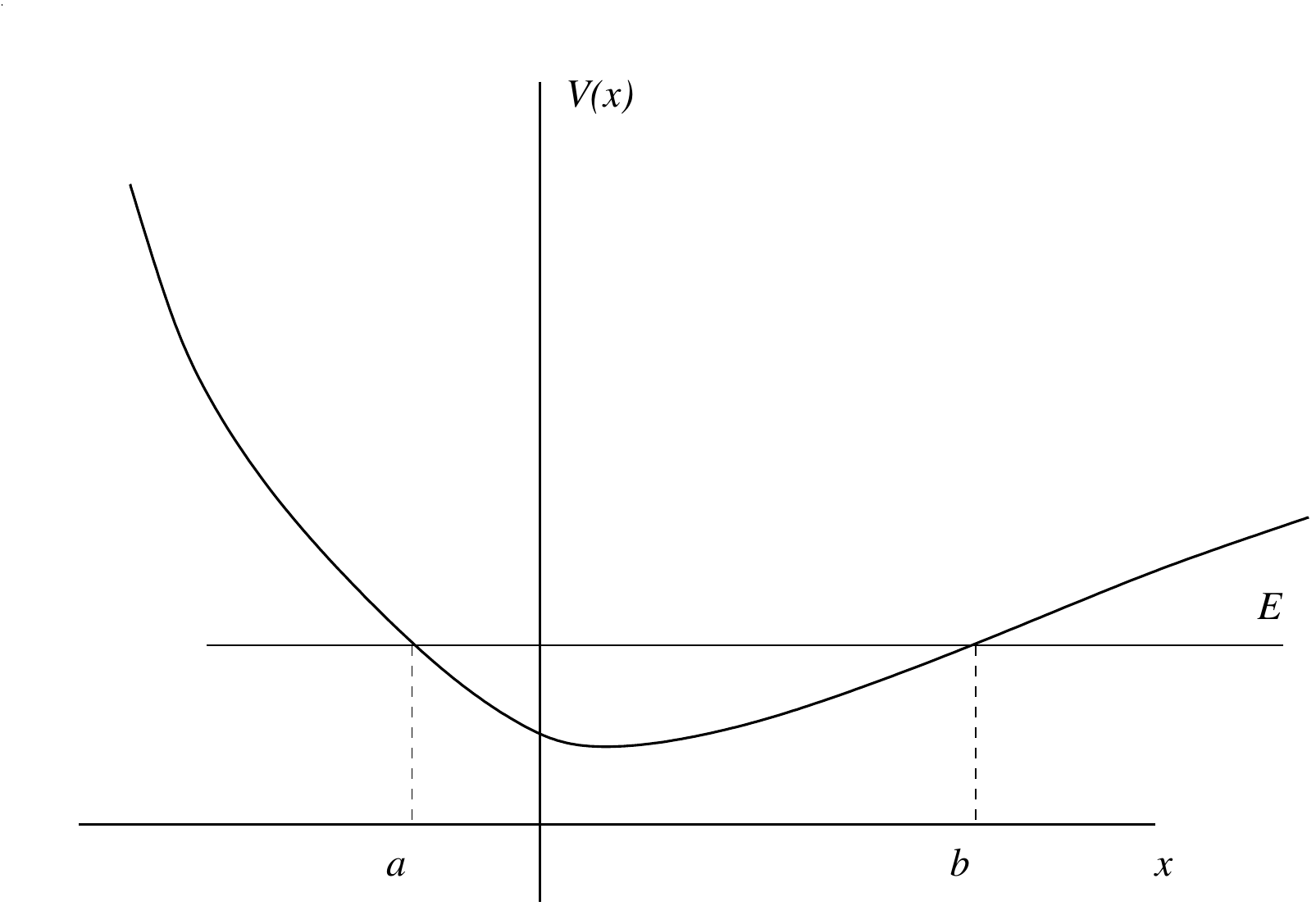}
\hspace*{0.15cm}
\includegraphics[scale=0.54]{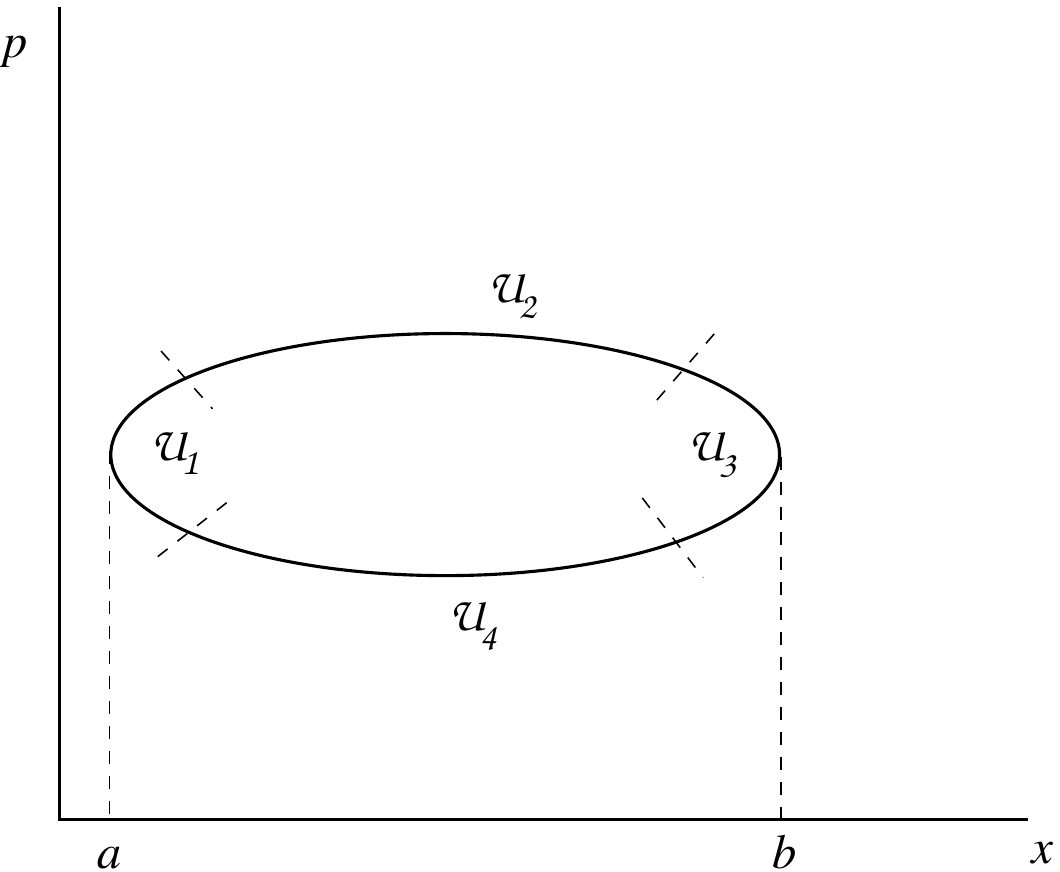}

\caption{
The left figure depicts a graph of 
the molecular dynamics potential $\lambda(X)$ 
in the case which exhibits caustics at $X=a$ and $X=b$ for a given energy $E$.
The right figure shows a general case of the Lagrangian manifold with two caustic points $X=a$ and $b$ and its covering with charts $\mathcal{U}_i$. In the 
charts $\mathcal{U}_i$, $i=2,4$ the manifold is defined by $P=\GRADX\theta_i(X)$ and the solution to Schr\"odinger
equation is constructed by simple WKB modes. The caustics belong to the charts $\mathcal{U}_i$, $i=1,3$ and in this
case the manifold is defined by $X = \GRADX\theta_i(P)$ and the solutions are given by the Fourier integrals.} 
\label{phasespacecaustics}

\end{figure}

\medskip

We see that the weight function $G$, in \eqref{liouville_form}, based on a single WKB-mode \eqref{wkb_form} blows up
at caustics, where $\DET(\partial (\CHECKX)/\partial (\CHECKP))=0$, and that the weight function $G$ in
\eqref{caustic_transport} for the Fourier integral \eqref{caustic_ansats}
blows up at points where $\DET(\partial (\CHECKP)/\partial (\CHECKX))$ vanishes.
Therefore, in neighborhoods around
caustic points we need to use the representation  $\LFT{\theta}(\HATX,\CHECKP)$ of the phase 
based on the Fourier integrals, while around points where $\DET(\partial (\CHECKP)/\partial (\CHECKX))$ vanishes
we apply the representation $\theta(\HATX,\CHECKX)$ based on the Legendre transform,
as pointed out by Maslov in \REF{maslov} and described in the simplifying setting of the harmonic oscillator in \REF{eckman}.

One way to make a global construction of a WKB solution, which is slightly different than in 
\REF{maslov},
is to use the characteristics and a partition of the phase-space as follows, also explained constructively by the numerical algorithm \ref{alg_2} in the next section.
Start with a Fourier integral representation in a neighborhood $\NEIGH$ of a caustic point, which gives
a representation of the Schr\"odinger solution $\Phi$ in $\NEIGH$.
Then we use the stationary phase expansion, see Section~\ref{stat_phase_sec}, 
 to find an asymptotic approximation $\tilde\Phi$ (accurate to any order $\bar N\in \mathbb{N}$) 
at the boundary points $\CHECKX$ of
$\NEIGH$ as a sum of single WKB-modes with phase functions $\theta_j$
\[
   \int_{\rset^d} \chi(\CHECKP)\EXP{-\Iunit M^{1/2}(\EPROD{\CHECKX}{\CHECKP} -\LFT{\theta}(\HATX,\CHECKP))}\, d\CHECKP
     =\sum_j \EXP{-\Iunit M^{1/2}\theta_j(X)} \hpsi_j(X) +\BIGO(M^{-\bar N})
\]
where each phase function $\theta_j(X):=\EPROD{\CHECKX}{\CHECKP_{X,j}}-\LFT{\theta}(\HATX,\CHECKP_{X,j})$
corresponds to a branch of the boundary and  the index $j$ corresponds to different solutions $\CHECKP_{X,j}$
of the stationary phase equation $\CHECKX=\nabla_{\CHECKP}\LFT{\theta}(\HATX,\CHECKP_X)$.
The single WKB-modes $\hpsi(x,X) \EXP{\Iunit M^{1/2}\theta(X)}$ are then constructed along the characteristics 
to be Schr\"odinger solutions in a domain 
around the point where $\DET(\partial (\CHECKP)/\partial (\CHECKX))$ vanishes, following the construction 
in Theorem~\ref{thr:schrodinger-hamiltonian_first}
using the initial data of $\tilde\Phi$ at $\partial\NEIGH$. We note that
the tiny error of size $\BIGO(M^{-\bar N})$ that we make in the initial data for $\hpsi$
 also yields a tiny perturbation error in $\hpsi$ of size $\BIGO(M^{-\bar N})$ along the
path, due to the assumption of the finite hitting times. 
A small error we make in the
expansion therefore leads to a negligible error in the Schr\"odinger solution and 
the corresponding density.

When a characteristic leaves the domain and enters another region around a caustic
we again use the stationary phase method at the boundary to
give initial data for $(X,P,\hpsi,G)$. When the characteristic finally returns to the first boundary 
$\partial \NEIGH$,
there is a compatibility condition  to have a global solution, by having the incoming final phase equal to the
initial phase function in $\mathcal C^1$.
We can think of this as trying to find a co-dimension one surface $I$ in $\rset^{3N}$
where the incoming and outgoing phases are equal. First to have one point where they
agree is possible if we
restrict the possible solutions to a discrete set of energies $E$, i.e., the eigenvalues, and
therefore the compatibility condition is called  a quantization condition. Then, having one
point where the difference of the two phase function is zero, we can 
combine this with 
the assumption that the Lagrangian manifold generated by the characteristics path 
$(X_t,P_t)$ is continuous: the two phases have the same gradient on $I$, since 
$(X,P)=(X,\GRADX\theta(X))=\Big(\big(\hat X,\nabla_{\check P} \theta^*(\hat X,\check P)\big), \big(\nabla_{\hat X}\theta^*(\hat X,\check P),\check P\big)\Big)$ 
so the phases are $\mathcal C^1$.
In this way we define the $(X,P,\hpsi,G)$ globally, for the eigenvalue energies $E$.
To evaluate observables we use a partition of unity to restrict the 
observable to a domain with a single representation, either
a Fourier integral representation for a caustic or a single WKB-mode
when $\DET(\partial (\CHECKP)/\partial (\CHECKX))=0$.

\section{Numerical examples}

In order to demonstrate the presented theory we consider two different low dimensional 
Schr\"odinger problems. For both of these problems we show that there exists a
Schr\"odinger eigenfunction density which converges weakly to 
the corresponding molecular dynamics density as $M\to \infty$ 
with a convergence rate within the upper bound predicted in the theoretical 
part of this paper.

\subsection{Example 1: A single WKB state}

The first problem we consider is the time-independent Schr\"odinger
equation
\begin{equation}\label{eq:schrod_example_1}
\HOPER\Phi := \left(-\frac{1}{2M}\partial_{XX} + \overline \VOPER\right) \Phi = E \Phi
\end{equation}
with heavy coordinate $X \in (-\pi,\pi]$ and two-state light
coordinate $x \in \{x_-,x_+\}$. Periodicity is assumed over the heavy
coordinate, $\Phi(X,x) = \Phi(X+2\pi,x)$, and the potential operator 
$\overline\VOPER$ is defined by the matrix
\begin{equation} \label{eq:VXiMatrix}
   \overline \VOPER(X) = \left[\begin{array}{cc}
                               V(X)& \tfrac{1}{2}V(X)e(X) +c\\
                               \tfrac{1}{2}V(X)e(X) + c & 0
                         \end{array}\right]\COMMA
\end{equation}
where we have chosen $V(X) = -2\cos(X)+ \cos(4X)$,  $e(X) = 1+X^2$ and 
$c$ to be a non-negative constant relating to the size 
of the spectral gap of $\overline \VOPER$. The action
$\overline \VOPER \Phi$ is thus defined by
$$
   (\overline \VOPER\Phi)(X,\cdot)\equiv \overline \VOPER(X) 
                       \left(\begin{array}{c} 
                          \Phi(X,x_-)\\
                          \Phi(X,x_+)
                       \end{array}\right)\PERIOD
$$
For each $X$ the potential matrix \eqref{eq:VXiMatrix} gives rise to 
the eigenvalue problem
$$
   \overline \VOPER (X) \upsilon = \lambda_\pm (X) \upsilon
$$
with the eigenvalues
$$
  \lambda_{\pm}(X) = \frac{1}{2} \left( V(X) \pm \SGNF(X) \sqrt{V(X)^2+4\big(V(X)e(X)/2+c\big)^2}\right)\COMMA
$$ 
where $\SGNF(X)=\pm 1$ as defined below.
When constructing the molecular dynamics density for this problem
$$
\rho_{\MD}(X)  = \frac{C}{\sqrt{2(E-\lambda(X))}}\COMMA
$$ 
one has to determine on which of the two eigenfunctions $\lambda_{\pm}$
to base this density. When $c=0$ the difficulty that the eigenvalue
functions $\lambda_+$ and $\lambda_-$ can cross is added to the
problem. In order to determine the continuation of eigenvalue functions at
the crossings we introduce a function $\SGNF(X)$ which is a sign function with
$\SGNF(-\pi) = 1$ that changes sign at points where
$$
V(X)^2+4\left(\frac{1}{2}V(X)e(X)+c\right)^2=0\PERIOD
$$ 
Since this situation can only occur when $c=0$, it is possible to set
$$
\SGNF(X):=\SGN(V(-\pi))\SGN(V(X))\PERIOD
$$
See Figure~\ref{fig:ev_functions_crossing} for a typical
eigenvalue function crossing, which makes the function 
$\lambda_{\pm}:\rset\rightarrow\rset$ smooth (in contrast to the choice $\SGNF\equiv 1$).

To solve \eqref{eq:schrod_example_1} numerically, we use the finite difference method to 
discretise the operator $\HOPER$ on a grid $\{X_j\}_{j=1}^N \times \{x_- , x_+\}$ 
with the step-size $h= 2\pi/N$ and $X_j = jh$. The discrete eigenvalue problem 
$$
\HOPER^{(h)} \Upsilon_j = E_j \Upsilon_j
$$
is solved for the 10 eigenvalues being closest to the fixed energy $E$ and 
a molecular dynamics approximation of the eigensolution is constructed by
$$
\Phi_{\MD}(X,x) := \sqrt{\rho_{\MD}(X)}\,\EXP{\Iunit M^{1/2}\Theta(X)}\upsilon(X,x)\COMMA
$$
where $\upsilon(X,\cdot)$ is one of the eigenvectors of $\overline \VOPER(X)$ and  
\begin{equation}\label{eq:Theta_FunctionA}
   \Theta(X) := \int_0^X \sqrt{2(E_1-\lambda(s))}\, ds
\end{equation}
is approximated by a trapezoidal quadrature yielding $\Theta^{(h)}$.
Thereafter a Schr\"odinger eigensolution $\WIDEHATPHI$ 
which is close to the molecular dynamics
eigensolution is obtained by projecting $\Phi_{MD}$ onto the subspace
spanned by $\{\Upsilon\}_{j=1}^{\bar J}$ as described in 
Algorithm~\ref{alg:projection_algorithm}.
By denoting $\rho_{\WIDEHATPHI}(X) = \LPROD{\WIDEHATPHI}{\WIDEHATPHI}$ 
and  $\rho_{\MD}(X) = \LPROD{\Phi_{\MD}}{\Phi_{\MD}}$, the 
observables $g_1(X) = X^2$ and $g_2(X) = V(X)$ are used to compute 
the convergence rate of 
\begin{equation}\label{eq:conv_rate_densities}
  \left|\frac{\int_{-\pi}^{\pi} g_i(X) \rho_{\MD}(X)\, dX - 
              \int_{-\pi}^{\pi} g_i(X) \rho_{\WIDEHATPHI}(X)\,dX}%
             {\int_{-\pi}^{\pi} g_i(X) \rho_{\MD}(X)\, dX}
  \right|\COMMA
\end{equation}
as $M$ increases. Further details of the numerical solution idea are
described in Algorithm~\ref{alg:pseudoCode}.

Plots of the results for the test case with the spectral gap
$c=5$ and $E=0$, and for the test case with crossing eigenvalue
functions when $c=0$ and $E=1.2$ are given below. Most noteworthy is
Figure~\ref{fig:conv_rates}, which demonstrates that the obtained convergence
rate for \eqref{eq:conv_rate_densities} is $\BIGO(M^{-1})$ for both
scenarios.

\begin{figure}[h!]
  \centering
    \includegraphics[width=.49\textwidth]{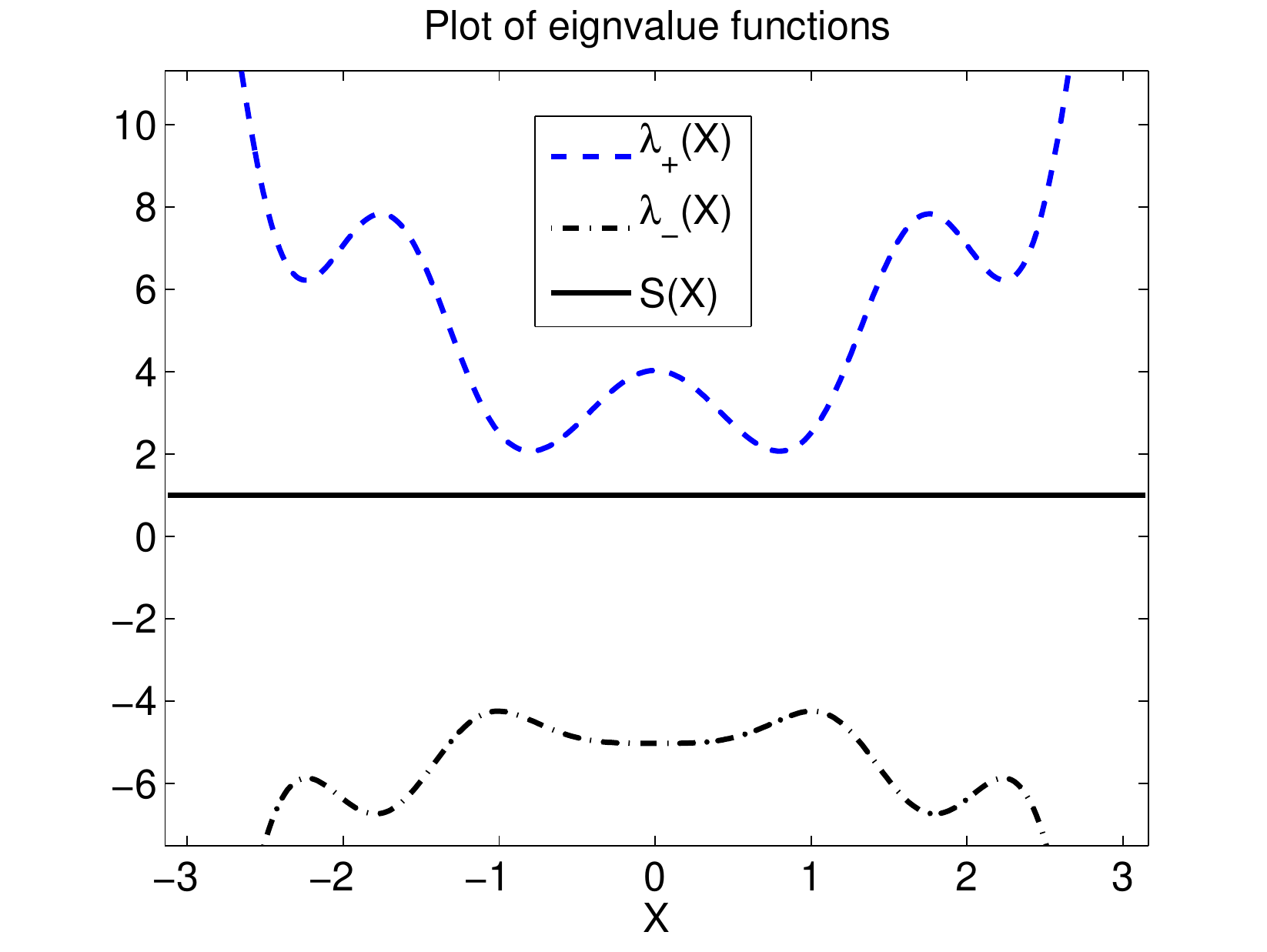}
    \includegraphics[width=.49\textwidth]{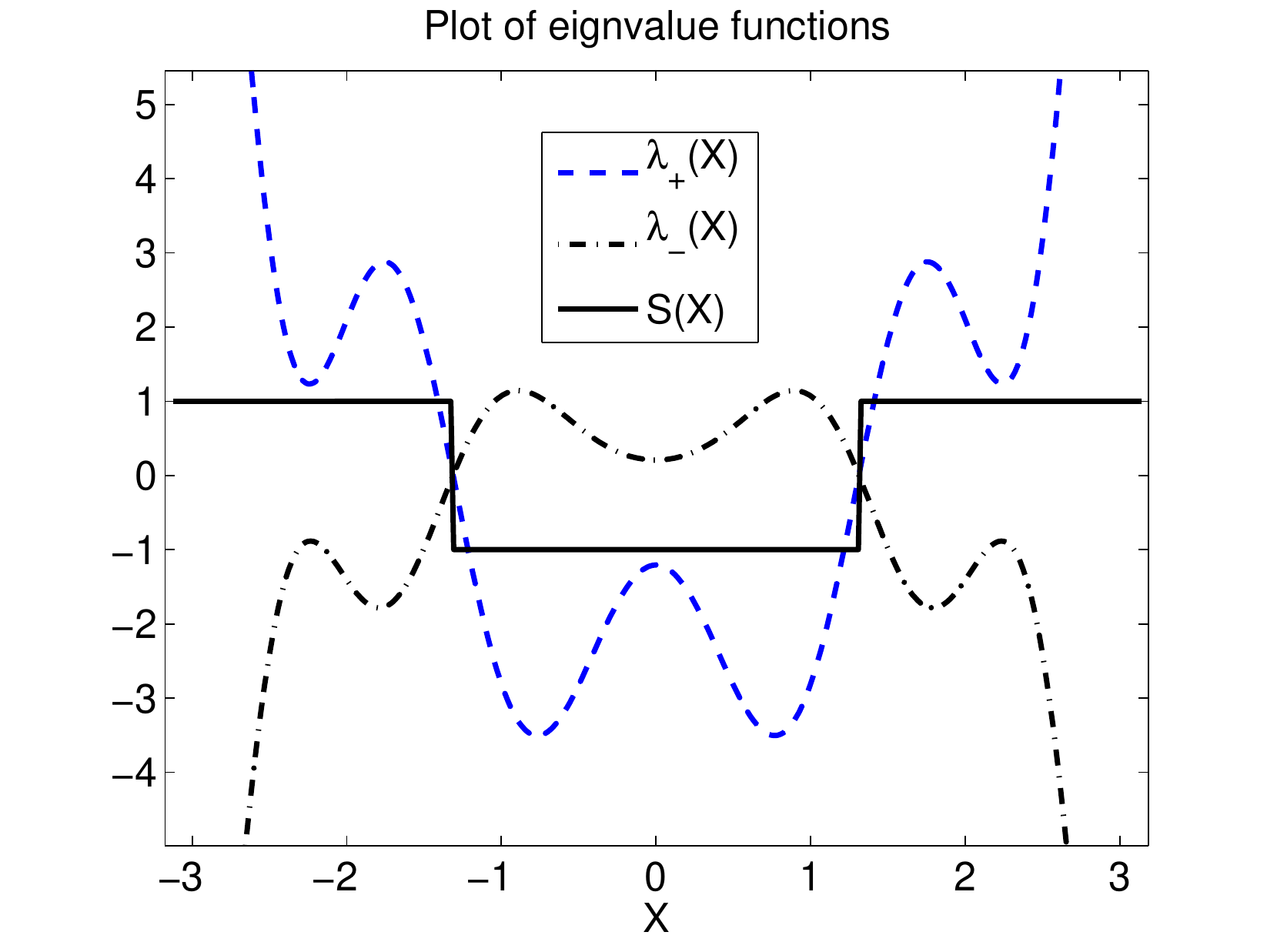}
    \caption{Left plot: Eigenvalue functions when $c=5$. There is a 
      spectral gap which makes the sign function constant $S=1$.  
      Right plot: Eigenvalue functions when $c=0$. The eigenvalue functions
      exhibit crossing, consequently the function $S$ changes its sign from $\pm 1$ to $\mp
      1$ at the crossing points.}
\label{fig:ev_functions_crossing}
\end{figure}

\begin{algorithm}[!h]
\caption{Algorithm for problems in Example 1}
\begin{algorithmic}\label{alg:pseudoCode}
\STATE {\bf Input: } Energy $E$; potential functions $V$, $e$ and $c$; mass $M$;
number of grid points $N$ and grid $\{X_i\}_{i=1}^N$.

\STATE {\bf Output: } Schr\"odinger projection density $\rho_{\WIDEHATPHI}$.

\medskip

\STATE {\bf 1.} Construct the discrete operator $\HOPER^{(h)}$ from
\eqref{eq:schrod_example_1} using finite differences and solve the
eigenvalue problem
$$
  \HOPER^{(h)} \Upsilon_i = E_i \Upsilon_i
$$
for the 10 eigenvalues being closest to $E$ by using MATLAB \textbf{eigs(H,10,E)}.

\medskip

\STATE {\bf 2.} Sort the eigenvalues and eigenvectors by distance from $E$ and 
keep only the $E_i$s which are less than $M^{-1/2}$ away from $E$.
Let $\bar J$ be the number of kept eigenvalues and $E_0$ the eigenvalue closest to $E$.

\medskip

\STATE {\bf 3.}
\FOR{$i=1$ to $N$}
   \STATE Solve the eigenvalue problem 
   $$
     \overline \VOPER(X_i,\cdot) \upsilon_\pm(X_i,\cdot) = \lambda_\pm(X_i) \upsilon_\pm(X_i,\cdot)\COMMA
   $$ 
   where $\overline \VOPER$ is the matrix defined in \eqref{eq:VXiMatrix}.
\ENDFOR

\medskip

\STATE {\bf 4.} Construct the molecular dynamics density according to the formula
$$
  \rho_{\MD}(X) = \frac{\left(E_0-\lambda(X)\right)^{-1/2}}{\int_{[0,2\pi]} \left(E_0-\lambda(X)\right)^{-1/2}\, dX}\COMMA
$$
where we choose $\lambda(X)$ above from the two eigenvalues $\lambda_\pm(X)$ by
the criterion that the eigenvalue chosen must fulfil $\|\lambda\|_\infty < E_0$.

\medskip 

\STATE  {\bf 5.} Construct a discrete molecular dynamics approximation to the eigenfunction
\begin{equation}\label{eq:phi_Function}
  \Phi_{\MD}(X,x) = \sqrt{\rho_{\MD}(X)}\EXP{\Iunit M^{1/2}\Theta(X)}\upsilon(X,x)\COMMA
\end{equation}
where $\upsilon(X,x)$ is one of the eigenvectors $\upsilon_\pm$, 
\begin{equation}\label{eq:Theta_FunctionB}
  \Theta(X) :=  \int_0^X \sqrt{2(E_1-\lambda(s))}\, ds\COMMA
\end{equation}
and we approximate $\Theta$ by a trapezoidal quadrature $\Theta^{(h)}$.

\medskip

\STATE {\bf 6.} Project the molecular dynamics solution $\Phi_{\MD}$ onto the eigenspace 
$\{\Upsilon_i\}_{i=1}^{\bar J}, \  \bar J\leq 10$ by 
Algorithm~\ref{alg:projection_algorithm} to obtain a projection
solution $\WIDEHATPHI$.

\medskip

\STATE {\bf 7.} Derive the Schr\"odinger projection density 
by
\FOR{$i=1$ to $N$}
\STATE
$$
{\rho_{\WIDEHATPHI}}(X_i) = |\WIDEHATPHI(X_i,x_-)|^2+|\WIDEHATPHI(X_i,x_+)|^2\COMMA
$$
\ENDFOR
\STATE and scaling $\rho_{\WIDEHATPHI} =  \rho_{\WIDEHATPHI}/\|\rho_{\WIDEHATPHI}\|$.

\end{algorithmic}
\end{algorithm}

\begin{algorithm}
\caption{Projection algorithm}
\begin{algorithmic}\label{alg:projection_algorithm}
\STATE {\bf Input: } Mass $M$; wave solution $\Phi$; eigenvalues $\{E_i\}_{i=1}^{\bar J}$ and corresponding 
eigenvectors $\{\Upsilon_i\}_{i=1}^{\bar J}$.

\STATE {\bf Output: } Schr\"odinger 
projection wave solution $\WIDEHATPHI$.

\medskip

\STATE {\bf 1.} Organize eigenvalues by multiplicity by a numerical 
approximation. Construct a $\bar J \times \bar J$, zero matrix $A$ which  
keeps track of multiplicity relations as follows:
\FOR{$i=1$ to $\bar J$}
\FOR{$j=i$ to $\bar J$}
\IF {$|E_i-E_j| < M^{-3/4}$ }
\STATE Consider eigenvalues equal since the expected spectral gap is $\BIGO(M^{-1/2})$,
and store this relation by 
\IF {$A_{kj}=0$ for all $k<i$}
\STATE Set $A_{ij}=1$.
\ENDIF
\ENDIF
\ENDFOR
\ENDFOR

\medskip
\STATE {\bf 2.} For vectors $b \in \{0,1\}^{\bar J}$, define the projection 
$$
\Phi^{(h,b)} := \sum_{j,k=1}^{\bar J} b_k A_{k,j} 
\langle\!\langle\Phi,\Upsilon_j\rangle\!\rangle \Upsilon_j
$$
and, letting $\rho$ and $\rho_{\Phi^{(h,b)}}$ denote the densities generated by 
$\Phi$ and $\Phi^{(h,b)}$ respectively, set
$$
b^* = \arg \min_{b \in \{0,1\}^{\bar J}}\|\rho-\rho_{\Phi^{(h,b)}}\|.
$$

\STATE {\bf 3.} Return the projection $\Phi^{(h)} := \Phi^{(h,b^*)}$.

\end{algorithmic}
\end{algorithm}

\begin{figure}[h!]
  \centering
   \includegraphics[width=.49\textwidth]{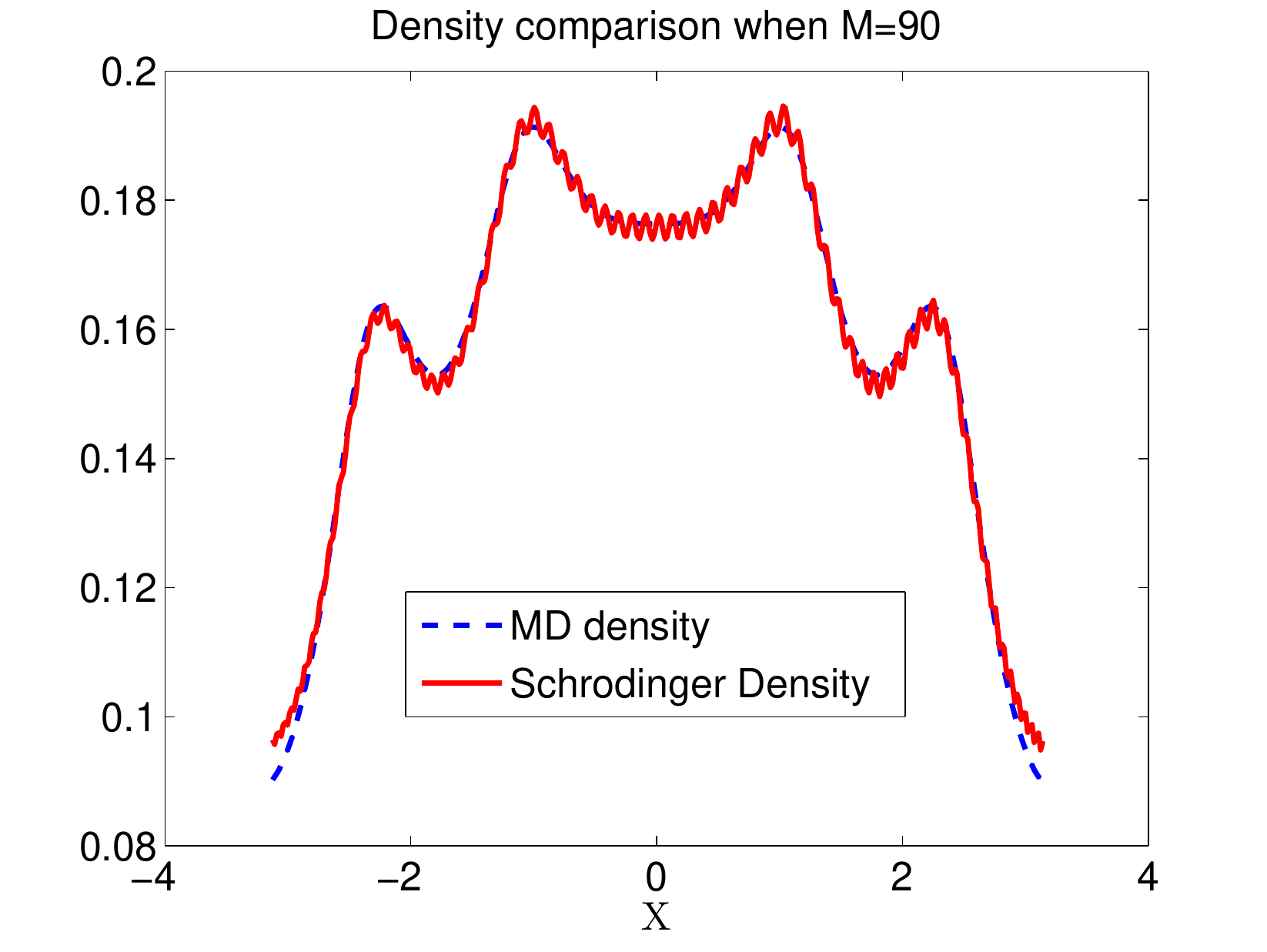}
   \includegraphics[width=.49\textwidth]{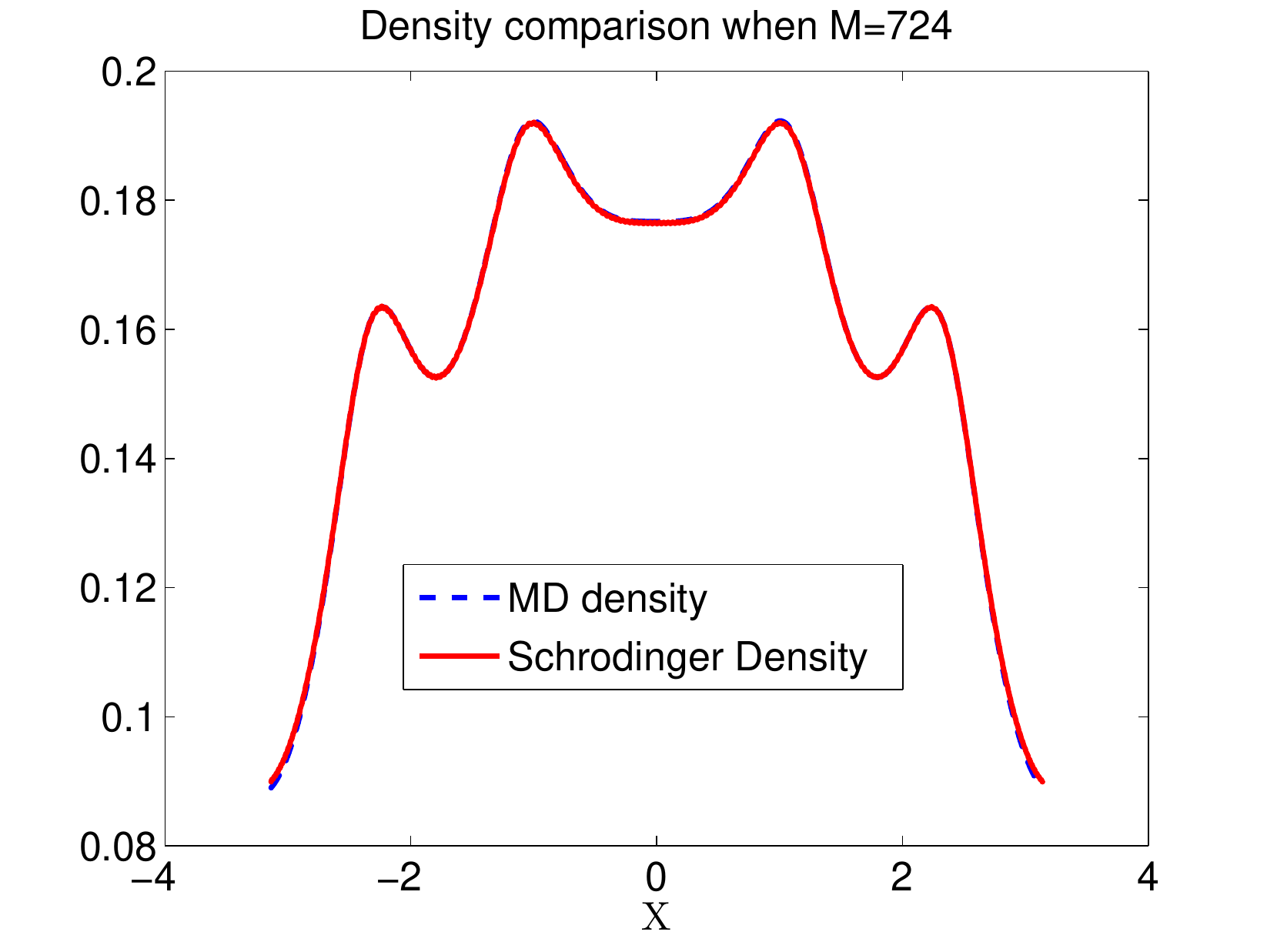}
   \caption{Plot of the MD density $\rho_{\MD}$ and the Schr\"odinger projection 
     density $\rho_{\WIDEHATPHI}$ in the case $c=5$ and $E=0$ for the two
     different masses $M=90$ (left plot) and $M=724$ (right plot)
     illustrating the convergence of the densities.}
   \label{fig:density_plots_gap}
\end{figure}

\begin{figure}[h!]
  \centering
   \includegraphics[width=.49\textwidth]{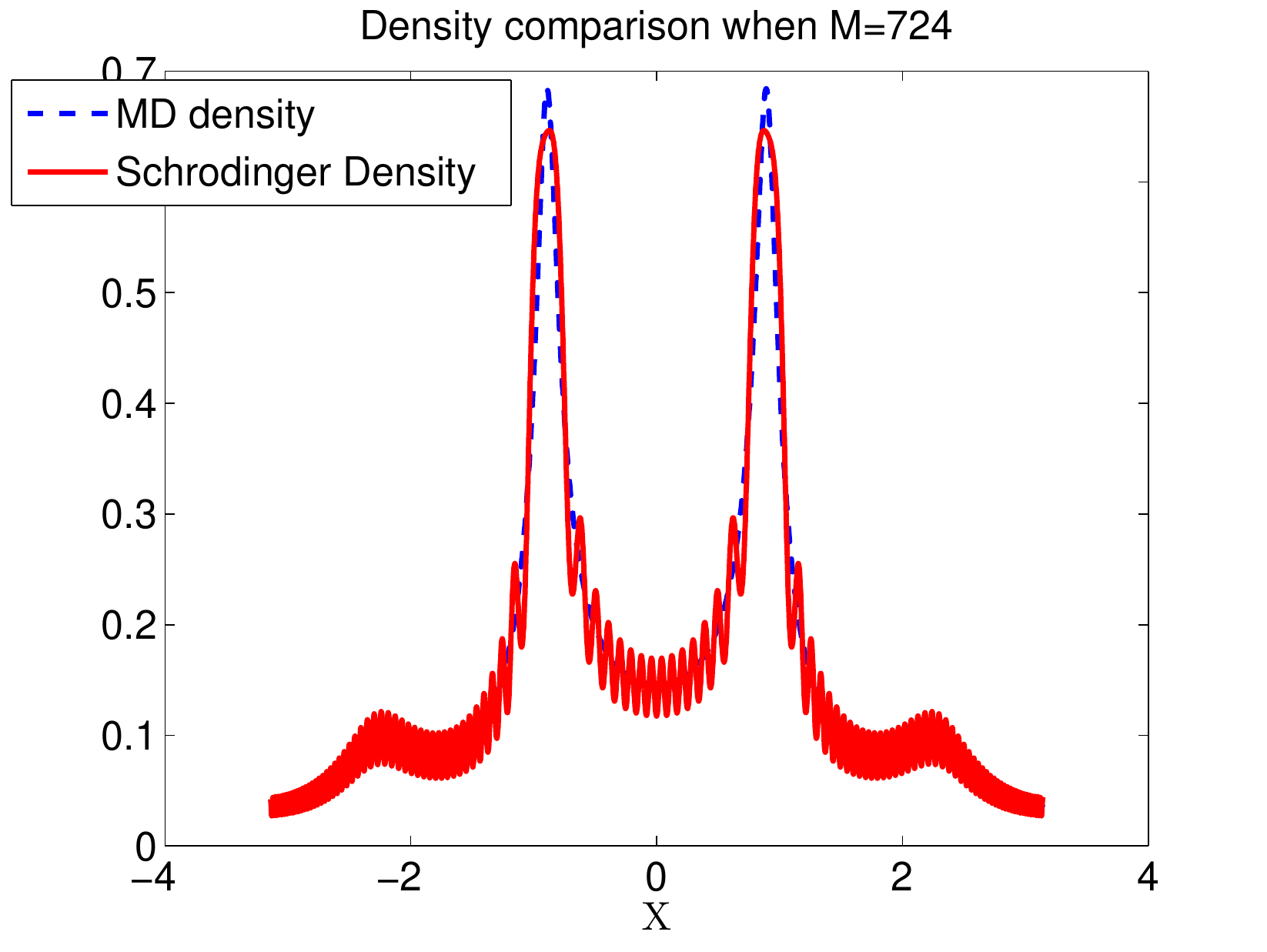}
   \includegraphics[width=.49\textwidth]{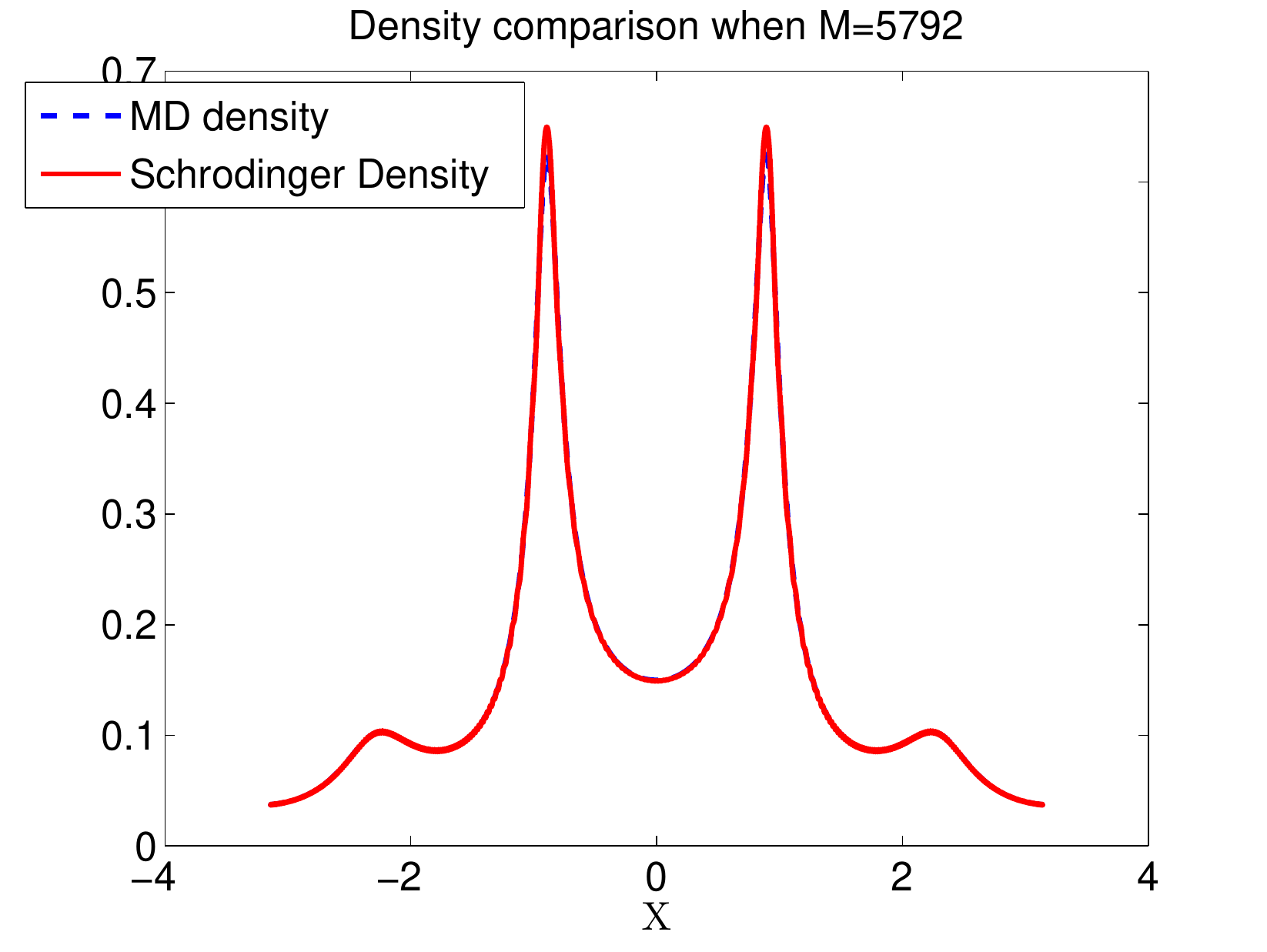}
    \caption{Plot of the MD density $\rho_{\MD}$ and Schr\"odinger projection 
      density
      $\rho_{\WIDEHATPHI}$ in the case $c=0$ and $E=1.2$ for the two
      different masses $M=724$ (left plot) and $M=5792$ (right plot).}
    \label{fig:density_plots_no_gap}
\end{figure}

\begin{figure}[h!]
  \centering
    \includegraphics[width=.49\textwidth]{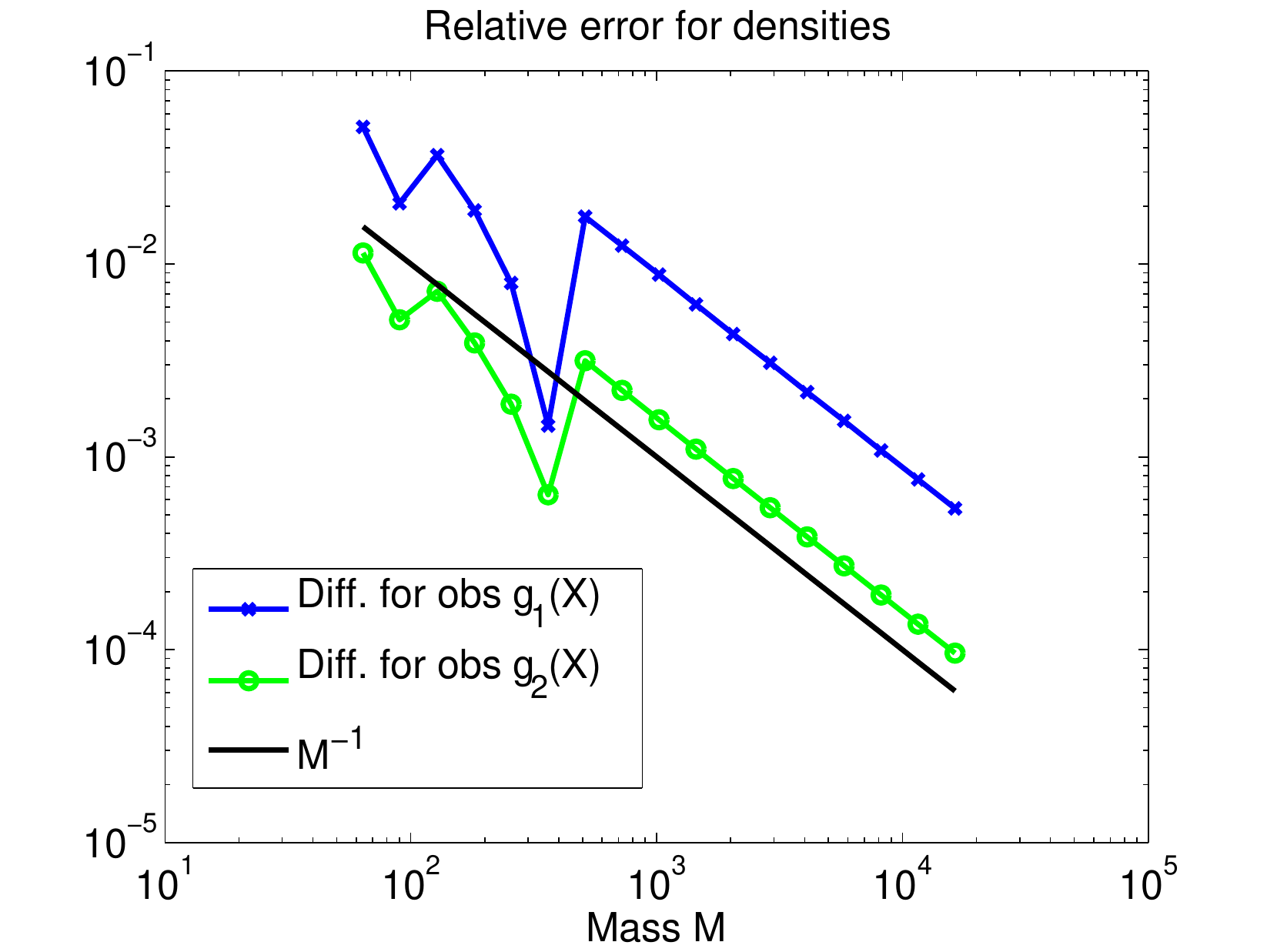}
    \includegraphics[width=.49\textwidth]{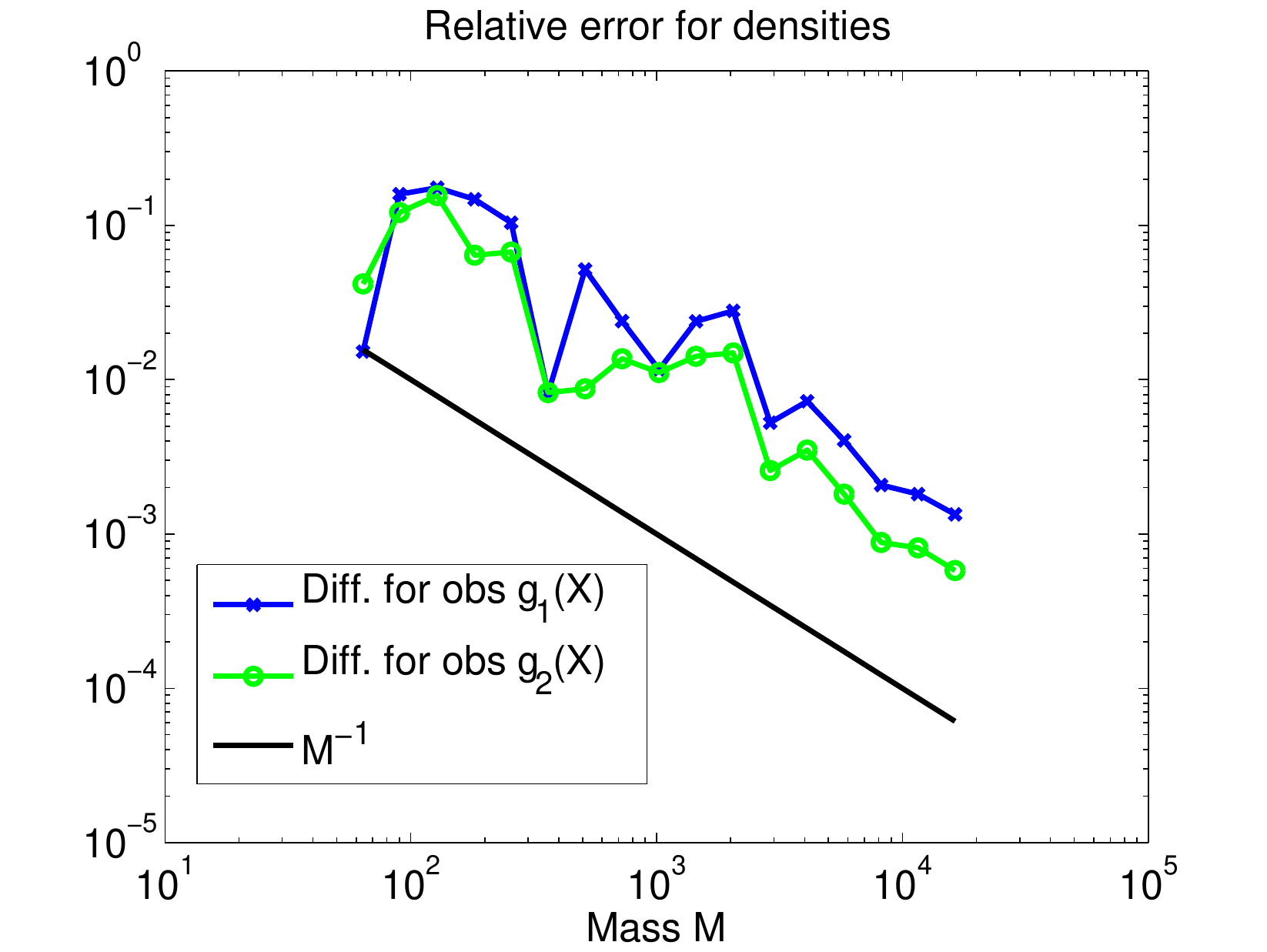}
    \caption{Left plot: Plot of the observable density errors given in 
      \eqref{eq:conv_rate_densities} with an eigenvalue gap, when $c=5$ and $E=0$. 
      Right plot: Plot of the observable density errors given in 
      \eqref{eq:conv_rate_densities} with an eigenvalue crossing, when $c=0$ and $E=1.2$.}
\label{fig:conv_rates}
\end{figure}

\subsection{Example 2: A caustic state}
Next, we consider the one dimensional, time independent, periodic Schr\"odinger equation 
\begin{equation}\label{eq:schrod_example_2}
  \left(-\frac{1}{2M}\partial_{XX}+ V\right) \Phi = E \Phi\COMMA \;\;\; X\in (-2\sqrt{E},2\sqrt{E})
\end{equation}
with $V(X) = X^2$ and $E=1$.
The eikonal equation corresponding to \eqref{eq:schrod_example_2} is
\begin{equation}\label{eq:caustics_hamiltonian}
   \frac{1}{2}{P^2} + V(X) = E\PERIOD
\end{equation}
As in Example~1, we would like to use the eikonal equation to construct a numerical 
approximate solution of \eqref{eq:schrod_example_2}
whose density converges weakly as $M \to \infty$ to the density generated from a
solution of \eqref{eq:schrod_example_2}.
The molecular dynamics density corresponding to this eikonal equation becomes by \eqref{r_bo_dens}
$\rho_{\BO}=C (E-V(X))^{-1/2}$. The density $\rho_{\BO}$ goes to  infinity at the caustics 
$X = V^{-1}(E)=\pm\sqrt{E}$ and the approach in Example~1 does not work directly. 
We will instead construct the numerical
approximate solution using the stationary phase method as outlined
below based on the WKB Fourier integral ansatz.

By the Legendre transform 
$$
  \LFT{\theta}(P) = \min_{X} \big(X P-\theta(X)\big)
$$ 
an invertible mapping between the momentum and
position coordinates fulfilling $X = \GRADP \LFT\theta(P)$ is constructed.
Using  equation \eqref{eq:caustics_hamiltonian}, 
one sees that $\GRADP\LFT\theta(P) = V^{-1}(E-P^2/2)$. Since
$\LFT \theta(0) = 0$, one can derive that for this particular 
choice of $V$
$$
  \LFT{\theta}(P) = \int_0^P  \sqrt{E-s^2/2} \, ds =
  \frac{E}{\sqrt{2}}\left[\sin^{-1}\left(\frac{P}{\sqrt{2E}}\right)
   + \frac{P}{\sqrt{2E}}\sqrt{1-\frac{P^2}{2E}}\right]. 
$$ 
In neighbourhoods of the caustics $[-2E^{1/2},-X_0)$ and $(X_0,2E^{1/2}]$, we
construct the approximate solution by
$$ 
  \Phi(X) = \frac{u(X)}{\sqrt{|\GRADX V(X)|}}
$$
where $u$ is the inverse Fourier transform
$$
  u(X):= \int_{-2\sqrt{E}}^{2\sqrt{E}} \EXP{\Iunit M^{1/2}(-X P+\LFT\theta(P))}\, dP
$$ 
and $X_0\in (-V^{-1}(E),V^{-1}(E))$ is a value yet to be chosen.
In the region $(-X_0,X_0)$ the approximate solution is constructed by 
\begin{equation}\label{eq:approximate_solution_interior}
  \Phi(X) =C \frac{\overline u(X)}{(E-V(X))^{1/4}}\PERIOD
\end{equation}
Here
\begin{equation}\label{eq:stat_phase_u}
\overline u(X) := \EXP{-\Iunit M^{1/2}\theta(X)} \psi_+ + \EXP{\Iunit M^{1/2}\theta(X)} \psi_-\COMMA
\end{equation}
with, according to the Legendre transform, 
$\theta(X):=X \sqrt{2(E-V(X))}-\LFT\theta\left(\sqrt{2(E-V(X))}\right)$
and $\psi_\pm$ determined by the stationary phase method:
\begin{enumerate}
\item[{\bf 1.}] Set $P(p) = P_0+p$ with $P_0=\sqrt{2(E-V(X_0))}$ and let
  $$
    Y(p) := 
       \SGN(p)\sqrt{2\frac{-X(P_{0}+p) +\LFT\theta(P_{0}+p)+\theta(X_0)}{\partial_{PP}\LFT\theta(P_0)}}\COMMA
  $$
  using 
  \begin{equation}\label{eq:theta_function_caustics}
     \theta(X):= X \sqrt{2(E_0-V(X))}-\LFT\theta\left(\sqrt{2(E_0-V(X))}\right)\COMMA
  \end{equation}
  and determine its inverse $p(Y)$ in a neighbourhood of $Y=0$ by computing 
  $(p_i,Y(p_i))$ on a grid around $p=0$ and, for $k\geq 3$, fit a 
  $3k+1$th degree polynomial to the values $(Y(p_i),p_i)$ using the method of least squares.
\item[{\bf 2.}] 
  Evaluate the stationary phase expansion
  \begin{equation}\label{eq:stat_phase_equation_example_2}
    \begin{split}
       u(X_0)& = \sum_{p_0=\pm \sqrt{2(E-V(X_0))}} \EXP{\Iunit \pi\SGN (\partial_{PP}\LFT\theta(P_0))/4}
                 \left[ \left|\frac{1}{2}\partial_{PP}\LFT\theta(P_0)\right|^{-1/2} 
                      \EXP{-\Iunit M^{1/2}\theta(X_0) }\right.\\
             & \qquad \times 
               \left. \sum_{j=0}^k \frac{M^{-j/2}}{j!}
                    \left(\Iunit \left(\frac{1}{2}\partial_{PP}\LFT\theta(P_0)\right)^{-1}%
                    \partial_{YY}\right)^j|\partial_Y p| \right|_{Y=0}
               + \BIGO(M^{-j/2})\Bigg]
    \end{split}
  \end{equation}
  to obtain
  \[
     u(X_0^-)= \EXP{ \Iunit M^{1/2}\theta(X_0)} (\psi_+  +\BIGO(M^{-k/2})) 
              +\EXP{-\Iunit M^{1/2}\theta(X_0)} (\psi_- + \BIGO(M^{-k/2}))\COMMA
\]
where 
$$
  \psi_{\pm} := \EXP{\Iunit \pi\SGN (\partial_{PP}\LFT\theta(\pm P_0))/4}
                \left|\frac{1}{2}\partial_{PP}\LFT\theta(\pm P_0)\right|^{-1/2} 
                \sum_{k=0}^3 \frac{M^{-k/2}}{k!}
                \left(\Iunit\left(\frac{\partial_{PP}\LFT\theta(\pm p_0)}{2}\right)^{-1}
                \partial_{YY}\right)^k|\partial_Y p| \Big|_{Y=0}.
$$
\end{enumerate}

\begin{figure}[h!]
  \centering
   \includegraphics[width=.6\textwidth]{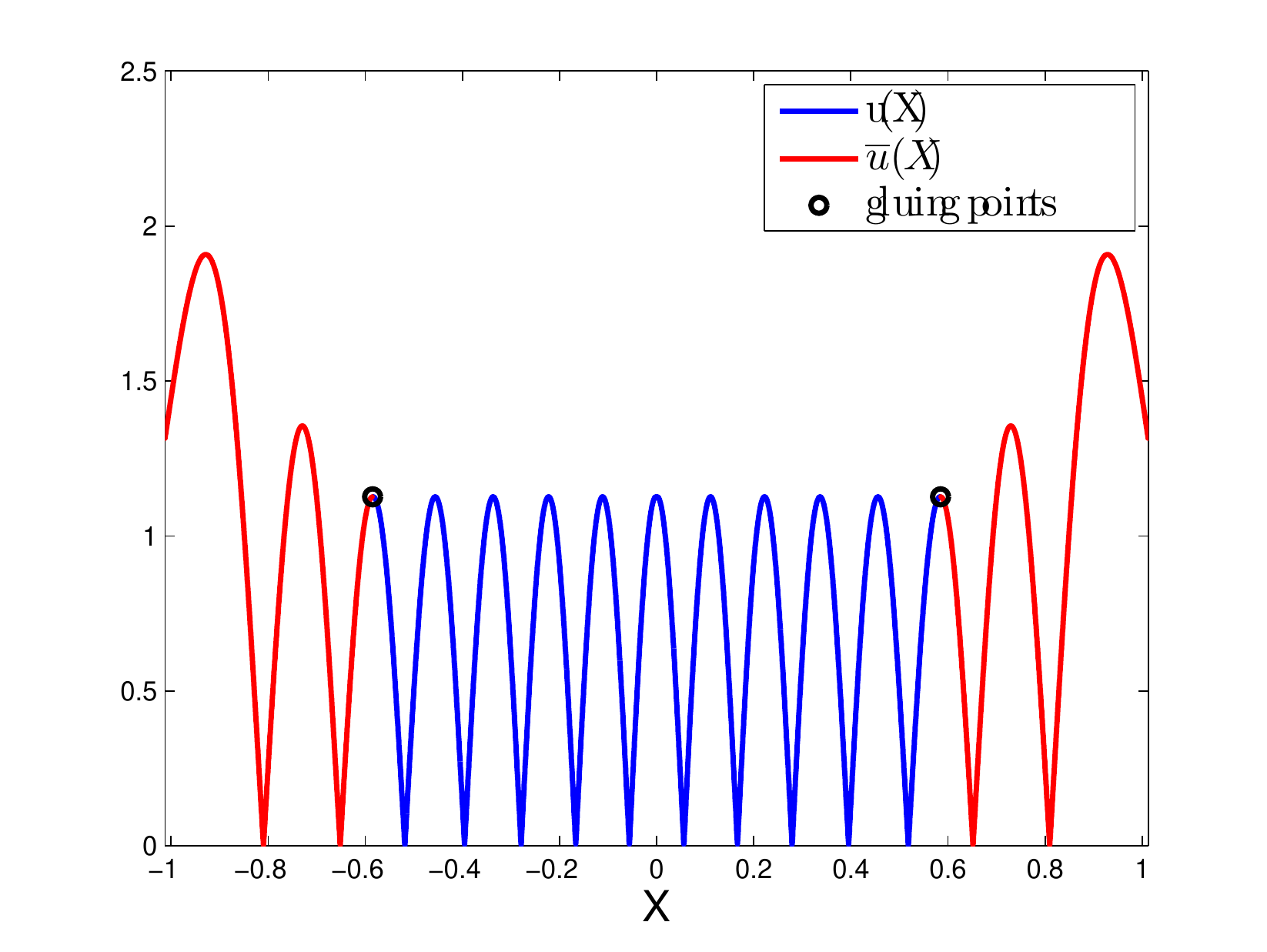}
   \caption{Plot illustrating the gluing procedure of the functions
     $u(X)$ and $\overline u(X)$ at the points $\pm X_0$.
   \label{fig:gluing_procedure}}
\end{figure}

The constant $C$ in \eqref{eq:approximate_solution_interior} 
is chosen so that the wave solution parts are
continuous at the gluing point, $\Phi(\pm X_0^-) = \Phi(\pm X_0^+)$.
It is most easy to determine $C$ when $X_0$ is chosen so that 
$|u(X_0)|$ is at a local maximum; see Figure~\ref{fig:gluing_procedure} 
for an illustration of the gluing procedure.  

At the end a Schr\"odinger eigenfunction solution 
$\WIDEHATPHI$ is obtained by projecting $\Phi$ onto the space
spanned by a set of eigensolutions to the discretized version of
the Schr\"odinger problem, $\{\Upsilon_j\}_{j=1}^{\bar J}$, as is
described in Algorithm~\ref{alg:projection_algorithm}.

Two convergence results are needed to make the method work.  First,
the density generated from the stationary phase based on the approxmiate
solution $\rho(X):=|\Phi|^2(X)/\|\Phi\|_2^2$ must converge weakly to
the Schr\"odinger projection based density $\rho_{\WIDEHATPHI}(X):=
|\WIDEHATPHI|^2(X)/\|\WIDEHATPHI\|_2^2$ as $M\to \infty$; see
Figure~\ref{fig:density_plots_stat_phase_caustics} for an illustration
of how these functions converge. Second, $\rho_{\WIDEHATPHI}$ must
converge to the molecular dynamics density $\rho_{\MD}(X):=C
(E-V(X))^{-1/2}$ as $M$ increases; see Figure~\ref{fig:density_plots_projected_caustics}.
\begin{figure}[h!]
  \centering
   \includegraphics[width=.49\textwidth]{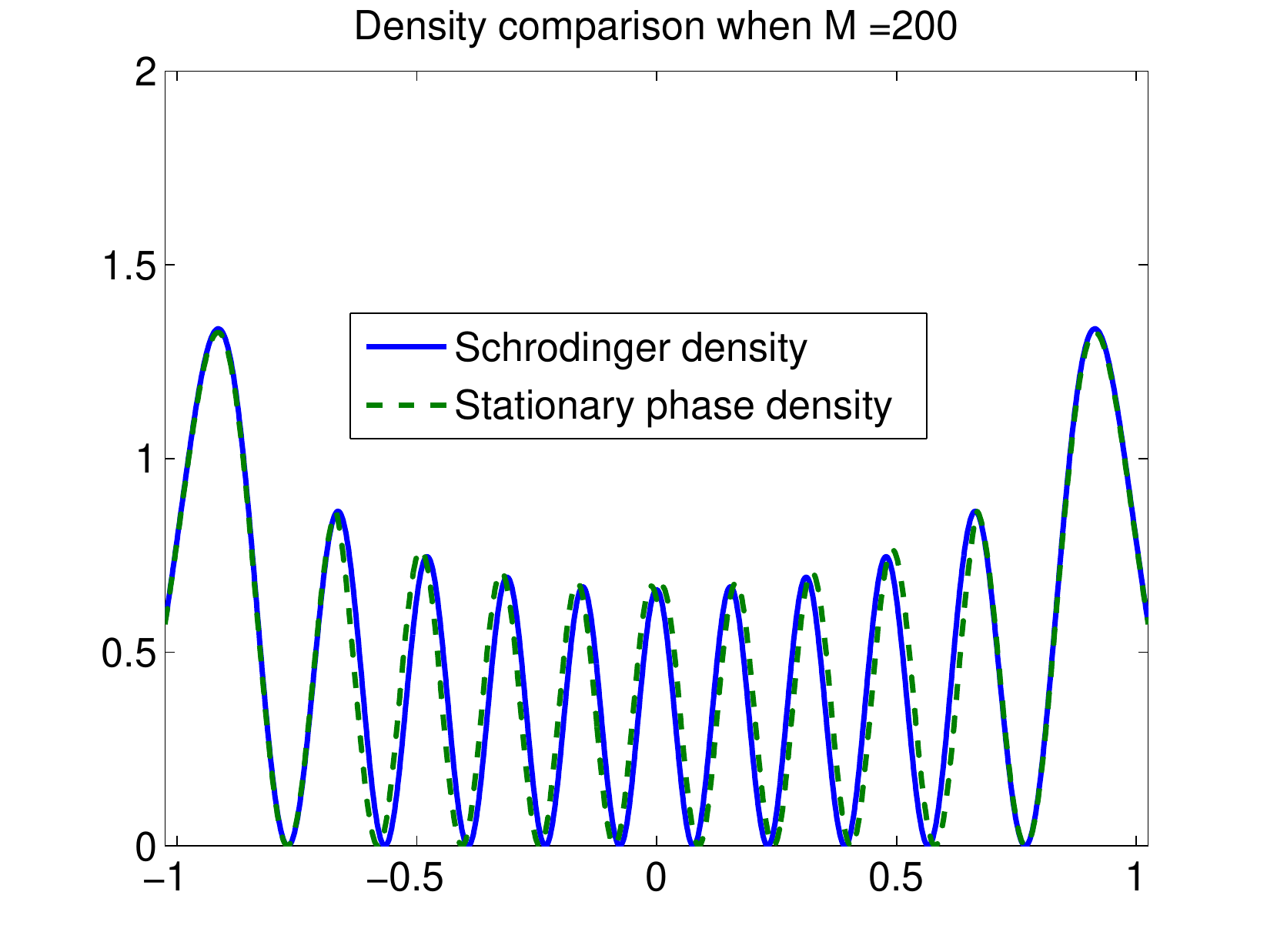}
   \includegraphics[width=.49\textwidth]{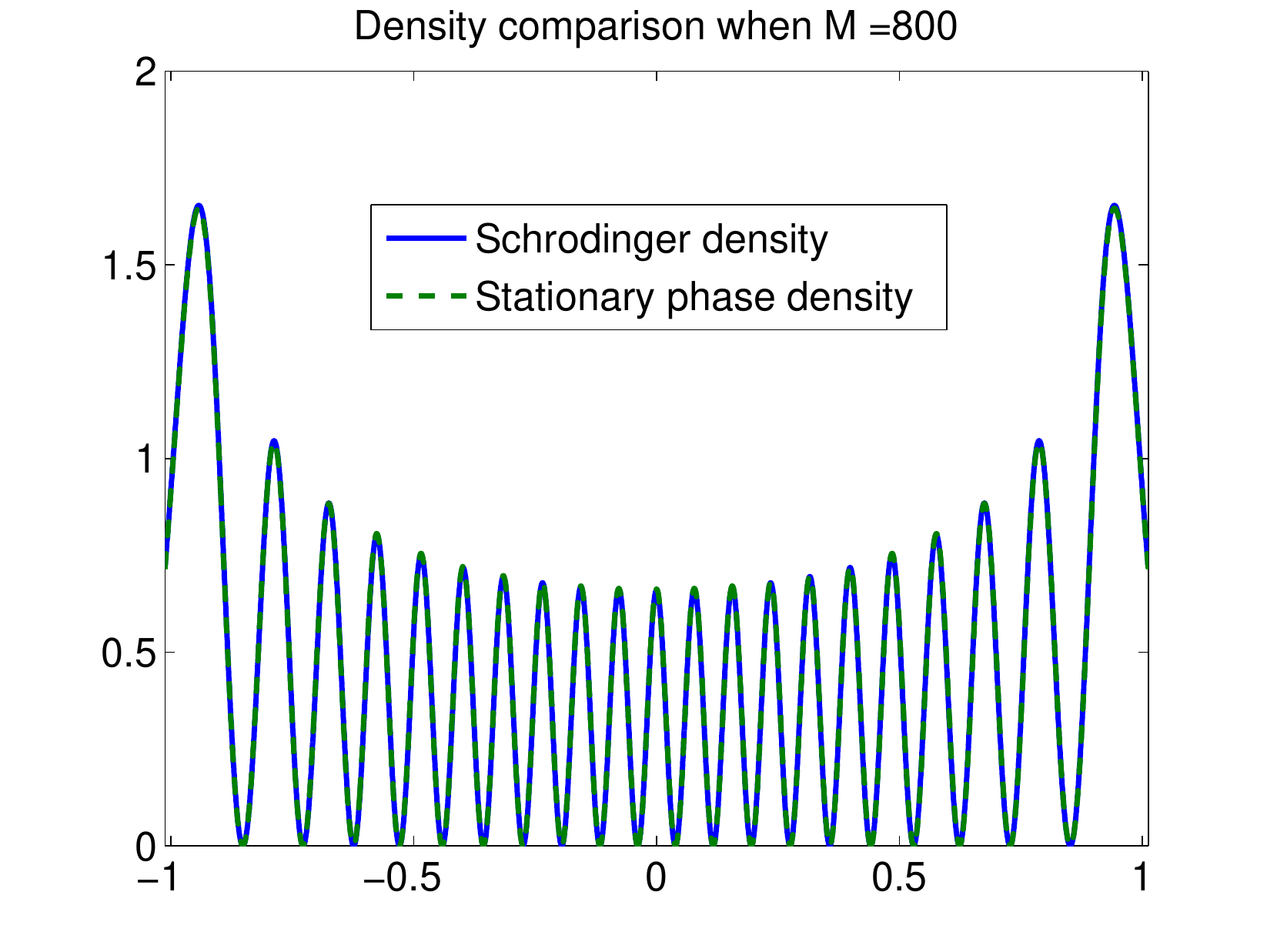}
    \caption{Comparison of the approximate solution based density $\rho$ 
      and the Schr\"odinger projection based solution $\rho_{\WIDEHATPHI}$ for 
      $M=200$ (left plot) and $M=800$ (right plot).}
    \label{fig:density_plots_stat_phase_caustics}
\end{figure}

\begin{figure}[h!]
  \centering
   \includegraphics[width=.49\textwidth]{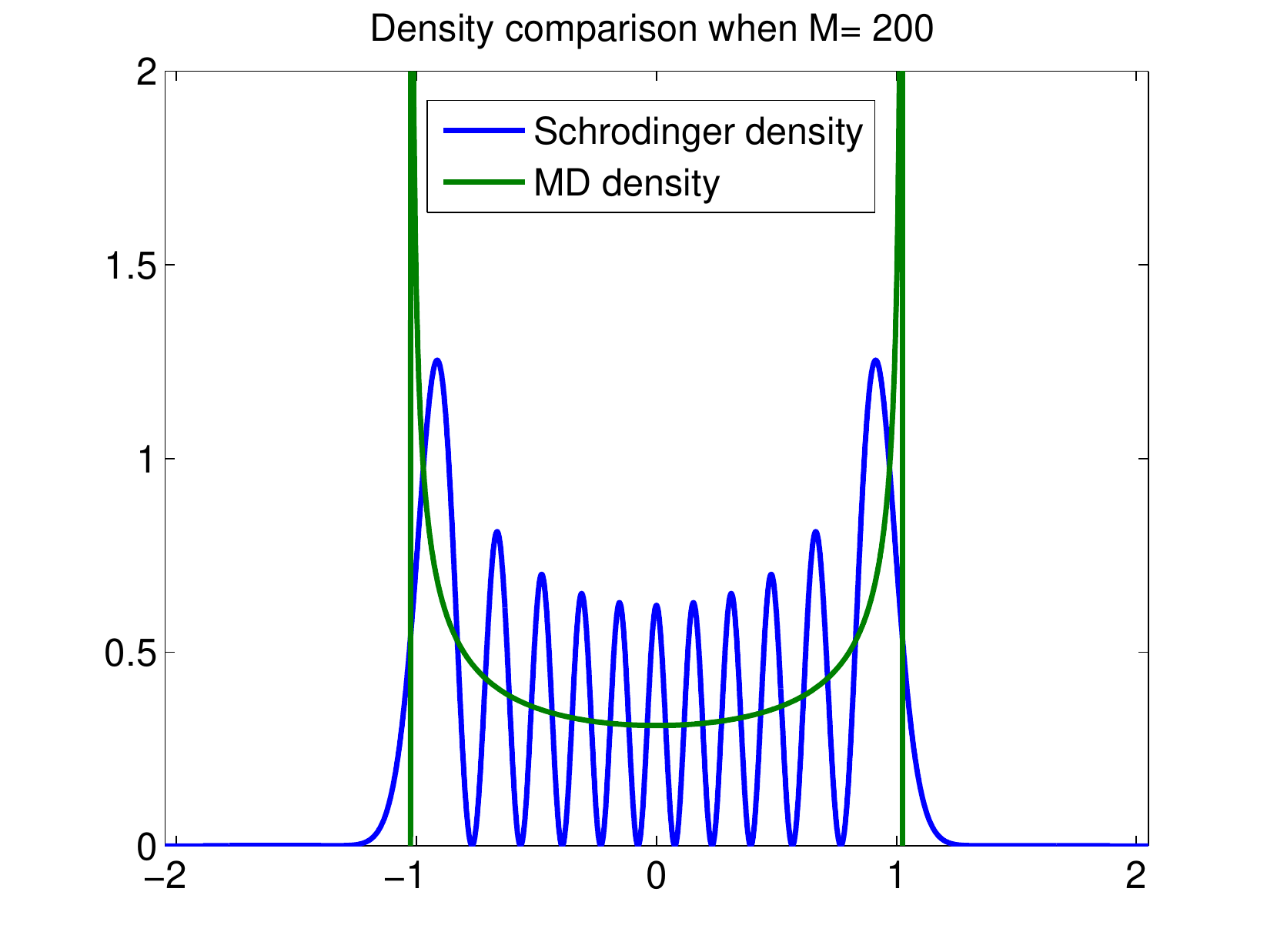}
   \includegraphics[width=.49\textwidth]{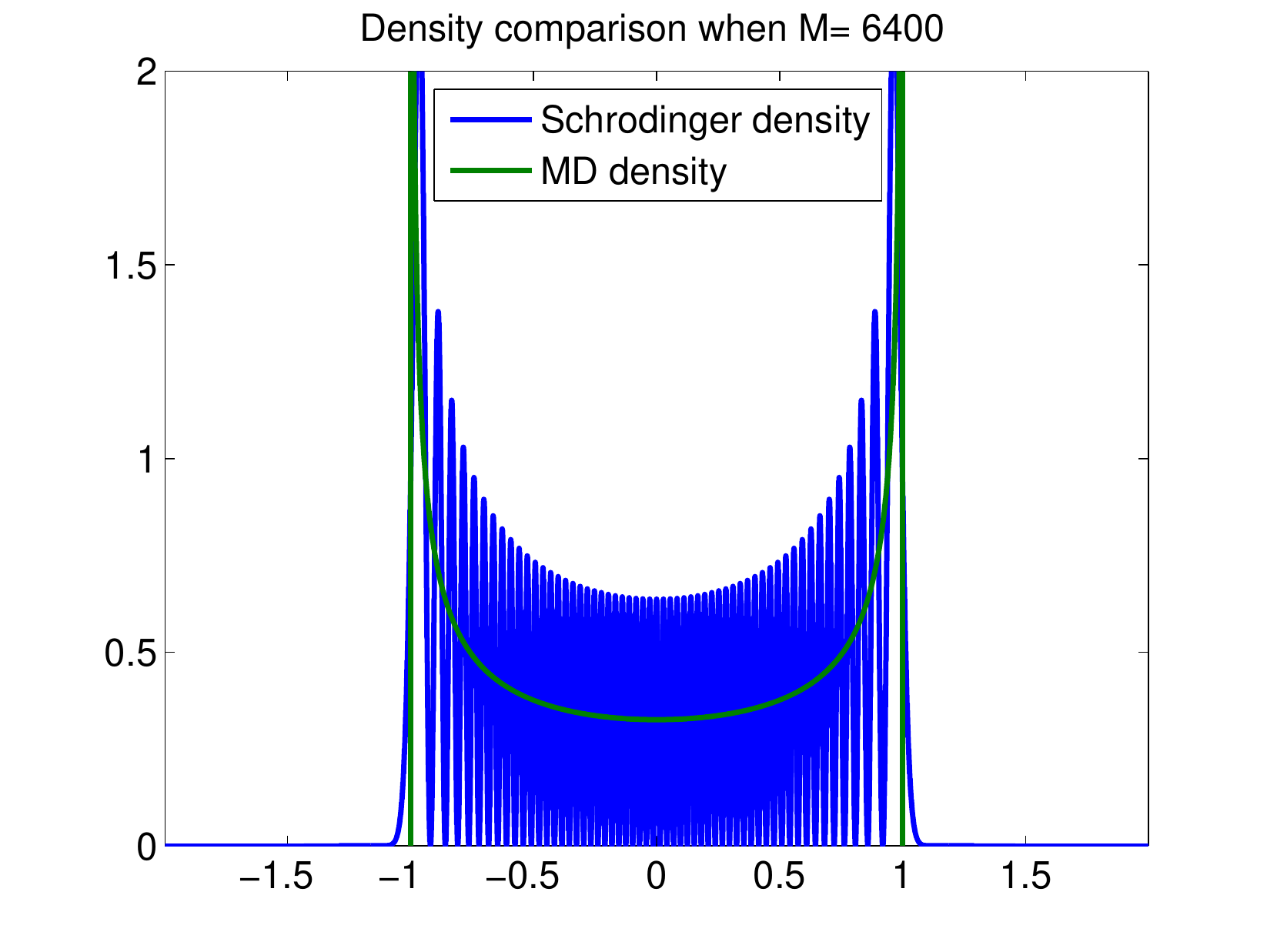}
    \caption{Comparision of the Schr\"odinger projection density $\rho_{\WIDEHATPHI}$ and 
      the molecular dynamics density $\rho_{\MD}$ for $M=200$ (left plot) and $M=6400$ (right plot).}
    \label{fig:density_plots_projected_caustics}
\end{figure}
A numerical test of the convergence rate of 
\begin{equation}\label{eq:conv_rate_densities_bis}
  \left|\frac{ \int_{-2\sqrt{E_0}}^{2\sqrt{E_0}} g_1(X) \rho_{\MD}(X)\, dX}{ \int_{-2\sqrt{E_0}}^{2\sqrt{E_0}} g_2(X) \rho_{\MD}(X)\, dX} -  
  \frac{\int_{-2\sqrt{E_0}}^{2\sqrt{E_0}} g_1(X) \rho_{\WIDEHATPHI}(X)\,dX}{\int_{-2\sqrt{E_0}}^{2\sqrt{E_0}} g_2(X) 
   \rho_{\WIDEHATPHI}(X)\, dX}
  \right|
\end{equation}
as $M$ increases is illustrated in Figure~\ref{fig:convergence_rate_caustics} 
for the observables 
\begin{equation}\label{eq:caustic_observables}
  g_1(X) = \frac{(1.5-X)^6(1.5+X)^6(1+e^{-X^2})}{1.5^{12}}\;
         \text{ and } g_2(X) = \frac{(1.5-X)^6(1.5+X)^6(1-X^2+X^4)}{1.5^{12}}\PERIOD
\end{equation}
Further details of the solution procedure in Exampe~2 are given in Algorithm~\ref{alg:pseudoCode2}.
\begin{figure}[h!]
  \centering
   \includegraphics[width=.55\textwidth]{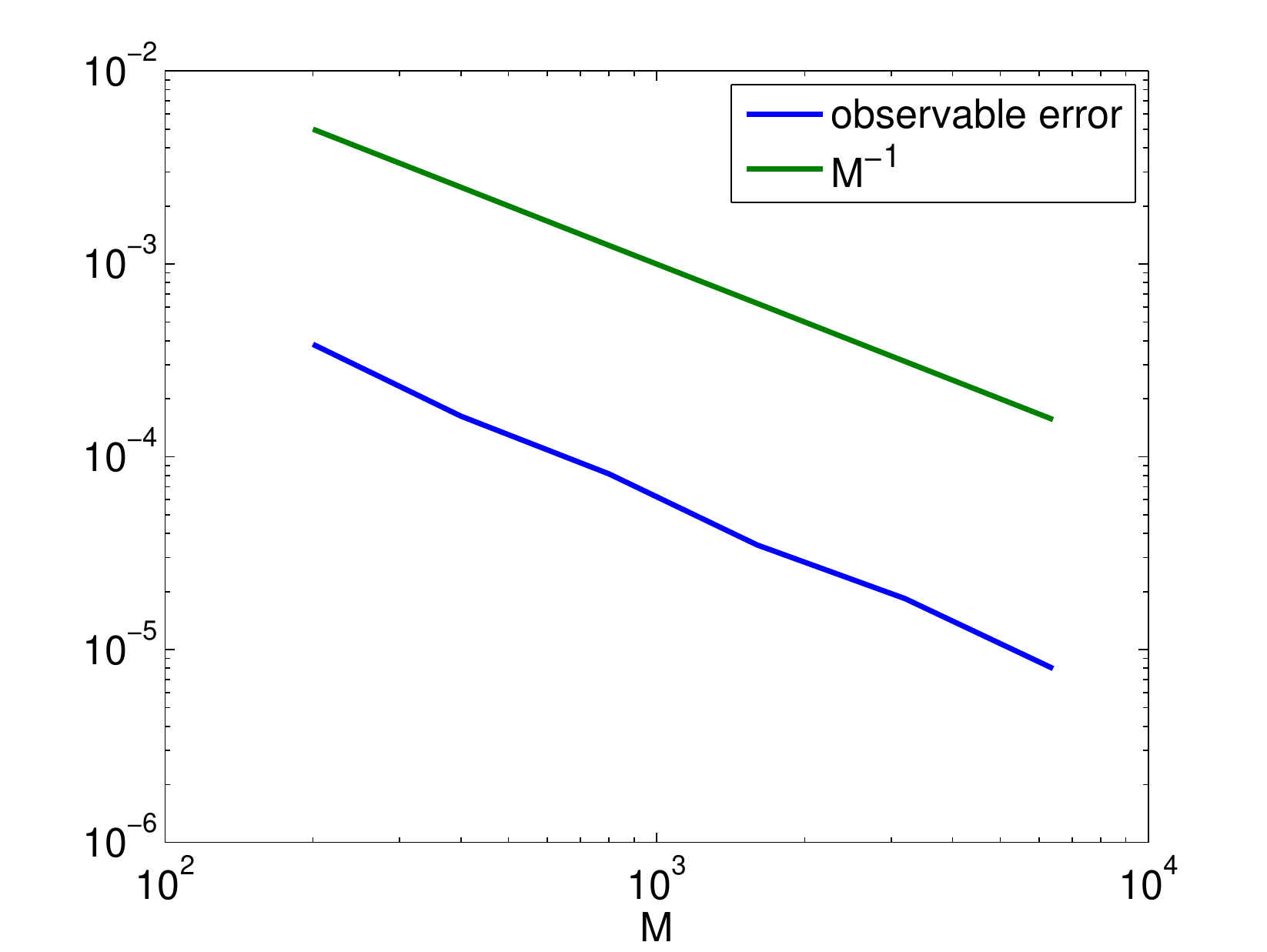}
    \caption{Convergence rate of \eqref{eq:conv_rate_densities_bis} for
      the observables $g_1$ and $g_2$ as defined in \eqref{eq:caustic_observables}.}
    \label{fig:convergence_rate_caustics}
\end{figure}

\begin{algorithm}[!h]\label{alg_2}
\caption{Algorithm for Example 2}
\begin{algorithmic}\label{alg:pseudoCode2}
\STATE {\bf Input: } An energy $E$, an one-dimensional
        potential function $V$, mass $M$, Schr\"odinger equation \eqref{eq:schrod_example_2}. 

\STATE {\bf Output: } The Schr\"odinger projection density $\rho_{\WIDEHATPHI}$.

\medskip

\STATE {\bf 1.} 
   Identify the right caustic point $X_+>0$ satisfying $X_+=V^{-1}(E)$.
   For a fixed $E \in \mathbb{R}$, consider the periodic eigenvalue problem.
   Solve \eqref{eq:schrod_example_2} numerically by constructing the discretised operator form of $-(2M)^{-1}\partial_{XX} + V$ 
   using finite differences and denoted $\HOPER^{(h)}$,
   and solve the eigenvalue problem
   \begin{equation}\label{eq:discrete_eigenvalue_problem_caustics}
      \HOPER^{(h)} P_i = E_i P_i
   \end{equation}
   for the 10 eigenvalues closest to $E$ using the Matlab
   eigenvalue solver \textbf{eigs(H,10,E)}. Let $E_0$ denote the
   eigenvalue closest to $E$ and consider from now on solving 
   \eqref{eq:schrod_example_2} for the energy $E_0$
   and its corresponding eikonal equation $\tfrac{1}{2} P^2+V(X) = E_0$.

\medskip

\STATE {\bf 2.} Determine $\LFT \theta(P)$ by 
\[
   \LFT\theta(P)=\int_0^P\GRADP\LFT\theta(p)\, dp
\]

\medskip

\STATE {\bf 3.} Evaluate the Fourier integral
\begin{equation}\label{eq:u_int}
   u(X):=\int_{-2\sqrt{E}}^{2\sqrt{E}}  \EXP{\Iunit M^{1/2}(-X P+\LFT\theta(P))}\,dP\COMMA \;\; |X|>X_0\COMMA
\end{equation}
where $X_0$ is chosen as the smallest value $X>X_+/2$ such that $|u(X)|$ is
at a local maximum, and for $|X|\leq X_0$ compute $\overline u$ by \eqref{eq:stat_phase_u}
using the stationary phase method.

\medskip

\STATE {\bf 4.}
Construct the approximate solution
$$
   \Phi(X):= \begin{cases} 
               C \overline u(X) (E_0-V(X))^{-1/4} &  |X|\leq X_0\COMMA\\
               u(X)/\sqrt{|\GRADX V(X)|} & |X| \geq X_0\COMMA
             \end{cases}
$$ 
with 
$$
  C = \frac{ u(X_0)(E_0-V(X_0))^{1/4}}{\sqrt{|\GRADX V(X_0)|}\overline u(X_0)}\PERIOD
$$
\medskip

\STATE {\bf 5.} 

Project $\Phi$ onto the eigenspace 
$\{\Upsilon_i\}_{i=1}^{\bar J}, \  \bar J\leq 10$ by 
Algorithm~\ref{alg:projection_algorithm} to obtain a projection
solution $\WIDEHATPHI$ and compute its corresponding 
approximate density
$$
  \rho_{\WIDEHATPHI} = \frac{|\WIDEHATPHI|^2(X)}{\|\WIDEHATPHI\|_2^2}\PERIOD
$$

\end{algorithmic}
\end{algorithm}

\section{The stationary phase expansion}\label{stat_phase_sec}
Consider the phase function
$\EPROD{\CHECKX}{\CHECKP} - \LFT\theta(\HATX,\CHECKP)$
and let  $\CHECKP_0(\HATX)$ be any solution to the stationary phase equation
$\CHECKX=\GRAD_{\CHECKP}\LFT \theta(\HATX,\CHECKP_0)$. We rewrite the phase function
\[
   \EPROD{\CHECKX}{\CHECKP} - \LFT\theta(\CHECKX,\CHECKP)= 
   \underbrace{\EPROD{\CHECKX}{\CHECKP_0} - \LFT\theta(\CHECKX,\CHECKP_0)}_{=\theta(\HATX,\CHECKX)} +
                         \EPROD{(\CHECKP-\CHECKP_0)}{\int_0^1(1-t)
                         \partial_{PP}\LFT\theta\left(\CHECKP_0+ t[\CHECKP-\CHECKP_0]\right)\, dt}\,
                         [\CHECKP-\CHECKP_0]\PERIOD
\]
The relation
\[
 \frac{1}{2} \EPROD{Y}{\partial_{PP}\bar\theta(\CHECKP_0) Y}= 
    \EPROD{(\CHECKP-\CHECKP_0)}{\int_0^1 (1-t)
      \partial_{PP}\bar\theta\left(\CHECKP_0 + t[\CHECKP-\CHECKP_0]\right)\, dt}\, [\CHECKP-\CHECKP_0]
 \]
 defines the function $Y(\CHECKP)$, and its inverse $\CHECKP(Y)$, so that the phase is a quadratic function in $Y$.
 The stationary phase expansion of an integral takes the form, see \REF{duistermaat},
 \begin{equation}\label{caustic_expansion}
   \begin{split}
     &\int_{\rset^d} w(\CHECKP)\, \EXP{-\Iunit M^{1/2}(\EPROD{\CHECKX}{\CHECKP} 
                     -\LFT\theta(\HATX,\CHECKP))}\, d\CHECKP\\
     &\simeq \sum_{\GRADP\LFT\theta(\CHECKP_0) =\CHECKX} (2\pi M^{-1/2})^{d/2} 
            \left| \DET \frac{\partial(\CHECKP)}{\partial(\CHECKX)}\right|^{1/2} \,
            \EXP{\Iunit\frac{\pi}{4}\SGN(\partial_{PP}\LFT\theta(\CHECKP_0))}\,
             \EXP{-\Iunit M^{1/2}\theta(\HATX,\CHECKX)}\\
     &\qquad \times \sum_{k=0}^\infty \frac{M^{-k/2}}{k!} 
             \left(\sum_{l,j} \Iunit (\partial_{P^l P^j}\LFT\theta)^{-1}(\CHECKP_0) 
                              \partial_{Y^l Y^j}\right)^k
                  \left(w(\CHECKP(Y))\,\left|\DET \frac{\partial(\CHECKP)}{\partial(Y)}\right|\right)\PERIOD
   \end{split}
 \end{equation}
 
\section*{Acknowledgment}
The research of P.P.  and A.S. was partially supported by the 
National Science Foundation under the grant
NSF-DMS-0813893 and  Swedish Research Council grant 621-2010-5647, 
respectively. P.P. also thanks KTH and Nordita for their hospitality
during his visit when the presented research was initiated.
%


\providecommand{\bysame}{\leavevmode\hbox to3em{\hrulefill}\thinspace}
\providecommand{\MR}{\relax\ifhmode\unskip\space\fi MR }
\providecommand{\MRhref}[2]{%
  \href{http://www.ams.org/mathscinet-getitem?mr=#1}{#2}
}
\providecommand{\href}[2]{#2}

\end{document}